\documentclass[10pt]{article}

\usepackage{geometry}  
\usepackage[english]{babel}
\usepackage[utf8]{inputenc}
\usepackage[T1]{fontenc}

\usepackage[dvipsnames]{xcolor}
\usepackage{paralist}
\usepackage{graphicx}
\usepackage{subcaption}
\usepackage{longtable} 
\usepackage{multirow}
\usepackage{listings}
\usepackage{makecell}
\usepackage{array}
\usepackage{float}
\usepackage{dsfont}
\usepackage{rotating}
\usepackage{booktabs}
\usepackage{enumerate}
\usepackage{tikz}
\usepackage{pgf}
\usepackage{xcolor}
\usepackage{enumitem}

\usepackage{amsmath}
\usepackage{amssymb}
\usepackage{amsthm}
\usepackage{bm}
\usepackage{bbm}
\usepackage{mathrsfs}
\usepackage{algorithm}
\usepackage{algorithmic}

\usepackage{natbib}
\usepackage{hyperref}
\hypersetup{
  colorlinks=true,
  citecolor=blue,
  linkcolor=red,
  anchorcolor=red,
  urlcolor=blue
}

\definecolor{maroon}{RGB}{128,0,0}
\definecolor{gold}{RGB}{255,215,0}
\usepackage{authblk}
\usepackage{todonotes}

\usepackage{tikz}
\usepackage{graphics}

\usetikzlibrary{shapes.geometric, arrows}
\tikzstyle{startstop} = [rectangle, rounded corners, minimum width=2.5cm, minimum height=2.5cm,text centered, draw=black, fill=blue!20]
\tikzstyle{process} = [rectangle,rounded corners, minimum width=2.5cm, minimum height=2.5cm, text centered, draw=black, fill=green!20]
\tikzstyle{arrow} = [very thick,->,>=stealth]
\tikzstyle{decision} = [rectangle, rounded corners, minimum width=2.5cm, minimum height=2.5cm, text centered, draw=black, fill=red!20]

\theoremstyle{plain}

\newtheorem{theorem}{Theorem}
\newtheorem{lemma}{Lemma}

\newtheorem{corollary}{Corollary}

\newtheorem{proposition}{Proposition}

\newcommand{\ind}{\mathbbm{1}}
\newcommand{\var}{\mathrm{Var}}

\usepackage{xspace}

\def\##1\#{\begin{align}#1\end{align}}
\def\$#1\${\begin{align*}#1\end{align*}}

\definecolor{myblue}{rgb}{.8, .8, 1}
\definecolor{mathblue}{rgb}{0.2472, 0.24, 0.6} 
\definecolor{mathred}{rgb}{0.6, 0.24, 0.442893}
\definecolor{mathyellow}{rgb}{0.6, 0.547014, 0.24}

\newcommand{\tP}{{\tilde{P}}}

\newcommand{\tX}{{\tilde{X}}}

\newcommand{\calB}{{\mathcal{B}}}

\newcommand{\calN}{{\mathcal{N}}}

\newcommand{\calP}{{\mathcal{P}}}

\newcommand{\calX}{{\mathcal{X}}}
\newcommand{\calY}{{\mathcal{Y}}}

\newcommand{\RR}{\mathbb{R}}

\def\P{{\mathbb P}}
\def\E{{\mathbb E}}

\makeatletter

\let\emptyset\varnothing
\let\hat\widehat
\let\tilde\widetilde

\def \eps {\varepsilon}

\usepackage{hyperref}
\makeatletter
\newcommand{\printfnsymbol}[1]{%
  \textsuperscript{\@fnsymbol{#1}}%
}
\makeatother

\newcommand*\samethanks[1][\value{footnote}]{\footnotemark[#1]}
\newcommand{\keywords}[1]{\textbf{\textit{Keywords:}} #1}

\usepackage{comment}
\usepackage{xcolor}
\title{Conformal prediction with local weights:\\ randomization enables robust guarantees}
\author[1]{Rohan Hore\thanks{Department of Statistics, University of Chicago} \ and Rina Foygel Barber\samethanks}
\date{\today}

\begin{document}

\maketitle

\begin{abstract}
    In this work, we consider the problem of building distribution-free prediction intervals with finite-sample conditional coverage guarantees. Conformal prediction (CP) is an increasingly popular framework for building such intervals with distribution-free guarantees, but these guarantees only ensure marginal coverage: the probability of coverage is averaged over both the training and test data, meaning that there might be substantial undercoverage within certain subpopulations. Instead, ideally we would want to have local coverage guarantees that hold for each possible value of the test point's features. While the impossibility of achieving pointwise local coverage is well established in the literature, many variants of conformal prediction algorithm show favourable local coverage properties empirically. Relaxing the definition of local coverage can allow for a theoretical understanding of this empirical phenomenon.
   
    We propose randomly-localized conformal prediction (RLCP), a method that builds on localized CP and weighted CP techniques to return prediction intervals that are not only marginally valid but also offer relaxed local coverage guarantees and validity under covariate shift. 
    Through a series of simulations and real data experiments, we validate these coverage guarantees of RLCP while comparing it with the other local conformal prediction methods.
\end{abstract}

\keywords{distribution free prediction intervals, conformal prediction, local coverage guarantees, localized conformal prediction, weighted conformal prediction, distribution shift}

\section{Introduction}\label{sec:motivation}
Consider a training dataset $(X_1,Y_1),(X_2,Y_2),\cdots,(X_n,Y_n)$ and a test point $(X_{n+1},Y_{n+1})$, where each data point $(X_i,Y_i)$ consists of a (possibly high-dimensional) feature $X_i\in \calX$ and a response $Y_i\in \calY$. 
The problem of predictive inference is as follows: if we observe training data $\{(X_i,Y_i)\}_{i\in[n]}$ (where $[n]$ denotes $\{1,\dots,n\}$), and are given a new feature vector $X_{n+1}$ for a new test point, we want to construct a prediction interval $\hat{C}_n (X_{n+1})$---a subset of $\calY$ that is likely to contain, or ``cover'', the true value of the unseen test response $Y_{n+1}$. 
The framework of \emph{distribution-free prediction} aims to provide this type of guarantee with no assumptions on the distribution $P$ of the data: that is, the construction $\hat{C}_n$ should satisfy
\begin{equation}\label{eqn:marginal_coverage}\P\{Y_{n+1}\in\hat{C}_n(X_{n+1})\}\geq 1-\alpha,\end{equation}
for some target coverage level $1-\alpha$, for data drawn i.i.d.\ from \emph{any} distribution $P$.  This type of distribution-free guarantee is achieved by 
the conformal prediction method \citep{vovk2005algorithmic} and related methodologies, which we will describe more below.

While a distribution-free guarantee of the type above is very useful in terms of avoiding relying on strong assumptions about the data distribution, it is limited in other ways. The coverage probability is averaged over the random drawn of both the training data (the data points $(X_1,Y_1),\dots,(X_n,Y_n)$ used to fit the prediction interval $\hat{C}_n$) and over the test point $(X_{n+1},Y_{n+1})$; moreover, it relies on the test point being drawn from the same distribution as the training data. For example, in a medical setting where we are trying to predict patient outcomes $Y$ based on patient features $X$ (e.g., genetic or demographic information, medical history, etc), this guarantee ensures, say, $90\%$ predictive accuracy, \emph{on average} over a new patient $(X_{n+1},Y_{n+1})$ sampled from the \emph{same} distribution as the training sample. It does not guarantee that $90\%$ coverage will hold within a particular demographic of patients (e.g., among patients over age $50$), and also does not guarantee that $90\%$ coverage will hold on average if the general population differs from the training sample (e.g., if the patients for whom we generate predictions, have a different age distribution than the population of patients that were recruited into our study).

\paragraph{Test-conditional coverage.} 
For our prediction intervals to be more useful, then, we might like to have the assurance that the prediction interval $\hat{C}_n (X_{n+1})$ would indeed contain test response $Y_{n+1}$ with fairly high probability at all features $X_{n+1}$ (i.e., among all patient demographics), rather than only on average.
In other words, we would ideally like to be able to achieve the \emph{test-conditional coverage} guarantee (sometimes also referred to as \emph{local coverage}):
\begin{equation}\label{eqn:test_conditional_coverage}\P\{Y_{n+1}\in\hat{C}_n(X_{n+1}) \mid X_{n+1}\}\geq 1-\alpha\textnormal{ for all $P$, almost surely over $X_{n+1}$.}\end{equation}
Comparing to the marginal coverage guarantee~\eqref{eqn:marginal_coverage}, this stronger guarantee would address the issues of uneven coverage over different patient demographics, and of differences between the feature distribution within training sample versus the general population.
Unfortunately, however, in the setting of a nonatomic feature distribution (i.e., the distribution of $X$ has no point masses, meaning, no value $x$ has $\P\{X=x\}>0$), this guarantee is known to be impossible \citep{vovk2012conditional,lei2014distribution}, aside from trivial solutions such as returning a meaningless prediction interval $\hat{C}_n(X_{n+1})\equiv\calY$.

Consequently, we are interested in finding a meaningful compromise between test-conditional coverage~\eqref{eqn:test_conditional_coverage}, which cannot be achieved in a distribution-free regime (aside from settings where $X$ is discrete), and marginal coverage~\eqref{eqn:marginal_coverage}, which can be achieved distribution-free but is too weak for many practical applications. In particular, we might ask for coverage to hold when we condition, not on a specific value of $X_{n+1}$, but instead in a weaker sense, e.g., for any $X_{n+1}$ within some neighbourhood---for instance, $90\%$ coverage may not be guaranteed for patients who are exactly 25.75 years old, but perhaps could be guaranteed to hold on average over any 5-year age range.
To formalize this notion mathematically, let's consider a collection $\calB$ of subsets of feature space, $B\subseteq \calX$, such that we seek valid coverage guarantees if the test feature lies in $B$, for each $B\in\calB$ (for instance, each $B$ might be a subset of patients defined by a particular 5-year age range). 
This can be expressed as the  following desired guarantee:
\begin{equation}\label{eqn:test_conditional_coverage_in_B}
    \P\{Y_{n+1}\in \hat{C}_n(X_{n+1})\mid X_{n+1}\in B\}\geq 1-\alpha ~~~\mbox{for all $P$ and for all $B\in\calB$.}
\end{equation}
Different choices of the collection of subsets $\calB$ represent different coverage goals for $\hat{C}_n(X_{n+1})$. For example, if we take $\calB=\{\calX\}$, then the coverage guarantee~\eqref{eqn:test_conditional_coverage_in_B} is simply the marginal coverage guarantee~\eqref{eqn:marginal_coverage}. 
At the other extreme, if $\calB = \{ \{x \} : x\in\calX\}$ contains all the singleton sets, then~\eqref{eqn:test_conditional_coverage_in_B} becomes equivalent to the (impossible) pointwise test-conditional coverage guarantee~\eqref{eqn:test_conditional_coverage} (and more broadly, the hardness result in \citet{barber2021limits} states that it is impossible to achieve such a guarantee whenever the Vapnik–Chervonenkis (VC)-dimension of $\calB$ is too large). Therefore, to capture a notion of coverage that lies in between the two, we are interested in collections $\calB$ that lie somewhere in between these two extremes---for instance, if $\calX=\RR^d$, we might consider all open balls of a certain radius.

\paragraph{Coverage with respect to covariate shift.} A different relaxation of pointwise test-conditional coverage~\eqref{eqn:test_conditional_coverage} frames the question as a covariate shift problem. Suppose that the test point is drawn from a distribution $\tP = \tP_X \times \tP_{Y|X}$, while training data points are drawn from $P = P_X\times P_{Y|X}$. Here $P_X$ denotes the feature distribution, i.e., the marginal distribution of $X$, while $P_{Y\mid X}$ is the conditional distribution of $Y\mid X$, and similarly for $\tP_X$ and $\tP_{Y\mid X}$. The covariate shift setting assumes that the conditional distribution of $Y|X$ is shared across the two, with $P_{Y|X} = \tP_{Y|X}$; the difference then lies in the marginal distributions of the features (or covariates) $X$, i.e., $P_X$ versus $\tP_X$. If pointwise test-conditional coverage~\eqref{eqn:test_conditional_coverage} were to hold, then this would imply that
\begin{equation}\label{eqn:covariate_shift_all}
    \P_{P^n\times \tP}\{Y_{n+1}\in \hat{C}_n(X_{n+1})\}\geq 1-\alpha ~~~\mbox{for all $P,\tP$ with $P_{Y|X}=\tP_{Y|X}$},
\end{equation}
where the probability is taken with respect to training data $(X_1,Y_1),\dots,(X_n,Y_n)$ sampled i.i.d.\ from $P$, and test point $(X_{n+1},Y_{n+1})$ sampled from $\tP$. This aim is again impossible (indeed, by taking $\tP_X$ to be a point mass at some point $x$, we see that the guarantee~\eqref{eqn:covariate_shift_all} would imply pointwise test-conditional coverage~\eqref{eqn:test_conditional_coverage}, which is impossible to attain). Thus we can again consider taking a relaxation:
\begin{equation}\label{eqn:covariate_shift_some}
    \P_{P^n\times \tP}\{Y_{n+1}\in \hat{C}_n(X_{n+1})\}\geq 1-\alpha ~~~\mbox{for all $P,\tP$ with $P_{Y|X}=\tP_{Y|X}$ and $\tP_X\in\calP$},
\end{equation}
for some class of allowed covariate shifts $\calP$---e.g., we might constrain $\tP_X$ to be near $P_X$ in some metric, or, we might constrain $\tP_X$ to be sufficiently smooth relative to $P_X$ in some sense.

\subsection{Background: conformal prediction, local coverage, and covariate shift}\label{sec:background_CP}
We will next give some background on the conformal prediction methodology, and prior work that aims to achieve approximate local guarantees for conformal prediction (some of which are similar to~\eqref{eqn:test_conditional_coverage_in_B} or to~\eqref{eqn:covariate_shift_some} in flavour).

Conformal inference (CP) \citep{vovk2005algorithmic} is an extremely 
popular framework for calibrating prediction uncertainty, and is used to build prediction intervals satisfying distribution-free marginal validity as in~\eqref{eqn:marginal_coverage}. A common variant of this methodology is the \emph{split conformal prediction} procedure \citep{papadopoulos2002inductive}, which we briefly define now. Assume we have two datasets available $\{(X_i^{\textnormal{pre}},Y_i^{\textnormal{pre}})\}_{i=1,\dots,n_{\textnormal{pre}}}$, a dataset used for pretraining, and $\{(X_i,Y_i)\}_{i=1,\dots,n}$, the holdout set which will be used for calibration.\footnote{In practice, of course, we would simply partition a single available dataset into two disjoint subsets---but it will be convenient here to use the notation of two separate datasets.}  Suppose that we fit a model on $\{(X_i^{\textnormal{pre}},Y_i^{\textnormal{pre}})\}_{i=1,\dots,n_{\textnormal{pre}}}$, and obtain a \emph{conformity score function} $s(x,y)$ that returns a large value if the point $(x,y)$ appears unusual (does not ``conform'') relative to the trends observed in the data. For example, if we run a regression procedure to train a model $\hat{f}(x)$ (which predicts $Y$ given $X=x$), we might choose the commonly used \emph{residual score} $s(x,y) = |y - \hat{f}(x)|$, which returns a large value if the value of $y$ is 
not near the predicted value $\hat{f}(x)$. The dataset $\{(X_i,Y_i)\}_{i=1,\dots,n}$, which contains $n$ many data points that were not used for training $\hat{f}$, is then used for calibrating a threshold for the score function: defining $\hat{q}$ as the $(1-\alpha)(1+1/n)$ quantile of the holdout scores $\{s(X_i,Y_i)\}_{i=1,\dots,n}$, we return the prediction interval\footnote{We use the terminology ``prediction interval'' throughout the paper as this is most familiar in statistical settings where often we have a real-valued response (i.e., $\calY=\RR$), but it might be more correct to say ``prediction set'', since even when $\calY=\RR$ the resulting set may not necessarily consist of a single interval.}
\begin{equation}\label{eqn:split_CP}\hat{C}_n(X_{n+1}) = \left\{y\in\calY : s(X_{n+1},y)\leq \hat{q}\right\}.\end{equation}
It will be useful later on to observe that this threshold can equivalently be defined as a quantile of the empirical distribution of the scores with mass added at  $+\infty$: concretely, $\hat{q} = \textnormal{Quantile}_{1-\alpha}\left(\sum_{i=1}^n \frac{1}{n+1} \delta_{s(X_i,Y_i)} + \frac{1}{n+1}\delta_{+\infty}\right)$,
where $\delta_s$ denotes the point mass at $s$. This can equivalently be defined by choosing $\hat{q}$ as the $(1-\alpha)(1+1/n)$-quantile of the empirical distribution of the holdout scores $s(X_1,Y_1),\dots,s(X_n,Y_n)$.

This method is guaranteed to satisfy distribution-free marginal coverage as in~\eqref{eqn:marginal_coverage}---but what if we would like to achieve (relaxed) test-conditional coverage and/or coverage with respect to (some classes of) covariate shift?
If our goal is to achieve test-conditional coverage relative to a \emph{single} known subset $B\subseteq\calX$, we can simply run split conformal prediction using only holdout points $i$ with $X_i\in B$ to calibrate $\hat{q}$. Similarly, if our goal is to achieve coverage under covariate for a \emph{single} known shifted distribution $\tP_X$, the Weighted Conformal Prediction (WCP) method of \citet{tibshirani2019conformal} uses a weighted quantile, with weights reflecting the shift from $P_X$ to $\tP_X$, to define $\hat{q}$ in a way that achieves the desired coverage guarantee: the threshold $\hat{q}$ in~\eqref{eqn:split_CP} is replaced with a \emph{weighted} quantile,
$\hat{q} = \textnormal{Quantile}_{1-\alpha}\left(\sum_{i=1}^n w_i \delta_{s(X_i,Y_i)} + w_{n+1}\delta_{+\infty}\right)$, where weights $w_i$ are chosen to correct for the covariate shift; for example, if $P_X$ has density $f_X$ and $\tP_X$ has density $\tilde{f}_X$, we choose $w_i \propto \tilde{f_X}(X_i)/f_X(X_i)$. 

However, in this work, we are interested in achieving coverage conditional on $X_{n+1}\in B$ for \emph{any} $B$ (in some broad class of allowed sets $\calB$, e.g., all balls of a certain radius), or, coverage with respect to covariate shift for \emph{any} $\tP_X$ within some allowed class of shifted distributions. 
For these more challenging problems, a range of methods have been proposed in the literature, including work by \citet{lei2014distribution,cauchois2024robust,izbicki2022cd,guan2023localized,gibbs2023conformal}. Some of these methods offer theoretical results aiming towards relaxed guarantees similar to~\eqref{eqn:test_conditional_coverage_in_B} (for test-conditional coverage) or to~\eqref{eqn:covariate_shift_some} (for covariate shift). We defer a more detailed discussion of this broad field to Section~\ref{sec:discussion_literature}, where we will be able to compare more closely to our own results.

\subsection{Our contributions}
As described above, many methods within the conformal prediction framework have attempted to provide approximate guarantees of test-conditional coverage and coverage with respect to covariate shift, but combining sharp empirical performance with interpretable theoretical guarantees has proved to be challenging. In our work, we aim to unify these questions, and bridge the gap between theoretical validity and empirical performance for these problems. Our aim is to provide a $\hat{C}_n$ that simultaneously satisfies all of the following aims: 

\begin{enumerate}[label=(\textbf{A\arabic*})]
    \item \label{A1} $\hat{C}_n$ provably achieves exact finite-sample distribution-free marginal coverage, i.e., $\P\{Y_{n+1}\in \hat{C}_n(X_{n+1})\}\geq 1-\alpha$ for data drawn i.i.d.\ from any distribution $P$. 
    \item \label{A2} $\hat{C}_n$ provably achieves a meaningful relaxed test-conditional coverage guarantee, with $\P\{Y_{n+1}\in \hat{C}_n(X_{n+1})\mid X_{n+1}\in B\}\gtrapprox 1-\alpha$ holding over all sets $B$ satisfying some regularity conditions that we will detail later on.

    \item \label{A3} $\hat{C}_n$ provably achieves approximate coverage guarantees for covariate shift problems, with $\P_{P^n\times \tP}\{Y_{n+1}\in \hat{C}_n(X_{n+1})\}\gtrapprox 1-\alpha$ holding over all $P$ and all covariate-shifted distributions $\tP=\tP_X\times P_{Y|X}$, where the shift from $P_X$ to $\tP_X$ satisfies some regularity conditions that we will detail later on.
\end{enumerate}
Of course, all of these aims could again be achieved by simply returning the uninformative interval $\hat{C}_n(X_{n+1})\equiv \calY$---so our goal is to provide a $\hat{C}_n$ that satisfies aims \ref{A1}--\ref{A3} without being too conservative, i.e., producing an interval that is as narrow as possible. We will see that randomly localized conformal prediction (RLCP) provides a tool that can allow us to achieve these
aims without being overly conservative.

\paragraph{Organization of the paper.}
The remainder of the paper is organized as follows. We will begin by reviewing the literature on framework for locally weighted conformal prediction methods in Section~\ref{sec:review}, giving details on some existing methods and their shortcomings in order to gain motivation for our proposed method RLCP. Then in Section~\ref{sec:method}, we formally introduce RLCP, and in Section~\ref{sec:theory}, we present theoretical results towards establishing that RLCP satisfies the desired properties \ref{A1}, \ref{A2}, \ref{A3}. In Section~\ref{sec:experiment}, we demonstrate the performance of RLCP (and compare it with these other locally weighted CP methods described in Section~\ref{sec:review}) on both simulated and real data. Finally, we conclude with a discussion in Section~\ref{sec:discussion}, where we compare our method with existing literature and examine the role that randomization plays in our method. Some proofs and some details of implementing the methods are deferred to the Appendix.

\section{Conformal prediction with local weights}
\label{sec:review}
In this section, we will review the generalized framework of locally weighted conformal prediction from \citet{guan2023localized} where two locally adaptive conformal methods have been introduced. Our terminology here will differ somewhat from the terminology used in that reference, to help us introduce the new perspectives and frameworks in this paper later on.

As for the split conformal method defined in~\eqref{eqn:split_CP}, we begin by pretraining a score function $s=s(x,y)$, where large values indicate that $(x,y)$ does not agree or ``conform'' with the trends observed in the pretraining set $\{(X^{\textnormal{pre}}_i,Y^{\textnormal{pre}}_i)\}_{i=1,\dots,n_{\textnormal{pre}}}$. Unlike the split conformal method, however, under a localized framework we will compute a local threshold for the score function, to produce a prediction interval of the form
\begin{equation}\label{eqn:LCP_generic_Chatn}\hat{C}_n(X_{n+1}) = \left\{y\in\calY: s(X_{n+1},y)\leq \hat{q}(X_{n+1})\right\},\end{equation}
where $\hat{q}(x)$ may now vary with $x$, to reflect local differences in the distribution of the score. For example, for the regression score $s(x,y) = |y - \hat{f}(x)|$, if the fitted model $\hat{f}$ appears to be a poor fit at (or near) some particular $x$, we would want $\hat{q}(x)$ to be larger at (or near) this value $x$ to ensure local predictive coverage.

The remaining question, then, is how to compute this local threshold $\hat{q}(x)$. \citet{guan2023localized}'s localized conformal prediction (LCP) framework addresses this question, as we will now describe.

\subsection{An initial approach: baseLCP}\label{sec:baseLCP}
The threshold $\hat{q}$ constructed by the split conformal method may not lead to local coverage at (or near) some particular $x$, because the distribution of the scores $s(X_i,Y_i)$ over the holdout set $i=1,\dots,n$ may be quite different from the distribution of the score $s(X_{n+1},Y_{n+1})$ if we condition on $X_{n+1}=x$ (or $X_{n+1}\approx x$). To address this, a natural approach is to compute a quantile $\hat{q}(x)$ that places more weight on holdout points $X_i$ that lie near $x$. This method is discussed by \citet{guan2023localized}, and formalized as follows. We first choose some \emph{localization kernel} given by a function
\[H: \calX\times\calX\rightarrow\RR_{\geq 0},\]
for instance, for $\calX=\RR^d$, common choices include a Gaussian kernel (with bandwidth $h$) given by
\[H(x,x') = \exp\left\{ - \|x-x'\|^2_2/2h^2\right\},\]
or a box kernel (again with bandwidth $h$), given by
\[H(x,x') = \ind\{\|x-x'\|_2\leq h\}.\]
The local quantile is then computed as
\begin{equation}\label{eqn:hatqx_baseLCP} \hat{q}_{1-\alpha}(X_{n+1}) = \textnormal{Quantile}_{1-\alpha}\left(\sum_{i=1}^n w_i \delta_{s(X_i,Y_i)} + w_{n+1} \delta_{+\infty}\right),\end{equation}
where $\delta_s$ is the point mass at $s$, and where the weights are given by
\[w_i = \frac{H(X_i,X_{n+1})}{\sum_{j=1}^{n+1} H(X_j,X_{n+1})}, \ i\in[n+1].\]
(Note that $\hat{q}_{1-\alpha}(X_{n+1})$ depends on the training data $\{(X_i,Y_i)\}_{i\in[n]}$ as well as the test point $X_{n+1}$, but this dependence is not made explicit in the notation.)

Finally, the resulting prediction interval is defined as in~\eqref{eqn:LCP_generic_Chatn}, i.e.,
\[\hat{C}_n^{\textnormal{baseLCP}}(X_{n+1}) = \left\{y\in\calY : s(X_{n+1},y)\leq \hat{q}_{1-\alpha}(X_{n+1})\right\}.\]
Observe that, if we choose the trivial localization kernel $H(x,x')\equiv 1$ (i.e., no localization), then we would have $\hat{q}_{1-\alpha}(x)$ constant over $x$, and we recover the original split conformal method~\eqref{eqn:split_CP}. 

 For the Gaussian kernel and box kernel described above (as well as other natural choice of $H$), the choice of bandwidth $h$ controls the extent of localization. At a high level, for larger values of $h$ the method behaves similarly to split conformal, while for smaller $h$ the method becomes more localized, and can offer better empirical local coverage. However, if we choose $h$ too small, performance will suffer due to low effective sample size---the quantile $\hat{q}(x)$ may even be estimated as $+\infty$. We will return later on to the question of choosing a bandwidth $h$ that appropriately addresses this trade-off.

This approach is very intuitive: to (approximately) achieve local coverage at $x$, we simply give more weight to holdout points for which $X_i\approx x$. Interestingly, however, this method may fail in practice: even marginal coverage may fail; empirically, the marginal coverage may be either higher or lower than $1-\alpha$, and \citet[Propositions 2.4, 2.5]{guan2023localized} show theoretically that marginal coverage may be arbitrarily high or arbitrarily low for this method. Consequently, we will refer to this method as the ``base'' version of localized conformal prediction, or ``baseLCP'' for short, to indicate that it forms a baseline that other methods will build upon. To correct for this issue in the marginal coverage level, \citet{guan2023localized} proposes an adjusted version of the method to address the problems of baseLCP, which we will present next.

\subsection{Calibrating to recover marginal coverage: calLCP}\label{sec:defining calLCP}

Now we describe \citet{guan2023localized}'s proposed method---this method is referred to as localized conformal prediction in that paper but we will refer to it as ``calibrated'' localized conformal prediction, or calLCP for short, to distinguish it from baseLCP and from our own method to be defined later on. The idea is to recover marginal coverage by changing the definition of the local threshold $\hat{q}(x)$. 

Recall that baseLCP uses the threshold $\hat{q}_{1-\alpha}(X_{n+1})$ for the score $s(X_{n+1},y)$, to determine whether this value $y$ belongs in the prediction interval---but as we have discussed above, this might lead to actual coverage being below or above the target level $1-\alpha$. The idea of calLCP is to replace the level $1-\alpha$ in the quantile with a different value, to recalibrate to the target coverage level.

Specifically, \citet{guan2023localized} defines the prediction interval as \[\hat{C}_n^{\textnormal{calLCP}}(X_{n+1}) = \left\{y\in\calY : s(X_{n+1},y)\leq \hat{q}_{1-\tilde{\alpha}(X_{n+1},y)}(X_{n+1})\right\},\]
for
\[
\hat{q}_{1-\tilde{\alpha}(X_{n+1},y)}(X_{n+1}) = \textnormal{Quantile}_{1-\tilde{\alpha}(X_{n+1},y)}\left(\sum_{i=1}^n w_i~ \delta_{s(X_i,Y_i)} + w_{n+1} ~\delta_{s(X_{n+1},y)}\right),
\]
where the adjusted quantile value $1-\tilde{\alpha}(X_{n+1},y)$ is chosen ensure marginal coverage, as follows. To define this value, we introduce notation 
\[Y^y_i = \begin{cases} Y_i, & i\in [n], \\ y, & i=n+1,\end{cases}\]
i.e., $(Y^y_1,\dots,Y^y_{n+1})$ is equal to the vector containing the $n$ training data response values $Y_1,\dots,Y_n$, augmented with the value $y$ in the last argument.
Then $\tilde{\alpha}(X_{n+1},y)$ is defined as
\begin{equation}\label{eqn:definition of true alpha_tilde}
\tilde\alpha(X_{n+1},y) = \max\left\{a \in \Gamma(w): \sum_{i=1}^{n+1} \ind\left\{s(X_i,Y^y_i)\leq \textnormal{Quantile}_{1-a}\left( \sum_{j=1}^{n+1} w_{i,j} \delta_{s(X_i,Y^y_i)}  \right) \right\}\geq (1-\alpha)(n+1)\right\},
\end{equation}
where, overloading notation, we let 
\[w_{i,j} = \frac{H(X_j,X_i)}{\sum_{k=1}^{n+1} H(X_k,X_i)}, \ i,j\in[n+1]\]
(and note that $w_{n+1,i}$ coincides with the weight $w_i$ defined above),
and where
\begin{equation}\label{eqn:def_Gamma_w_calLCP}\Gamma(w) = \left\{\sum_{j\in I} w_{i,j} : i\in[n+1], ~I\subseteq[n+1]\right\}\end{equation}
is the set of possible sums of weights that might occur in the calculation of the discrete Cumulative Distribution Fucntion (CDF) inside the quantile.

Examining this definition, it appears that implementing calLCP requires computing a recalibrated threshold $\tilde{\alpha}(X_{n+1},y)$ for every value $y\in\calY$ (i.e., a similar computational challenge to that faced by full conformal prediction \citep{vovk2005algorithmic}). However, one can bypasses this issue of computational complexity via an efficient implementation proposed in \citet[Section 3.2]{guan2023localized}.

With this recalibration step, 
\citet[Theorem 2.1]{guan2023localized} establishes that marginal coverage at level $1-\alpha$ holds for the calLCP method (unlike for baseLCP).
Empirically, the method also shows improvement in local coverage, but there are no known distribution-free theoretical results to guarantee this (however, under additional smoothness assumptions, a conditional coverage result is established in \citet[Theorem 5.2]{guan2023localized}). Interestingly, the marginal coverage result can be recovered via a new insight: we will show in Section~\ref{sec:calLCP_fullCP}
that \citet{guan2023localized}'s calLCP method is, in fact, an instance of full conformal prediction---but implemented with a modified score function, not with the original pretrained score function $s=s(x,y)$.

\subsection{Another look at baseLCP: achieving a different coverage guarantee}\label{sec:new look at baseLCP}
In this section, we will talk about an alternative formulation of baseLCP in order to get a better understanding of the properties of this method, and why marginal coverage can fail. (This in turn will help motivate the construction of RLCP, which we will present next.)

Consider a new feature variable $\tX_{n+1}$, generated as $\tX_{n+1}|X_{n+1}\sim P_X\circ H(\cdot,X_{n+1})$, where for any function $g: \calX\rightarrow\RR_{\geq 0}$ with $0<\E_{P_X}[g(X)]<\infty$, we define $P_X\circ g$ as the distribution $P_X$ reweighted by $g$, i.e.,
\begin{equation}\label{eqn:define_PX_circ_g}(P_X\circ g)(A) = \frac{\int_A g(x)\;\mathsf{d}P_X(x)}{\int_{\calX}g(x)\;\mathsf{d}P_X(x)}\textnormal{ for all }A\subseteq\calX.\end{equation}
If the localization kernel $H$ has a reasonably small bandwidth, then the reweighted distribution $P_X\circ H(\cdot,X_{n+1})$ places most of its mass near $X_{n+1}$. As a result,
we can interpret $\tX_{n+1}$ as a sort of ``synthetic sample'' designed to be similar to the test feature $X_{n+1}$---in a setting where each data point is a patient, we are generating new patient data $\tX_{n+1}$ for a patient that is \emph{similar} (in feature space) to the test patient for whom we want to perform prediction.

Now we return to the baseLCP method. How does constructing the ``synthetic sample'' $\tX_{n+1}$ relate to this procedure?
It turns out that baseLCP is actually related to running WCP \citep{tibshirani2019conformal} on the covariate shift problem defined by taking $\tP_X = P_X\circ H(\cdot,X_{n+1})$. Specifically, if we define
\[\hat{C}_n^{\textnormal{WCP}}(X_{n+1},\tX_{n+1}) = \left\{y\in\calY : s(\tX_{n+1},y) \leq \hat{q}_{1-\alpha}(X_{n+1})\right\},\]
then this prediction interval will have probability $\geq 1-\alpha$ of covering a new response value $\tilde{Y}_{n+1}$, drawn from the conditional distribution $P_{Y|X}$ at $X=\tX_{n+1}$ (as established by \citet[Theorem 5.3]{guan2023localized}).
An equivalent way of thinking about this is that the threshold $\hat{q}_{1-\alpha}(X_{n+1})$ offers a $1-\alpha$ probability of bounding the score for a synthetic test point, $s(\tX_{n+1},\tilde{Y}_{n+1})$---but not necessarily the original score $s(X_{n+1},Y_{n+1})$ itself for the actual test point of interest.

This interpretation leads us towards the RLCP method. In baseLCP, using the kernel $H(\cdot,X_{n+1})$ centred at the true test point $X_{n+1}$, leads to a threshold $\hat{q}_{1-\alpha}(X_{n+1})$ that provides coverage for a synthetic sample $\tX_{n+1}$---that is, for data points that are \emph{similar} to $X_{n+1}$, but not necessarily for $X_{n+1}$ itself. For RLCP, then, we will instead use a kernel that is \emph{not} centred at $X_{n+1}$, but will aim to achieve coverage for $X_{n+1}$ itself.

\section{Randomly localized conformal prediction}
\label{sec:method}
From our discussion above, we saw that baseLCP constructs the threshold $\hat{q}_{1-\alpha}(X_{n+1})$ with weights defined by a kernel centred at $X_{n+1}$, which then leads to coverage for a ``synthetic sample'' $\tX_{n+1}$ that is sampled \emph{near} $X_{n+1}$. To define RLCP, we will reverse this construction: we will define the threshold using weights centred at a random $\tX_{n+1}$ sampled near $X_{n+1}$, and will see later on that this leads to marginal coverage (and also approximate local coverage) at $X_{n+1}$.

We now define the method. We again begin with a function $H:\calX\times\calX\rightarrow \RR_{\geq 0}$, with an added assumption that, for each $x$, $H(x,\cdot)$ defines a density with respect to some base measure $\nu$ on $\calX$:
\begin{equation}\label{eqn:condition_on_H}
\int_{\calX} H(x,x')\;\mathsf{d}\nu(x') = 1\textnormal{ for all $x\in\calX$}.
\end{equation}
Returning to the earlier examples of the Gaussian kernel and the box kernel in the setting $\calX = \RR^d$, here we can use the same kernels up to a normalizing constant (and take $\nu$ to be Lebesgue measure): for the Gaussian kernel with bandwidth $h$, we take
\begin{equation}\label{eqn:H_gaussian_normalized}H(x,x')=\frac{1}{(2\pi h^2)^{d/2}}\exp\left(-\frac{\|x-x'\|^2_2}{2h^2}\right),\end{equation}
while for the box kernel with bandwidth $h$, we take
\begin{equation}\label{eqn:H_box_normalized}H(x,x')=\frac{1}{V_d h^d}\ind\{\|x-x'\|_2\leq h\},\end{equation}
where $V_d$ is the volume of the unit ball in $\RR^d$. Although both of these choices for $H$ happen to be symmetric in the two arguments, we do not assume this condition.

With this function selected, we now define the RLCP method: given the pretrained score $s=s(x,y)$, the holdout data $(X_1,Y_1),\dots,(X_n,Y_n)$, and the test point $X_{n+1}$,
\begin{multline}\label{eqn:def_RLCP}
\textnormal{Sample $\tX_{n+1}$ from density $H(X_{n+1},\cdot)$,}\\\textnormal{ then return $\hat{C}_n^{\textnormal{RLCP}}(X_{n+1},\tX_{n+1}) = \{y\in\calY : s(X_{n+1},y)\leq \hat{q}_{1-\alpha}(X_{n+1},\tX_{n+1})\}$,}
\end{multline}
where analogously to before, the threshold $\hat{q}_{1-\alpha}(X_{n+1},\tX_{n+1})$ is computed as
\[\hat{q}_{1-\alpha}(X_{n+1},\tX_{n+1}) = \textnormal{Quantile}_{1-\alpha}\left( \sum_{i=1}^n \tilde{w}_i \delta_{s(X_i,Y_i)} + \tilde{w}_{n+1} \delta_{+\infty}\right),\]
for weights
\[\tilde{w}_i = \frac{H(X_i,\tX_{n+1})}{\sum_{j=1}^{n+1}H(X_j,\tX_{n+1})}, \quad i\in[n+1].\]
(As before, $\hat{q}(X_{n+1},\tX_{n+1})$ depends also on the training data $\{(X_i,Y_i)\}_{i\in[n]}$, but this dependence is not explicit in the notation.)

The only difference between RLCP and baseLCP, then, is that the weights are given by the $\tilde{w}_i$'s, computed with the kernel $H$ centred at $\tX_{n+1}$, rather than the original $w_i$'s, which were computed with $H$ centred at $X_{n+1}$.
\subsection{The role of $\tX_{n+1}$}\label{sec:role_of_Xt}
Before proceeding, we pause to discuss $\tX_{n+1}$ a bit more. First, observe that the randomness of $\tX_{n+1}$ means that the resulting prediction interval is no longer a deterministic function of the data---it is randomized because the threshold $\hat{q}_{1-\alpha}(X_{n+1},\tX_{n+1})$ is computed with the kernel $H$ centred at a random point $\tX_{n+1}$ (hence the term ``randomly localized'' in the name of the method). 
To compare again with baseLCP,
\begin{itemize}
    \item For baseLCP, the weights used for computing the threshold $\hat{q}_{1-\alpha}(X_{n+1})$ are calculated by centring the kernel $H$ at $X_{n+1}$, and can lead to coverage for a ``synthetic sample'' $\tX_{n+1}$, drawn from $P_X\circ H(\cdot,X_{n+1})$. In other words, we can think of the original test point $X_{n+1}$ as a prototype: coverage will hold relative to the reweighted feature distribution $P_X\circ H(\cdot,X_{n+1})$, which generates new ``synthetic samples'' $\tX_{n+1}$ which are potentially very similar to this prototype---but coverage may fail for $X_{n+1}$ itself, i.e., marginal coverage may fail to hold, as observed by \citet{guan2023localized}.
    \item  For RLCP, the weights used for computing the threshold $\hat{q}_{1-\alpha}(X_{n+1},\tX_{n+1})$ are calculated by centring the kernel $H$ at a random point $\tX_{n+1}$ that was drawn near $X_{n+1}$, and (as we will see below) will lead to coverage for $X_{n+1}$ itself. We can now think of $\tX_{n+1}$ as a ``synthetic prototype'', and conditionally, the test sample $X_{n+1}$ is likely to lie near to this prototype $\tX_{n+1}$. More concretely, we will see that the conditional distribution of $X_{n+1}$ given $\tX_{n+1}$ is equal to the reweighted feature distribution $P_X\circ H(\cdot,\tX_{n+1})$, which (for a strongly localizing kernel $H$) places most of its mass near the ``synthetic prototype'' $\tX_{n+1}$.
\end{itemize}

We observe that, for RLCP, the ``synthetic prototype'' $\tX_{n+1}$ may \emph{not} be a plausible data value---for example, the feature distribution $P_X$ might be supported on integer values, but we might still choose to use a Gaussian kernel for $H$, leading to noninteger values of $\tX_{n+1}$ when it is sampled from density $H(X_{n+1},\cdot)$. In contrast, for baseLCP, $\tX_{n+1}$ is drawn from $P_X\circ H(\cdot,X_{n+1})$, whose support is equal to (or is a subset of) the support of $P_X$, and so we think of $\tX_{n+1}$ as a ``synthetic sample''---a new data point (e.g., a new patient) from the same population, that was drawn in such a way that it is similar to $X_{n+1}$.

\section{Theoretical guarantees}
\label{sec:theory}
In this section, we present theoretical coverage guarantees for RLCP. We will see that these guarantees verify that RLCP achieves the aims outlined earlier in \ref{A1}--\ref{A3}. 

We will begin in Section~\ref{sec:theory_marginal} by establishing a key property of the RCLP method in Proposition~\ref{prop:key_property} that underlies in all its coverage guarantees.
With this result in place, we then present the marginal coverage guarantee (aim \ref{A1}). Approximate test-conditional coverage (aim \ref{A2}) and robustness to covariate shift (aim \ref{A3}) will be addressed in Sections~\ref{sec:theory_test_conditional} and~\ref{sec:theory_covariate_shift}, respectively, where we will also explore the relationship between these two problems. RLCP also enjoys a training-conditional coverage guarantee, but since this result is of a different flavour, we defer it to Appendix~\ref{app:training_conditional}.

\subsection{Marginal coverage and a key property}\label{sec:theory_marginal}
We begin by defining some notation. In the field of distribution-free prediction, test-conditional coverage is often quantified by studying the conditional miscoverage rate
\[\alpha(X_{n+1}) = \P\left\{Y_{n+1}\not\in \hat{C}_n(X_{n+1}) \ \middle| \ X_{n+1}\right\}, \]
where $\hat{C}_n$ is the prediction interval produced by the procedure of interest. In other words, $\alpha(X_{n+1})$ is the miscoverage rate conditional on $X_{n+1}$, but marginalizing over the random draw of the training data (and the draw of $Y_{n+1}$). To relate this random variable to our earlier terminology and notation, we can observe that the marginal coverage guarantee~\eqref{eqn:marginal_coverage} is then equivalent to requiring $\E[\alpha(X_{n+1})]\leq \alpha$, while the pointwise test-conditional coverage guarantee~\eqref{eqn:test_conditional_coverage} would require $\alpha(X_{n+1})\leq \alpha$ almost surely.

Now we extend this idea to  RLCP. For our method, after marginalizing over the randomness in the training data, we have an additional remaining random variable $\tX_{n+1}$ in addition to the test feature $X_{n+1}$. We then define
\[\alpha(X_{n+1},\tX_{n+1}) = \P\left\{Y_{n+1}\not\in \hat{C}^{\textnormal{RLCP}}_n(X_{n+1},\tX_{n+1}) \ \middle| \  X_{n+1},\tX_{n+1}\right\}.\]

With this notation in place, the following result establishes a key property of the RLCP method.
For this proposition, and all subsequent results, we assume that the RLCP method is implemented with a localization kernel $H$ that satisfies the condition~\eqref{eqn:condition_on_H}, i.e., $H(x,\cdot)$ is a density (with respect to some fixed base measure) for each $x$.

\begin{proposition}[Key property of RLCP]\label{prop:key_property}
Let $(X_1,Y_1),\dots,(X_{n+1},Y_{n+1})\stackrel{\textnormal{iid}}{\sim} P$ for any distribution $P$. Then the RLCP method defined in~\eqref{eqn:def_RLCP} satisfies
\[\E\left[\alpha(X_{n+1},\tX_{n+1}) \ \middle| \  \tX_{n+1}\right] \leq \alpha\textnormal{ almost surely}.\]
\end{proposition}

This new quantity $\alpha(X_{n+1},\tX_{n+1})$ can again be related to marginal coverage by taking an expected value, and thus Proposition~\ref{prop:key_property} immediately yields a marginal coverage guarantee.

\begin{theorem}[Marginal coverage guarantee]\label{thm:marginal_coverage}
The RLCP method defined in~\eqref{eqn:def_RLCP} satisfies the marginal coverage property~\eqref{eqn:marginal_coverage}, i.e., 
for any distribution $P$, if  $(X_1,Y_1),\dots,(X_{n+1},Y_{n+1})\stackrel{\textnormal{iid}}{\sim} P$, then
\[\P\left\{Y_{n+1}\in\hat{C}^{\textnormal{RLCP}}_n(X_{n+1},\tX_{n+1})\right\}\geq 1-\alpha.\]
\end{theorem}
\begin{proof}[Proof of Theorem~\ref{thm:marginal_coverage}]
    Using the tower law, we calculate
    \begin{multline*}
        \P\left\{Y_{n+1}\not \in\hat{C}^{\textnormal{RLCP}}_n(X_{n+1},\tX_{n+1})\right\}
        = \E\left[\P\left\{Y_{n+1}\not\in\hat{C}^{\textnormal{RLCP}}_n(X_{n+1},\tX_{n+1}) \ \middle| \  X_{n+1}, \tX_{n+1}\right\}\right]\\
        = \E[\alpha(X_{n+1},\tX_{n+1})]
        = \E[\E[\alpha(X_{n+1},\tX_{n+1})\mid \tX_{n+1}]]
        \leq \E[\alpha]= \alpha,
    \end{multline*}
    where the inequality applies Proposition~\ref{prop:key_property}.
\end{proof}

In other words, this theorem follows immediately from Proposition~\ref{prop:key_property} by observing that the new notation offers an equivalent formulation of marginal coverage~\eqref{eqn:marginal_coverage},
\[\P\left\{Y_{n+1}\in \hat{C}_n^{\textnormal{RLCP}}(X_{n+1},\tX_{n+1})\right\}\geq 1-\alpha\iff \E\left[\alpha(X_{n+1},\tX_{n+1})\right]\leq \alpha.\]

We remark that we can similarly obtain an equivalent formulation of pointwise test-conditional coverage~\eqref{eqn:test_conditional_coverage} since the test-conditional coverage $\P\left\{Y_{n+1}\in \hat{C}_n^{\textnormal{RLCP}}(X_{n+1},\tX_{n+1}) \ \middle| \  X_{n+1} \right\}$ can be equivalently expressed as $\E\left[\alpha(X_{n+1},\tX_{n+1}) \ \middle| \  X_{n+1} \right]$. However, this last equivalence does not lend itself to an immediate application of Proposition~\ref{prop:key_property}, since we do not have a bound on $\E\left[\alpha(X_{n+1},\tX_{n+1}) \ \middle| \  X_{n+1} \right]$ (because the proposition conditions on $\tX_{n+1}$ rather than on $X_{n+1}$). Instead, we will need to take a different approach for the test-conditional coverage problem, as we will see in the next section. Before proceeding, however, we first prove the proposition, to examine how the construction of RLCP (and specifically, the manner in which $\tX_{n+1}$ is drawn) allows for this key property to hold.

\begin{proof}[Proof of Proposition~\ref{prop:key_property}]
First, we compute the distribution of the data conditional on $\tX_{n+1}$. Since the data points are i.i.d., and so the training data is independent of $\tX_{n+1}$ by construction, the training data has the same distribution as before, $(X_1,Y_1),\dots,(X_n,Y_n)\stackrel{\textnormal{iid}}{\sim}P$. For the test point, however, the distribution has changed: by construction, the joint distribution of $(X_{n+1},Y_{n+1},\tX_{n+1})$ is given by
\[\begin{cases} X_{n+1}\sim P_X, \\
Y_{n+1}\mid X_{n+1}\sim P_{Y|X}(\cdot |X_{n+1}),\\
\tX_{n+1}\mid (X_{n+1},Y_{n+1}) \sim H(X_{n+1},\cdot).\end{cases}\]
A direct calculation then yields
\begin{equation}\label{eqn:RLCP_is_covariate_shift}(X_{n+1},Y_{n+1})\mid \tX_{n+1}\sim \big(P_X \circ H(\cdot,\tX_{n+1})\big) \times P_{Y|X}.\end{equation}
We can see, then, that this is an instance of a covariate shift problem: the distribution of $X_{n+1}\mid \tX_{n+1}$ is no longer equal to $P_X$, but the conditional distribution of $Y_{n+1}$ remains the same as for the training data, i.e., $P_{Y|X}$.

Now we return to the RLCP method.
The quantile $\hat{q}_{1-\alpha}(X_{n+1},\tX_{n+1})$ is calculated using a reweighted empirical distribution of the scores. Examining the definition of these weights $\tilde{w}_i$ (defined in Section~\ref{sec:method}), and recalling the weighted conformal prediction method introduced in Section~\ref{sec:background_CP}, we can see that, conditional on $\tX_{n+1}$, the RLCP procedure is equivalent to running weighted conformal prediction (WCP) for this covariate shift problem---i.e., when the shifted covariate distribution is given by $\tP_X = P_X \circ H(\cdot,\tX_{n+1})$.
The desired coverage result then follows as a direct application of \citet{tibshirani2019conformal}'s result for weighted conformal prediction, which guarantees  coverage at level $1-\alpha$ when WCP is applied to a setting with a (known) covariate shift.
\end{proof}

\subsection{Test-conditional coverage}\label{sec:theory_test_conditional}
Next, we shift our attention to relaxed test-conditional coverage guarantees of the form outlined in Aim \ref{A2}. Our goal is to establish that, for a broad class of ``nice'' sets $B\subseteq\calX$, the RLCP prediction interval offers conditional coverage at level $\gtrapprox 1-\alpha$.
To state this result, we first need to introduce some more notation. First, we write $P_{(X,\tX)}$ to denote the distribution on $(X,\tX)\in \calX\times\calX$ defined by sampling
\[  X\sim P_X, \ 
    \tX\mid X \sim H(X,\cdot).
\]
Next, for a set $B\subseteq\calX$ that we will condition on, for any $r>0$ define
\[\textnormal{bd}_r(B) = \left\{x\in B : \inf_{x'\in B^c} \|x-x'\|  \leq  r\right\},\]
i.e., any $x\in B$ that lies within distance $r$ of the complement, $B^c$.\footnote{Implicitly, we are assuming that $\calX$ is a metric space, and that any $B\subseteq\calX$ we condition on is a measurable set. The norm $\|\cdot \|$ may be any norm on $\calX$.} We can think of $\textnormal{bd}_r(B)$ as an inflation of the boundary of the set $B$.

\begin{theorem}[Test-conditional coverage guarantee]\label{thm:test_conditional}
The RLCP method defined in~\eqref{eqn:def_RLCP} satisfies the following property: for any distribution $P$, if  $(X_1,Y_1),\dots,(X_{n+1},Y_{n+1})\stackrel{\textnormal{iid}}{\sim} P$, then 
for every $B\subseteq\calX$, it holds that
    $$
    \P\left\{Y_{n+1}\in \hat{C}_n^{\textnormal{RLCP}}(X_{n+1},\tX_{n+1})  \ \middle| \ X_{n+1}\in B\right\}\geq1- \alpha-\frac{\inf_{\epsilon>0}\left\{P_{(X,\tX)}\{\|X-\tX\|>\epsilon\}+P_X(\textnormal{bd}_{2\epsilon}(B))\right\}}{P_X(B)}.
    $$
    \end{theorem}
Let us pause to unpack this result and examine the terms appearing in the bound. We would like to determine for what sets $B$, and for what choice of the localization kernel $H$, the excess error bound $\frac{P_{(X,\tX)}\{\|X-\tX\|>\epsilon\}+P_X(\textnormal{bd}_{2\epsilon}(B))}{P_X(B)}$ is likely to be small (for some $\epsilon$).
First, in the denominator, the term $P_X(B)$ is the mass placed by the feature distribution $P_X$ on the set $B$, and hence can be thought of as an proxy for effective sample size. A low value of the denominator means that our dataset likely contains very few data points $(X_i,Y_i)$ with $X_i\in B$, and hence ensuring coverage conditional on this set is not feasible. 

Now we turn to the numerator. For the first term, $P_{(X,\tX)}\{\|X-\tX\|> \eps\}$, this term is small as long as $\eps$ is chosen to be sufficiently large relative to the bandwidth of the localization kernel $H$---for example, when $H$ is a box kernel,  this term is zero if we choose $\eps=h$. The second term in the numerator essentially measures regularity properties of the set $B$. For instance, if $P_X$ has a bounded density, then $P_X(\textnormal{bd}_{2\epsilon}(B))$ scales with the volume of the set $\textnormal{bd}_{2\epsilon}(B)$, which is a $2\eps$-inflated boundary of the set $B$ and will have relatively low volume for many classes of ``nice'' sets $B$---as we will see next for a specific example, namely, by taking $B$ to be a ball in the $\ell_2$ norm.

\subsubsection{Example: coverage conditional on a ball}\label{sec:coverage_on_balls}
To make the results of Theorem~\ref{thm:test_conditional} more concrete, let us see what these results imply for a specific setting: assuming $\calX=\RR^d$, we will take $B = \mathbb{B}(x_0,r) = \{x\in\RR^d : \|x-x_0\|_2\leq r\}$, the ball of radius $r$ centred at some point $x_0$. We will also take $H$ to be the normalized box kernel with bandwidth $h$, $H(x,x') = \frac{1}{V_d h^d}\ind\{\|x-x'\|_2\leq h\}$, with $2h\leq r$.
Applying Theorem~\ref{thm:test_conditional} to this setting (with $\eps = h$) yields the guarantee
\[\P\left\{Y_{n+1} \in \hat{C}_n^{\textnormal{RLCP}}(X_{n+1},\tX_{n+1})  \ \middle| \  X_{n+1}\in \mathbb{B}(x_0;r)\right\}\geq 1-\alpha-\frac{P_X(\mathbb{B}(x_0;r)\setminus \mathbb{B}(x_0;r-2h))}{P_X(\mathbb{B}(x_0;r))},\]
since we have $\textnormal{bd}_{2\epsilon}(B))=\mathbb{B}(x_0;r)\setminus \mathbb{B}(x_0;r-2h)$ for our choice of set $B$.

In particular, if $P_X$ is a continuous distribution, then the ratio appearing in the coverage guarantee will be small---if we assume $P_X$ has density $f_X$ with respect to Lebesgue measure, for instance, then
\begin{multline*}\frac{P_X(\mathbb{B}(x_0;r)\setminus \mathbb{B}(x_0;r-2h))}{P_X(\mathbb{B}(x_0;r))} \leq \frac{\sup_{x\in \mathbb{B}(x_0;r)}f_X(x) \cdot \textnormal{Leb}\big( \mathbb{B}(x_0;r) \backslash \mathbb{B}(x_0;r-2h)\big)}{\inf_{x\in \mathbb{B}(x_0;r)}f_X(x)\cdot \textnormal{Leb}\big( \mathbb{B}(x_0;r)\big)} \\ = \frac{\sup_{x\in \mathbb{B}(x_0;r)}f_X(x)}{\inf_{x\in \mathbb{B}(x_0;r)}f_X(x)} \cdot \left[1 - \left(1 - \frac{2h}{r}\right)^{\!d}\,\right].\end{multline*}

\paragraph{An asymptotic result.}
While the goal of our work is to provide finite-sample theoretical results, for intuition it is useful to consider how this guarantee behaves asymptotically. Suppose that the distribution $P$ is fixed (in particular, dimension $d$ is taken to be constant), and $P_X$ has density $f_X$ as above, but we take sample size $n\to\infty$. At each $n$, we will run RLCP with a box kernel of bandwidth $h_n$, and will examine coverage with respect to a ball of radius $r_n\geq 2h_n$. Rewriting the above calculation into this new asymptotic notation, we have
\[\P\left\{Y_{n+1} \in \hat{C}_n^{\textnormal{RLCP}}(X_{n+1},\tX_{n+1})  \ \middle| \  X_{n+1}\in \mathbb{B}(x_0;r_n)\right\}\geq 1-\alpha-\underbrace{\frac{\sup_{x\in \mathbb{B}(x_0;r_n)}f_X(x)}{\inf_{x\in \mathbb{B}(x_0;r_n)}f_X(x)}}_{\textnormal{Term 1}} \cdot \underbrace{\left[1 - \left(1 - \frac{2h_n}{r_n}\right)^{\!d}\,\right]}_{\textnormal{Term 2}}.\]
Term 1 has to do with the regularity of the distribution $P_X$. In particular, if $f_X(x)$ is continuous and positive at $x=x_0$, then $r_n\to 0$ implies that Term 1 converges to $1$ as $n\to\infty$. Term 2 has to do with the set $B$ that we are conditioning on---we need our localization (controlled by $h_n$) to be sufficiently strong for ensuring coverage conditional on $X_{n+1}\in B$ (and this event's probability is controlled by $r_n$). In particular, if we choose $h_n,r_n$ such that $h_n/r_n\to 0$, then Term 2 is vanishing as $n\to \infty$. To summarize, we have seen that in this setting, the following asymptotic result applies:
\begin{corollary}\label{cor:coverage_on_balls} Under the notation and assumptions defined above, if $r_n\to 0$ and $h_n/r_n\to 0$, then
\[\P\left\{Y_{n+1} \in \hat{C}_n^{\textnormal{RLCP}}(X_{n+1},\tX_{n+1})  \ \middle| \  X_{n+1}\in \mathbb{B}(x_0;r_n)\right\} \geq 1-\alpha - \mathrm{o}(1).\]
\end{corollary}
We remark that under some additional smoothness assumptions, we can also derive asymptotic test-conditional coverage of the form \eqref{eqn:test_conditional_coverage} under this setting---that is, pointwise coverage, which conditions on $X_{n+1}=x_0$, rather than $X_{n+1}\approx x_0$ as in the result above. We defer the theorem and its proof to the Appendix \ref{app:asymptotic_test_conditional}.

\subsubsection{The role of the bandwidth}\label{sec:role_of_bandwidth}

The coverage guarantee in Theorem~\ref{thm:test_conditional} is valid irrespective of how we choose the kernel $H$---in particular, we are free to choose any bandwidth $h$. However, as mentioned earlier in Section~\ref{sec:baseLCP}, we need to choose the bandwidth carefully in order for the guarantee, and the prediction interval itself, to be meaningful. If $h$ is too large, then the method is not very localized---and in particular, the numerator in the excess error term appearing in Theorem~\ref{thm:test_conditional} might then be arbitrarily large (since we would need to choose a very large value of $\eps$ in order for $P_{(X,\tX)}\{\|X-\tX\|>\eps\}$ to be small).
If instead we choose $h$ to be too small, on the other hand, the excess error term in Theorem~\ref{thm:test_conditional} will generally be very small, but this comes at the cost of an overly wide (perhaps even infinitely wide) prediction interval. 

To make this more concrete, suppose $H$ is the box kernel of bandwidth $h$. When constructing the RLCP prediction set, we need a large effective sample size---that is, a large number of training points $i\in[n]$ such that $\|X_i - \tX_{n+1}\|\leq h$. At the extreme, if \emph{none} of the training points $i\in[n]$ satisfy this bound, then the quantile $\hat{q}_{1-\alpha}(X_{n+1},\tX_{n+1})$ will be $+\infty$.

Now consider an asymptotic regime (as discussed above, for coverage conditional on a ball)---suppose that at sample size $n$, we take $H$ to be a box kernel with bandwidth $h_n$, where our features lie in $\calX = \RR^d$. Under regularity conditions on $P_X$, the mass of the ball $\mathbb{B}(\tX_{n+1},h_n)$ will typically scale as $n(h_n)^d$---i.e., this is the effective sample size after localization. Hence, in an asymptotic setting, to ensure that RLCP produces meaningful intervals we should choose the bandwidth $h_n$ such that $h_n\to 0$ but $h_n \cdot n^{1/d}\to\infty$.

\subsection{Coverage against arbitrary covariate shift}\label{sec:theory_covariate_shift}
As have been stated before, test-conditional coverage (or approximations of test-conditional coverage, as in Theorem~\ref{thm:test_conditional}) can also be viewed as a special case of coverage against covariate shift settings. To consider a general setting, recall that, for any reweighting function $g:\calX\rightarrow\RR_{\geq 0}$, with $\E_{P_X}[g(X)]$ positive and finite, the reweighted covariate distribution $P_X\circ g$ is defined as in~\eqref{eqn:define_PX_circ_g}. We will consider a setting where training data is drawn from $P=P_X\times P_{Y|X}$, while the test point distribution (i.e. the target distribution for which we would like to have predictive coverage) is $(P_X\circ g)\times P_{Y|X}$. For example, in a setting where each data point represents a patient and the feature space $\calX$ includes information about age, $g$ might be a monotone increasing function of age, indicating that older patients comprise a \emph{larger} fraction of the general population as compared to the training sample; equivalently, older patients are \emph{underrepresented} in the training sample.

Our goal, as described earlier in Aim~\ref{A3}, is to provide a predictive interval that offers an approximate coverage guarantee with respect to a large class of ``nice'' reweighting functions $g$. Leveraging our new notation, we would like to prove a bound of the form
\[\E\left[\alpha(X_{n+1},\tX_{n+1})\right]\lessapprox \alpha,\]
or equivalently, 
\[\P\{Y_{n+1}\in \hat{C}^{\textnormal{RLCP}}_n(X_{n+1},\tX_{n+1})\}\gtrapprox 1-\alpha,\]
where expectation and probability are taken with respect to the following data distribution: the training points $(X_i,Y_i)\stackrel{\textnormal{iid}}{\sim} P$, the test point $(X_{n+1},Y_{n+1})\sim (P_X\circ g)\times P_{Y|X}$, and then by construction of the RLCP method, also $\tX_{n+1}\mid X_{n+1}\sim H(X_{n+1},\cdot)$.

In this section, we will work towards a general result of this type for a broad class of reweighting functions $g: \calX\rightarrow\RR_{\geq 0}$. We will then see that the test-conditional coverage result in Theorem~\ref{thm:test_conditional} can be derived as a special case. 

To present the result, we need to define a parameter that captures the variability of the reweighting function $g$: for any $r>0$ and any $A\subseteq\calX$, define
\[L_{g,r,A} = \sup_{\substack{x,x'\in A\\ \|x-x'\|\leq r} }\frac{|g(x)-g(x')|}{r}.\]
We can essentially interpret $L_{g,r,A}$ as a relaxed version of the Lipschitz constant of $g$ on the set $A$. In particular, if $g$ is $L_g$-Lipschitz, then $L_{g,r,A}\leq L_g$ for all $r$ and all $A$. We will also write $\|g\|_\infty = \sup_{x\in\calX} |g(x)|=\sup_{x\in\calX} g(x).$

We are now ready to state our main result.
\begin{theorem}[Robustness to covariate shift]\label{thm:covariate_shift}
The RLCP method defined in~\eqref{eqn:def_RLCP} satisfies the following property: for any distribution $P$, 
and every $g:\calX\rightarrow\RR_{\geq 0}$ with $P_X\circ g$ well defined (i.e., with $0<\E_{P_X}[g(X)]<\infty$), it holds that
    \begin{multline*}
    \P\left\{Y_{n+1} \in \hat{C}_n^{\textnormal{RLCP}}(X_{n+1},\tX_{n+1}) \right\}\geq 1-\alpha-{}\\ \frac{\inf_{A\subseteq\calX, \epsilon>0}\left\{\epsilon \cdot L_{g,2\epsilon,A}  + \|g\|_\infty\cdot P_{(X,\tX)}\{\|X-\tX\|>\epsilon\}+ \|g\|_\infty \cdot P_X(A^c) \right\}}{\E_{P_X}[g(X)]},
    \end{multline*}
where the probability on the left-hand side is taken with respect to $(X_1,Y_1),\dots,(X_n,Y_n)\stackrel{\textnormal{iid}}{\sim} P$ and $(X_{n+1},Y_{n+1})\sim (P_X\circ g)\times P_{Y|X}$, and $\tX_{n+1}\mid X_{n+1}\sim H(X_{n+1},\cdot)$ as defined by the RLCP method~\eqref{eqn:def_RLCP}.
\end{theorem}

In general, computing the infimum in this upper bound, over all $A\subseteq \calX$ and $\epsilon>0$, might be highly nontrivial. However, in some special cases, we can choose $A$ strategically to get an interpretable upper bound on the miscoverage. These special cases include:
\begin{itemize}
    \item \textbf{Lipschitz functions:} if $g$ is $L_g$-Lipschitz, then by choosing $A=\calX$, the result of the theorem yields
\begin{equation}\label{eqn:bound_for_Lipschitz_g}
   \P\left\{Y_{n+1}\in \hat{C}_n^{\textnormal{RLCP}}(X_{n+1},\tX_{n+1}) \right\}\geq 1-\alpha-  \frac{\inf_{ \epsilon>0}\left\{\epsilon \cdot L_g  + \|g\|_\infty \cdot P_{(X,\tX)}\{\|X-\tX\|>\epsilon\}\right\}}{\E_{P_X}[g(X)]}.
\end{equation}
\item \textbf{Indicator functions:} suppose $g(x) = \ind\{x\in B\}$ is the indicator function for some set $B\subseteq\calX$. In this case, we can see that requiring coverage with respect to covariate shift (with $P_X\circ g$ as the shifted distribution) is exactly equivalent to requiring coverage conditional on $X_{n+1}\in B$. Consequently, by considering indicator functions, this general result leads to an immediate proof of Theorem~\ref{thm:test_conditional}. \begin{proof}[Proof of Theorem~\ref{thm:test_conditional}]
For any $B\subseteq\calX$, define $g(x) = \ind\{x\in B\}$. Define the set $A = (\textnormal{bd}_{2\epsilon}(B))^c\subseteq\calX$. Observe that for any $x,x'\in A$, if $\|x-x'\|\leq2\eps$, we cannot have $x\in B$ and $x'\in B^c$ (because this would imply $x\in \textnormal{bd}_{2\epsilon}(B)$, which is a contradiction); we must either have $x,x'\in B$ or $x,x'\in B^c$, and in particular this means $g(x) = g(x')$. Thus, we have $L_{g,2\epsilon,A} = 0$, and so Theorem~\ref{thm:covariate_shift} yields
\begin{multline*}\P\left\{Y_{n+1}\in \hat{C}_n^{\textnormal{RLCP}}(X_{n+1},\tX_{n+1})  \ \middle| \  X_{n+1}\in B\right\}\geq 1-\alpha\\- \frac{\inf_{\epsilon>0}\left\{\|g\|_\infty \cdot P_{(X,\tX)}\{\|X-\tX\|>\epsilon\}+ \|g\|_\infty \cdot P_X(\textnormal{bd}_{2\epsilon}(B)) \right\}}{\E_{P_X}[g(X)]}.\end{multline*}
Finally, by definition of $g$, we calculate $\|g\|_\infty=1$ and $\E_{P_X}[g(X)] = P_X(B)$, which yields the desired bound.
\end{proof}

\end{itemize}

More generally, we can interpret the various terms in Theorem~\ref{thm:covariate_shift}'s excess error bound as follows. To make the first two terms in the numerator small, we need to ensure that the localization kernel $H$ is indeed ``local'': there needs to be some small value $\epsilon>0$ for which $\P\{\|X-\tX\|>\epsilon\}$ is low. For the first and third term in the numerator to both be small, we need to choose a function $g$ that, while perhaps not globally Lipschitz, is nonetheless (approximately) Lipschitz over the bulk of the distribution of features $X$. Finally, the denominator simply serves to make the result scale-invariant: if we consider a reweighting function $g'$ defined as $g'(x) = c\cdot g(x)$ for some positive constant $c$, the two shifted distributions $P_X\circ g$ and $P_X\circ g'$ are equal, so our upper bound needs to be invariant to this type of rescaling in order to be meaningful.

\subsubsection{A closer look at the covariate shift guarantee: proof sketch of Theorem~\ref{thm:covariate_shift}}\label{sec:insights_to_results}
To develop our intuition for the construction of the RLCP method, and how the construction of $\tX$ plays a role, here we will take a closer look at some of the key ideas underlying the proof of Theorem~\ref{thm:covariate_shift}.

In the covariate shift setting, by definition, the joint distribution of the data along with the random point $\tX_{n+1}$ is given by
\[\begin{cases} 
\{(X_i,Y_i)\}_{i\in[n]}\stackrel{\textnormal{iid}}{\sim}P,\\
(X_{n+1},Y_{n+1})\sim (P_X\circ g)\times P_{Y|X},\\
\tX_{n+1}\mid X_{n+1}\sim H(X_{n+1},\cdot).\end{cases}\]
Equivalently, we can write this joint distribution as
\begin{equation}\label{eqn:explain_marginal_coverage_wrt_now_with_g}\begin{cases} 
\{(X_i,Y_i)\}_{i\in[n]}\stackrel{\textnormal{iid}}{\sim}P,\\
\tX_{n+1}\sim \tP_{\tX},\\
(X_{n+1},Y_{n+1})\mid \tX_{n+1}\sim (P_X\circ g\circ H(\cdot, \tX_{n+1}))\times P_{Y|X},\end{cases}\end{equation}
where  $\tP_{\tX}$ is the marginal distribution of $\tX_{n+1}$ under the covariate shift model: this distribution can be characterized by the marginal distribution of $\tX$ when we first draw $X\sim P_X\circ g$, and then draw $\tX\sim H(X,\cdot)$. In order to establish approximate coverage under covariate shift, then, we need to show that 
\[\P\left\{Y_{n+1}\in\hat{C}^{\textnormal{RLCP}}_n(X_{n+1},\tX_{n+1})\right\}\gtrapprox 1-\alpha\textnormal{ with respect to the joint distribution~\eqref{eqn:explain_marginal_coverage_wrt_now_with_g}}.\]

Now we will consider an approximation to this joint distribution of the data, which no longer depends on the unknown reweighting function $g$. Specifically, compare to the following joint distribution:
\begin{equation}\label{eqn:explain_marginal_coverage_wrt}\begin{cases} 
\{(X_i,Y_i)\}_{i\in[n]}\stackrel{\textnormal{iid}}{\sim}P,\\
\tX_{n+1}\sim \tP_{\tX},\\
(X_{n+1},Y_{n+1})\mid \tX_{n+1}\sim (P_X\circ H(\cdot, \tX_{n+1}))\times P_{Y|X}.\end{cases}\end{equation}
We observe the only difference between this distribution, and the joint distribution that appears in~\eqref{eqn:explain_marginal_coverage_wrt_now_with_g}, lies in the conditional distribution of $(X_{n+1},Y_{n+1})\mid \tX_{n+1}$.
Will marginal coverage hold with respect to this new distribution? In fact, in this new joint distribution~\eqref{eqn:explain_marginal_coverage_wrt}, the conditional distribution of the training data and test point $\{(X_i,Y_i)\}_{i\in[n+1]}$ conditional on $\tX_{n+1}$, is exactly the same as that calculated  in the proof of Proposition~\ref{prop:key_property} for the i.i.d.\ data setting (without covariate shift). If we examine the result of Proposition~\ref{prop:key_property}, we see that it establishes 
that RLCP has $\geq1-\alpha$  coverage conditional on $\tX_{n+1}$, 
\[\P\left\{Y_{n+1}\in\hat{C}^{\textnormal{RLCP}}_n(X_{n+1},\tX_{n+1}) \ \middle| \  \tX_{n+1}\right\}\geq 1-\alpha,\]
with respect to drawing $\{(X_i,Y_i)\}_{i\in[n]}\stackrel{\textnormal{iid}}{\sim}P$ and $(X_{n+1},Y_{n+1})\mid \tX_{n+1}\sim (P_X\circ H(\cdot, \tX_{n+1}))\times P_{Y|X}$,  as in~\eqref{eqn:explain_marginal_coverage_wrt}.
In particular, averaging over a random draw $\tX_{n+1}\sim \tP_{\tX}$ it holds that
\[\E_{\tX_{n+1}\sim\tP_{\tX}}\left[\P\left\{Y_{n+1}\in\hat{C}^{\textnormal{RLCP}}_n(X_{n+1},\tX_{n+1}) \ \middle| \  \tX_{n+1}\right\} \right] \geq 1-\alpha\]
or equivalently,
\[\P\left\{Y_{n+1}\in\hat{C}^{\textnormal{RLCP}}_n(X_{n+1},\tX_{n+1})\right\}\geq 1-\alpha\textnormal{ with respect to the joint distribution~\eqref{eqn:explain_marginal_coverage_wrt}}.\]

Therefore, we can see that, for the covariate shift setting (as in~\eqref{eqn:explain_marginal_coverage_wrt_now_with_g}), we have
the following lower bound on coverage:
\[\P\left\{Y_{n+1}\in\hat{C}^{\textnormal{RLCP}}_n(X_{n+1},\tX_{n+1})\right\}
\geq 1 - \alpha - \textnormal{d}_{\textnormal{TV}}\left(\begin{tabular}{c}\textnormal{the joint distribution}\\
\textnormal{of the data}\\
\textnormal{appearing in~\eqref{eqn:explain_marginal_coverage_wrt_now_with_g}}\end{tabular} , \begin{tabular}{c}\textnormal{the joint distribution}\\
\textnormal{of the data}\\
\textnormal{appearing in~\eqref{eqn:explain_marginal_coverage_wrt}}\end{tabular} \right).\]
Examining the difference between the two joint distributions derived in~\eqref{eqn:explain_marginal_coverage_wrt_now_with_g} and in~\eqref{eqn:explain_marginal_coverage_wrt}, we can simplify this result to
\[\P\left\{Y_{n+1}\in\hat{C}^{\textnormal{RLCP}}_n(X_{n+1},\tX_{n+1})\right\}
\geq 1 - \alpha - \E_{\tP_{\tX}}\left[\textnormal{d}_{\textnormal{TV}}\left(P_X\circ H(\cdot,\tX), P_X\circ g\circ H(\cdot, \tX)\right)\right].\]
From this point on, then, our only task is to bound this total variation distance. (In fact, we will explain in Appendix~\ref{app:smoothed_RLCP} that this total variation distance can also be used to bound the extent to which RLCP can \emph{overcover}, i.e., to ensure that the method is not too conservative.)

Now we can ask, for what functions $g$ might we expect this total variation distance to be small? If $H$ is a localizing kernel with a fairly small bandwidth, then the reweighted distribution $P_X\circ H(\cdot,\tX)$ is likely to place nearly all of its mass in a small neighbourhood around $\tX$. Consequently, if $g$ is Lipschitz (or some similar condition), then $g$ will be approximately constant on the bulk of the distribution $P_X\circ H(\cdot,\tX)$, and thus reweighting additionally by $g$ (i.e., computing the new distribution $P_X\circ g\circ H(\cdot,\tX)$) will only minimally change this distribution. The localization is very much necessary to this argument: if we ran conformal prediction without localization (i.e., CP rather than one of the LCP methods), following the same arguments would lead us to a term $\textnormal{d}_{\textnormal{TV}}(P_X,P_X\circ g)$, which is likely to be large even if $g$ is Lipschitz (aside from the trivial case where $g$ is a constant function, i.e., $P_X\circ g = P_X$).

To build a more precise understanding of this intuition, later on in Appendix~\ref{app:proof_thm:covariate_shift}, we will establish the bound
\begin{equation}\label{eqn:preview_TV_bound}\E_{\tP_{\tX}}\left[\textnormal{d}_{\textnormal{TV}}\left(P_X\circ H(\cdot,\tX), P_X\circ g\circ H(\cdot, \tX)\right)\right] \leq \frac{\E_{P_{\tX}}\left[ \sqrt{\frac{1}{2}\var_{P_{X|\tX}}(g(X)\mid \tX)}\, \right]}{\E_{P_X}[g(X)]}.\end{equation}
We can see that the conditional variance in the numerator is likely to be small when $H$ is strongly localized (i.e., when the conditional distribution of $X\mid \tX$ is concentrated in a small neighbourhood), while $g$ is Lipschitz, so that $g(X)$ has low variability conditional on $\tX$.
The remaining details, including proving the bound~\eqref{eqn:preview_TV_bound} and also establishing the upper bound in the theorem, are deferred to Appendix~\ref{app:proof_thm:covariate_shift}.

\section{Experiments}
\label{sec:experiment}
In our experiments, we will examine the empirical performance of our method, RLCP, and will compare its performance to the two locally weighted CP methods described in Section \ref{sec:review}---baseLCP and calLCP \citep{guan2023localized}. 
Code for reproducing all experiments can be found at \url{https://github.com/rohanhore/RLCP}.

Before we describe the experiments, we pause to comment on a technicality. For calLCP and RLCP, the marginal coverage guarantees ensure that coverage will be \emph{at least} $1-\alpha$; in some settings, the methods may be conservative, with coverage higher than the target level. (Of course, as discussed earlier, baseLCP does not give a marginal coverage guarantee.) The possibility of marginal overcoverage means that comparing the methods meaningfully is challenging: if one method achieves better local coverage, perhaps this might be because this method is more conservative in terms of its marginal coverage. With this in mind, in order to allow for a more meaningful comparison, in our experiments we implement \emph{smoothed} versions of each of the three methods, which (for calLCP and RLCP) ensures that marginal coverage is exactly $1-\alpha$, and leads to a more fair comparison. 
Details on the implementation of the smoothed versions of each of the methods can be found in Appendix \ref{app:smoothing}.

\subsection{Simulations}
For all our simulation experiments, our aim will be to achieve $90\%$ coverage, i.e., we choose $\alpha=0.1$. The feature space is given by $\calX = \RR^d$, and all experiments are run with localization kernel $H$ given by the Gaussian kernel, $H(x,x')=\frac{1}{(2\pi h^2)^{d/2}}\exp\left(-\frac{\|x-x'\|^2_2)}{2h^2}\right)$. \subsubsection{Univariate setting}\label{sec:simulation_univariate}
We begin with experiments in a univariate setting, i.e., $d=1$.
We consider the following two data generating distributions:
\begin{itemize}\label{simulation settings}
    \item \textbf{Setting $1$:} $X\sim \calN (0,1),~ Y|X\sim \calN \bigl(\frac{X}{2},|\sin (X)|\bigr)$
    \item \textbf{Setting $2$:} $X\sim \calN (0,1),~ Y|X\sim \calN \bigl(\frac{X}{2},\frac{4}{3}\phi\left(\frac{2X}{3}\right)\bigr)$,
    where $\phi(\cdot)$ is the density of a standard Gaussian.
\end{itemize}
 The two settings share the same feature distribution, $P_X = \calN(0,1)$. The difference lies in the conditional distribution of the response, $P_{Y|X}$. In Setting 2, the response has more variance for values of $X$ near the center of the distribution. In contrast, for Setting 1, the response has more variance for values of $X$ lying away from the mean. These different trends can be seen in Figure~\ref{fig:simulation_univariate_2}.

For both settings, our score function is given by $s(x,y) = |y-\hat{f}(x)|$, where $\hat{f}$ is a pretrained predictive model (trained using linear regression, using a separate dataset of $2000$ data points).
For each setting, we implement the methods using the localization kernel $H$ at five different bandwidths, $h\in \{0.1, 0.2, 0.4, 0.8, 1.6\}$; the lowest value represents a highly local $H$, while the highest value leads to a kernel $H$ that is essentially flat over the bulk of the feature distribution.
The methods are run with sample size $n=2000$, and evaluated on $2000$ test points. The entire experiment is repeated for $50$ independent trials.

\paragraph{Results: marginal coverage.}
Figure \ref{fig:simulation_univariate_1} displays the marginal coverage obtained by each of the three methods on the $2000$ test points, plotted over the $50$ independent trials of the experiment.
As expected from the theoretical guarantees, both RLCP and calLCP  achieve the target coverage level $1-\alpha = 90\%$. On the other hand, depending on the bandwidth choice, baseLCP can have marginal coverage smaller or larger than $1-\alpha$.
\begin{figure}[!h]
\centering
  \includegraphics[width =0.95\textwidth]{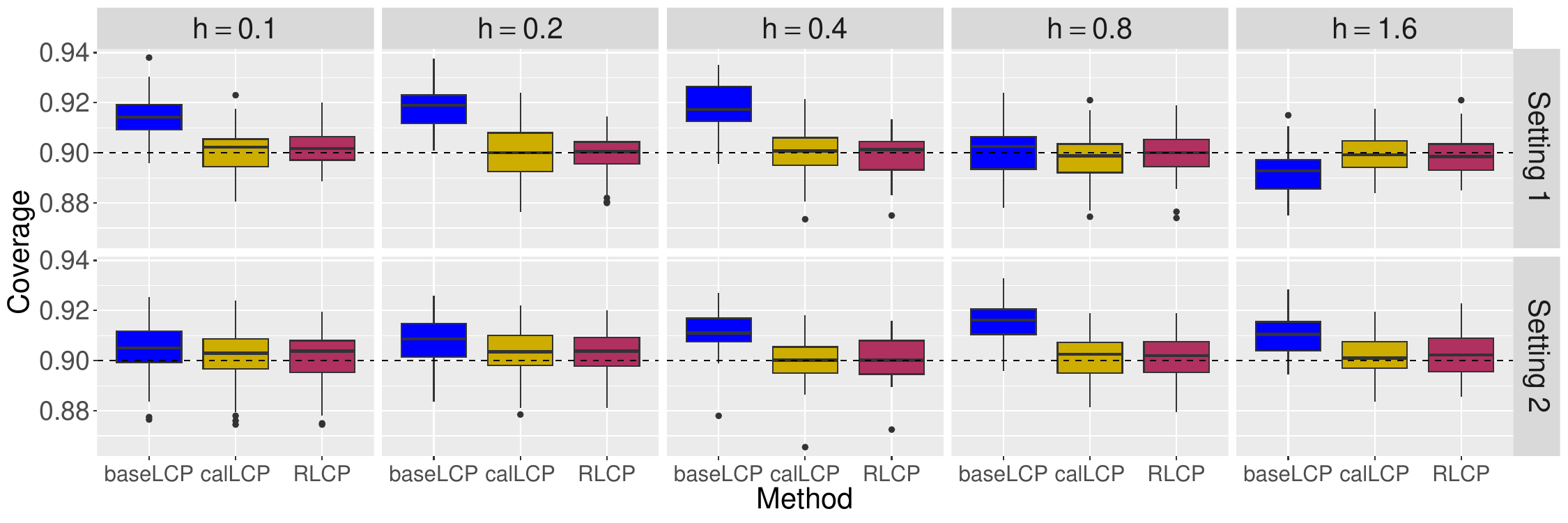}
  \caption{Marginal coverage of each method in  two univariate simulation settings, across different bandwidths. Results are shown for $50$ independent trials. See Section~\ref{sec:simulation_univariate} for details.}
   \label{fig:simulation_univariate_1}
\end{figure} 

\begin{figure}[!h]
\centering
    \includegraphics[width = 0.85\textwidth]{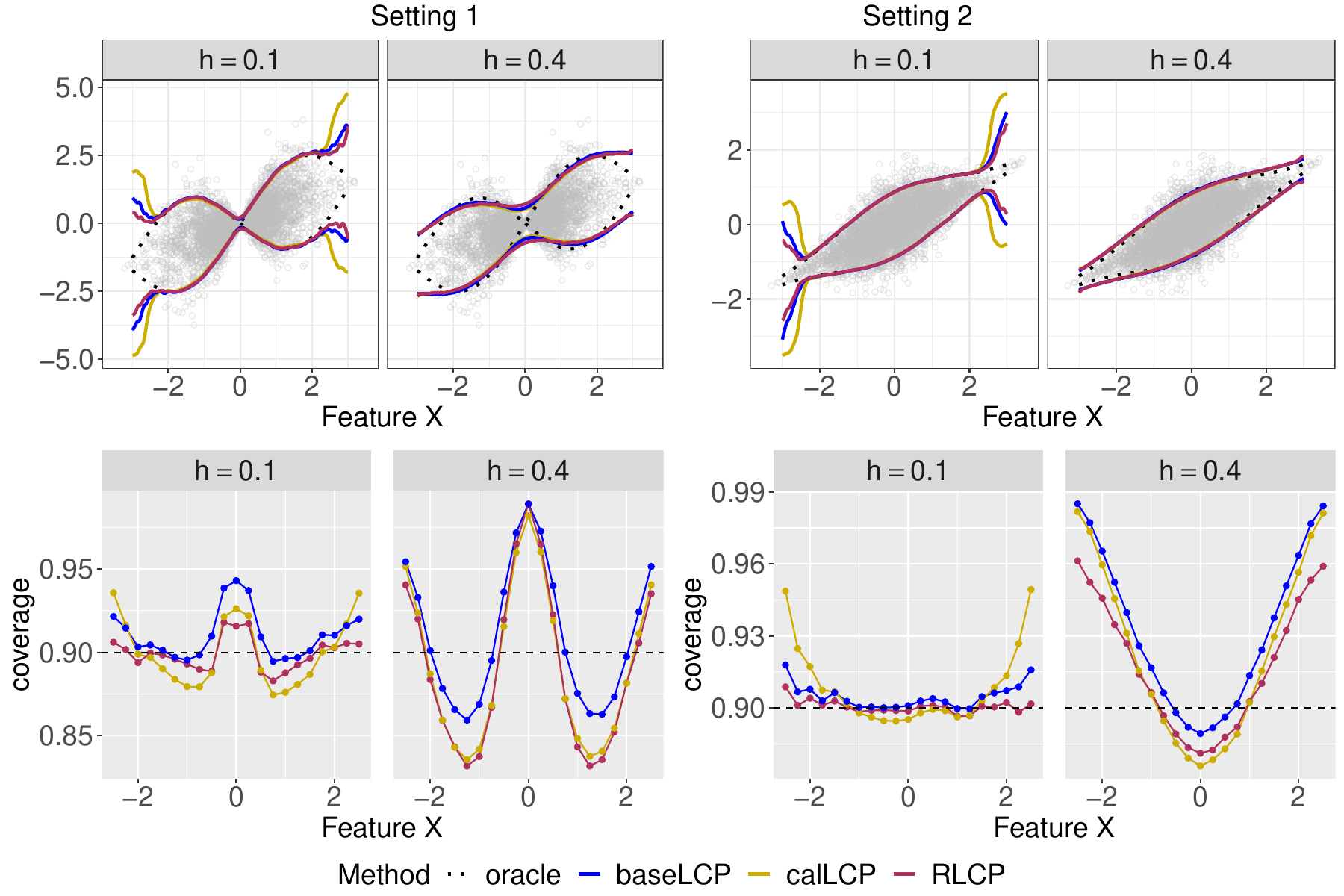}
    \caption{Top panels: the average prediction intervals produced by each method (at each feature value $x$, the endpoints of the prediction interval are averaged over $50$ independent trials), compared to the oracle prediction interval. Bottom panels: The corresponding local coverage of each method, averaged over $50$ independent trials. See Section~\ref{sec:simulation_univariate} for details.}
\label{fig:simulation_univariate_2}
\end{figure}

\paragraph{Results: local adaptivity.} 
One of the main goals addressed by calLCP in \citet{guan2023localized}'s work, and by RLCP in our work, is the goal of (approximate) test-conditional coverage, as in Aim~\ref{A2}. In order to study whether these methods are successfully addressing this aim empirically, we should look beyond marginal coverage, and examine whether the methods appear to provide local coverage in our experiments.

In the top panels of Figure \ref{fig:simulation_univariate_2}, we display the average values of the prediction interval returned by each method (that is, for each feature value $x\in\RR$, we compute the mean lower endpoint, and mean upper endpoint, of the prediction interval over the $50$ trials of the experiment). We compare the results from the three methods against the ``oracle'' prediction interval, computed with knowledge of the true distribution $P_{Y|X}$---i.e., the interval defined by the $5$th and $95$th percentiles of the conditional distribution of $Y$ given $X=x$, at each value $x$. In the bottom panels, we plot the local coverage of our prediction intervals, computed as the empirical coverage in balls of radius $0.4$ around a feature point, averaged over $50$ independent trials, and over the test observations lying within the ball. For brevity, we only present result for two bandwidths here; additional plots are shown in Appendix \ref{app:additional_univariate_results}.

For a smaller bandwidth $h$, the strong localization effect of the kernel $H$ ensures that the prediction intervals, for each of the three methods, adapt to the local variance of $P_{Y|X}$. For a larger bandwidth, however, this desirable performance is lost;
the localization effect diminishes (to be specific, the weighted quantile calculations in each method are being calculated with weights that are approximately constant), hence all three methods return approximately constant-width intervals that are not locally adaptive. Even for small $h$, we can see that the intervals appear conservative in the tails of the $X$ distribution, likely due to low effective sample size in these regions---however, this effect is more pronounced for calLCP, while RLCP offers more constant local coverage at the smaller $h$.

\paragraph{Additional results.} In Appendix \ref{app:additional_univariate_results}, in addition to showing plots for all values of the bandwidth $h$ (since Figure~\ref{fig:simulation_univariate_2} only shows results for $h=0.1,0.4$), we also compare with a related method from the recent literature---the work of \cite{gibbs2023conformal}, which we will describe below in Section~\ref{sec:discussion_literature}.

\subsubsection{Multivariate setting}\label{sec:simulation_multivariate}
Next, we will study the performance of these methods in a multivariate setting, again with the goal of comparing their local coverage properties. To do this in dimension $d>1$, since we cannot plot the prediction intervals at each value $x\in\RR^d$, we will instead examine the coverage properties of the methods on meaningful subsets of the feature space $\calX = \RR^d$.

We begin by generalizing the univariate setting: for a dimension $d\geq 1$, we consider the data distribution
$$ X\sim \calN_d (0,\mathbf{I}_d),~ Y|X\sim \calN \bigl(\sum_{i=1}^d\frac{ X_i}{2},\sum_{i=1}^d|\sin (X_i)|\bigr).$$
Note that in dimension $d=1$, this coincides with Setting 1 from our univariate simulation in Section~\ref{sec:simulation_univariate}. The details for implementing baseLCP, calLCP, and RLCP are the same as for the univariate setting:  we again use score function $s(x,y) = |y- \hat{f}(x)|$ for a pretrained (via linear regression) $\hat{f}$, and again use a Gaussian kernel $H$, and sample size $n=2000$. 

We will now examine the local coverage properties of each method in several different ways.

\paragraph{Results: coverage conditional on a set (with constant bandwidth $h$).}
To examine the test-conditional properties of these methods in high dimensions, we will focus on a specific task: coverage conditional on $X\in B$ for specific sets $B$. Concretely, we will define
\[B_{\textnormal{in}} = \{x\in\RR^d: \|x\|_2 \leq \tau_d\}, \quad 
B_{\textnormal{out}} = \{x\in\RR^d: \|x\|_2 > \tau_d\}\]
as the set of points $x$ lying inside and outside a ball of radius $\tau_d$. Here $\tau_d^2$ is chosen as the median of the $\chi^2_d$ distribution, which ensures that $P_X(B_{\textnormal{in}}) = P_X(B_{\textnormal{out}}) = 0.5$.

We will calculate local measures of predictive coverage for each of the three methods, across dimension $d\in\{1,5,10,15,\dots,50\}$. For each method and each dimension $d$, we compute an empirical estimate of the conditional coverage probability\footnote{For discussing the comparison between the three methods, we will denote the prediction interval as $\hat{C}_n(X_{n+1})$ for each method, even though for RLCP, the prediction interval is also a function of the random variable $\tX_{n+1}$.}
\[\P\{Y_{n+1}\in\hat{C}_n(X_{n+1})\mid X_{n+1}\in B\},\]
for three different choices of the set $B$: $B=B_{\textnormal{in}}$, $B=B_{\textnormal{out}}$, and $B=B_{\textnormal{in}}\cup B_{\textnormal{out}} = \calX$ (i.e., this last choice corresponds to marginal rather than conditional coverage). Intuitively, a method with strong local coverage properties should show approximately equal coverage on the inner and outer sets, $B_{\textnormal{in}}$ and $B_{\textnormal{out}}$.

We start with exploring the empirical local coverage of the conformal methods with a fixed bandwidth $h=1.5$ for all the methods across all dimensions (we will address the question of letting $h$ vary with dimension $d$ shortly). The results are shown in Figure \ref{fig:simulation_multivariate}.  RLCP performs significantly better than both baseLCP and calLCP in this setting. Aside from the smallest dimensions $d$, RLCP provides approximately equal coverage on both $B_{\textnormal{in}}$ and $B_{\textnormal{out}}$, with both conditional coverage levels approximately equal to the marginal coverage level $1-\alpha$. In contrast, while calLCP maintains exact marginal coverage (as guaranteed by the theory), it exhibits severely uneven coverages on the two sets $B_{\textnormal{in}}$ and $B_{\textnormal{out}}$. Finally, baseLCP shows extremely inconsistent performance across varying $d$, and can show marginal over- or undercoverage at various dimensions, but interestingly, the coverage level conditional on $B_{\textnormal{in}}$ or on  $B_{\textnormal{out}}$ are more similar to each other than for calLCP, in this setting.
\begin{figure}[!h]
\centering
  \includegraphics[width=0.8\textwidth]{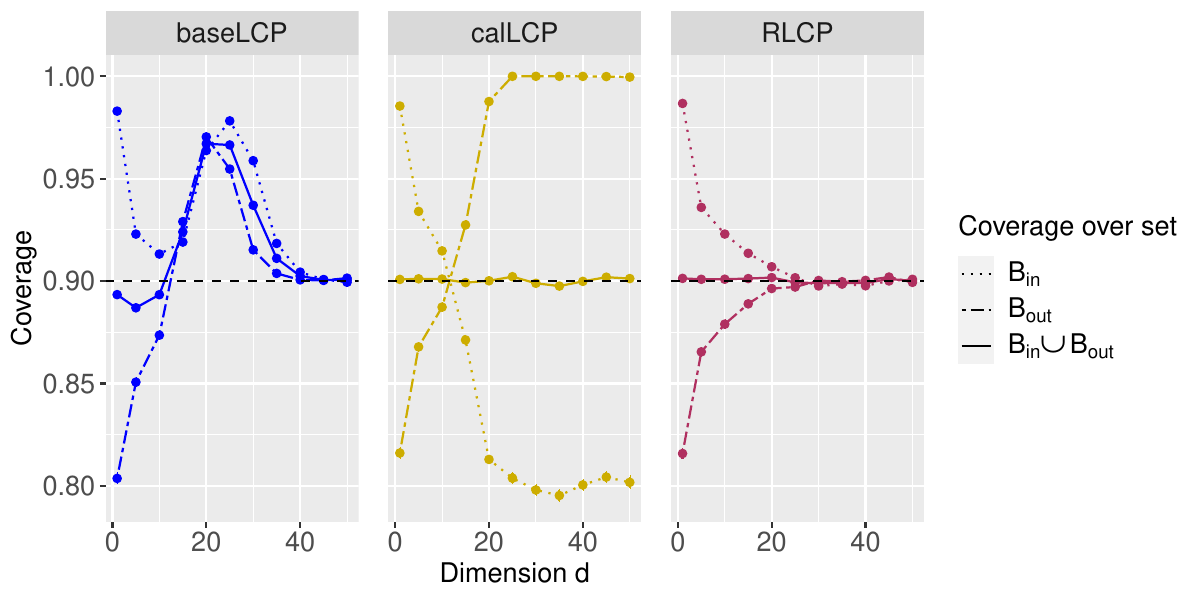}
  \caption{Predictive coverage conditional on $X_{n+1}\in B$, for three choices of the set $B$ (namely, $B_{\textnormal{in}}$, $B_{\textnormal{out}}$, and $B_{\textnormal{in}}\cup B_{\textnormal{out}} = \calX$). Results for each method with fixed bandwidth $h=1.5$ are plotted with respect to dimension $d$ of the feature space. Results are averaged over $50$ independent trials. See Section~\ref{sec:simulation_multivariate} for details.}
  \label{fig:simulation_multivariate}
\end{figure}

\paragraph{Results: coverage conditional on a set (with bandwidth $h$ varying with dimension $d$).}
The trends in performance observed above may be heavily influenced by the choice of the bandwidth $h$---and in particular, due to the curse of dimensionality, for a fixed value $h$ that does not change with dimension $d$, the amount of mass captured by $P_X$ within a radius-$h$ neighbourhood around any given $x$ will decrease exponentially rapidly as $d$ grows. In other words, using a fixed $h$ as $d$ increases, will lead to a sharply decreasing effective sample size for calculating the weighted quantiles; in a sense, the method is ``too localized'' for higher $d$. Therefore, for a more meaningful comparison across a range of dimensions $d$, we should choose the bandwidth $h$ in a dimension-adaptive way, to try to maintain a constant notion of effective sample size. 

Given feature values $X_1,\dots,X_{n+1}$, we recall the weights
\[w_i = \frac{H(X_i,X_{n+1})}{\sum_{j=1}^{n+1} H(X_j,X_{n+1})} .\]
 These weights appear in the calculation of the weighted quantile for both baseLCP and calLCP. To understand how the $w_i$'s relate to effective sample size,
let us consider the simple case where $H$ is a box kernel and each $X$ is likely to pick $k$ many neighbours out of a sample of size $n$ with this kernel: 
then for each $i$, we have (approximately) $k$ many $w_i$'s with value (approximately) $1/k$, and the rest zero. In this setting, for each $i$ we have
$\sum_i w_i^2 \approx 1/k$. Equivalently, we can write 
\[k \cdot \sum_i H(X_i,X_{n+1})^2 \approx \left(\sum_i H(X_i,X_{n+1})\right)^2.\]
Now, the terms in this expression are random,  but for large $n$ we can approximate this calculation by 
\[k \cdot n\E [H(X,X')^2] \approx \E[(n\E[H(X,X')\mid X])^2],\]
where the expected values are taken with respect to $X,X'\stackrel{\textnormal{iid}}{\sim}P_X$.
This then motivates the following definition of effective sample size:
\begin{equation}\label{eqn:n_eff}
n_{\textnormal{eff}}(h) = n\cdot \frac{\E[\E[H(X,X')\mid X]^2]}{\E [H(X,X')^2]}.\end{equation}

For baseLCP and calLCP, which each construct $\hat{C}_n(X_{n+1})$ with quantile weights $w_i$, we use a bandwidth $h$ that satisfies $n_{\textnormal{eff}}(h) = 50$ (where $n_{\textnormal{eff}}(h) $ is computed as in~\eqref{eqn:n_eff}, after using the pretraining set to compute estimates of the necessary expected values). For RLCP, on the other hand, the quantile calculation that defines the prediction interval is constructed using weights $\tilde{w}_i$, which is defined by calculating kernel values $H(X_i,\tX_{n+1})$ rather than $H(X_i,X_{n+1})$. Consequently, for this method, we define
\[
\tilde{n}_{\textnormal{eff}}(h) = n\cdot \frac{\E[\E[H(X,\tX')\mid X]^2]}{\E [H(X,\tX')^2]},\]
where now the expected value is taken with respect to drawing  $X,X'\stackrel{\textnormal{iid}}{\sim}P_X$ and then sampling $\tX'\sim H(X',\cdot)$.
We then implement RLCP with a bandwidth $h$ chosen to satisfy
$
\tilde{n}_{\textnormal{eff}}(h)= 50$
(where again, the expected values defining our notion of effective sample size, are estimated empirically on the pretraining set). 
Consequently, the various methods are being run with (potentially) different bandwidths, to ensure that they each have the same effective sample size---that is, to ensure that we are comparing these methods at approximately the same level of localization.

\begin{figure}[!h]
    \centering\includegraphics[width=0.8\textwidth]{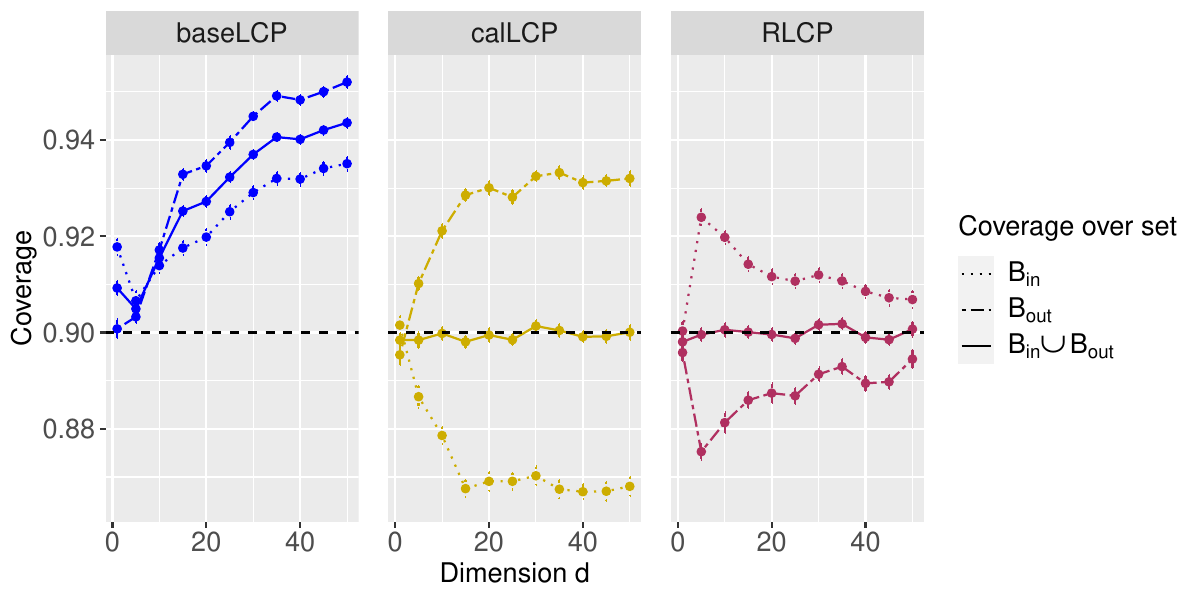}
    \caption{Predictive coverage conditional on $X_{n+1}\in B$, for three choices of the set $B$ (namely, $B_{\textnormal{in}}$, $B_{\textnormal{out}}$, and $B_{\textnormal{in}}\cup B_{\textnormal{out}} = \calX$). Results for each method, averaged over $50$ independent trials, are plotted with respect to dimension $d$ of the feature space. Each method is run with its own bandwidth choice, so that a constant $50$ effective sample size is maintained across dimensions. See Section~\ref{sec:simulation_multivariate} for more experimental details.}
    \label{fig:local_coverage_optimized_bandwidth_choices}
\end{figure}

The results, which are displayed in Figure \ref{fig:local_coverage_optimized_bandwidth_choices}, show that calLCP has more disparity in the coverage conditional on $B_{\textnormal{in}}$ versus on $B_{\textnormal{out}}$, as compared with RLCP: aside from the lowest values of dimension $d$, the coverage of RLCP on these subsets are around $1\%$ off from the target level, compared with around $3\%$ for calLCP. For both of these methods, as expected from the theory, the marginal coverage is exactly $1-\alpha$. (For baseLCP, the marginal coverage is quite conservative except at the smallest values of $d$, and hence the conditional coverage levels are high as well.)

\paragraph{Additional results.} 
The results shown above consider conditional coverage relative to only two specific sets $B$---namely, $B_{\textnormal{in}}$ and $B_{\textnormal{out}}$. For a more complete picture of the performance of the methods,
in Appendix \ref{app:experiments_coverage_on_highD_bins}, we consider a more general collection of sets $B$, and compare coverage of all these localized conformal methods conditioned on those sets. We will see a similar pattern of performance---divergence of coverage of calLCP from the target level is significantly higher compared with that of RLCP, for higher dimensions $d$.

\subsection{Real data example: predicting age of abalones}\label{sec:real_data}
Next, we will compare the performance of the localized CP methods on a real dataset  on abalones (a species of mollusc) \citep{nash1994population}.\footnote{The data for this experiment were obtained from \url{https://archive.ics.uci.edu/dataset/1/abalone}.} The objective is to predict the age of abalones based on their physical measurements. Since the economic value of abalone is positively correlated with its age, but determining age precisely requires a costly procedure, accurately predicting age from physical characteristics alone becomes both a valuable and practical problem. This dataset has $4177$ data points, where the response \texttt{rings} determines the age of the abalone, and the features consist of the abalone's sex (categorized as Male, Female, or Infant), length (the longest shell measurement), diameter (perpendicular to length), height, whole weight, shucked weight (i.e., weight of the meat inside), viscera weight (i.e., gut weight after bleeding) and shell weight (the weight after drying).

For our experiment, we will use $Y=\texttt{ring}$ as the response, and $X_{\texttt{sex}}$, $X_{\texttt{length}}$, $X_{\texttt{diameter}}$, $X_{\texttt{height}}$ and $X_{\texttt{whole\_weight}}$ as the features. After splitting the data equally into a pretraining set, a calibration set, and a test set,\footnote{Here we depart from the terminology of the earlier sections in our paper---the term \emph{calibration set} refers to the data $(X_1,Y_1),\dots,(X_n,Y_n)$ on which we implement the methods baseLCP, calLCP, and RLCP, using a predefined score function $s$ that is based on a pretrained predictor $\hat{f}$. This same set was previously called the \emph{training set} throughout the paper, but here we refer to it as the calibration set to distinguish it from the pretraining set used for fitting $\hat{f}$.} we train three different base predictors $\hat{f}$ on the pretraining set, using a linear regression, a $2$-layer neural network, or  a random forest. For each $\hat{f}$ we define the corresponding pretrained score function $s(x,y) = |y-\hat{f}(x)|$. Next, we will define the localization kernel $H$ to place high weight for pairs of patients that have the same sex, and similar values for length, diameter, height and whole weight: fixing a choice of bandwidths $h_{\texttt{length}}$, $h_{\texttt{diameter}}$, $h_{\texttt{height}}$ and $h_{\texttt{whole\_weight}}$, we define
\begin{align*}
    H(x,x') &= \frac{1}{2h_{\texttt{sex}} \cdot 2h_{\texttt{diameter}}\cdot 2h_{\texttt{height}}\cdot 2h_{\texttt{whole\_weight}}} \cdot \\
    &\ind\bigg\{x_{\texttt{sex}} = x'_{\texttt{sex}}, |x_{\texttt{length}} - x'_{\texttt{length}}|\leq h_{\texttt{length}},|x_{\texttt{diameter}} - x'_{\texttt{diameter}}|\leq h_{\texttt{diameter}}, \\
    &|x_{\texttt{height}} - x'_{\texttt{height}}|\leq h_{\texttt{height}},|x_{\texttt{whole\_weight}} - x'_{\texttt{whole\_weight}}|\leq h_{\texttt{whole\_weight}}\bigg\}.\end{align*}
In our implementation, for simplicity we will take $h_{\texttt{length}}=h_{\texttt{diameter}}=h_{\texttt{height}}=h_{\texttt{whole\_weight}}=h$ to always equal the same value, chosen from $h\in\{0.05,0.1,\ldots,0.25,0.3\}.$
This satisfies the condition~\eqref{eqn:condition_on_H}, i.e., that $H(x,\cdot)$ defines a density with respect to an appropriate base measure: namely, sampling from $H(x,\cdot)$ consists of sampling each of length, diameter, height and whole\_weight measurement uniformly from a $2h$ length bin around the corresponding test data measurement (e.g. length is an uniform random sample from $[x_{\texttt{length}} - h,x_{\texttt{length}} + h$), while copying the value of $x_{\texttt{sex}}$).

\begin{figure}[!h]
    \centering
    \includegraphics[width=0.95\textwidth]{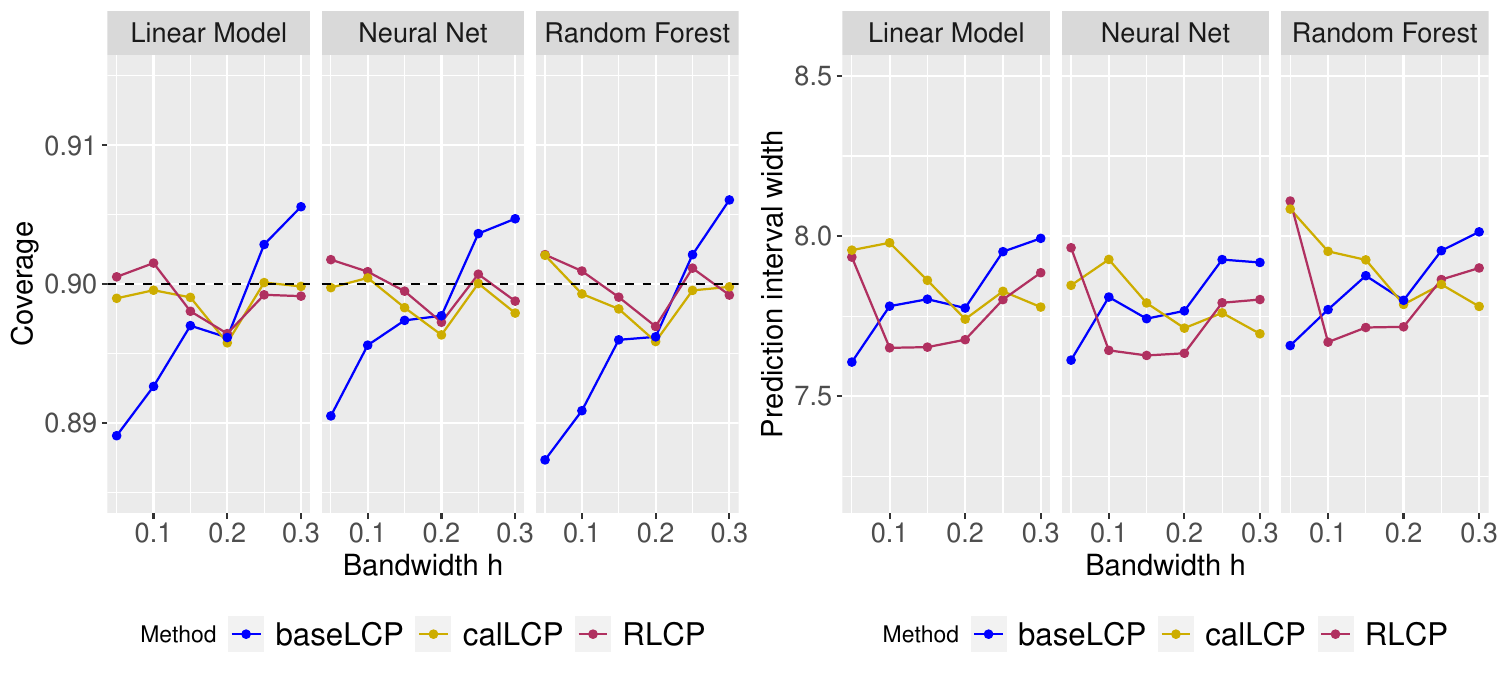}
    \caption{Marginal coverage (averaged over the test set) and  prediction interval width (median over the test set),  for each of the three methods across different bandwidths and for different choices of the pretrained base predictor $\hat{f}$, for the real data experiment. Results are averaged over $50$ random splits of the data into pretraining, calibration, and test sets. See Section~\ref{sec:real_data} for details.}
    \label{fig:real_data_1}
\end{figure}

\paragraph{Marginal coverage and prediction interval width.}
In Figure \ref{fig:real_data_1}, for each of the three localized CP methods, we present the marginal coverage and the median width of the prediction intervals on the test set, averaged over $50$ random splits of the data. While RLCP and calLCP attains exact target level of marginal coverage, coverage of baseLCP prediction intervals are lower and higher respectively for smaller and larger values of $h$. We further see that median length of prediction intervals are slightly narrower for RLCP prediction intervals for most bandwidth choices and for the various base algorithm choices. We will see next that RLCP offers better empirical conditional coverage at lower bandwidth values, which may justify the need for the prediction intervals with slightly larger width at the lowest values of $h$.

\paragraph{Conditional coverage.} In this real data setting, there are many practical questions we can ask regarding test-conditional coverage. In particular, since each data point corresponds to an individual abalone, test-conditional coverage can be interpreted as asking whether the predictive intervals are equally valid across different subgroups of abalones (e.g., smaller vs larger abalones). 
Here we aim to gain some insight into these questions by examining the empirical coverage of each method conditional on certain physical measurement for abalone. Ideally, we would expect approximately constant conditional coverage across different subgroups of abalones. 
In the left panel of Figure~\ref{fig:real_data_2}, we show coverage conditional on sex, i.e., conditional on the value of the three levels of covariate $X_{\texttt{sex}}$. The bar plot shows the value of $\P\{Y_{n+1}\in\hat{C}_n(X_{n+1})\mid (X_{n+1})_{\texttt{sex}} = x\}$ for each $x\in\{\texttt{Male},\texttt{Female},\texttt{Infant}\}$. While RLCP shows approximately equal coverage levels for the three groups across all three base predictors $\hat{f}$, for calLCP the performance is quite variable, with extremely disparate coverage levels for some choices of the base algorithm.

In the right panel of Figure~\ref{fig:real_data_2}, we ask the same question with respect to $X_{\texttt{length}}$. Since this is a continuous covariate, we cannot compute conditional coverage exactly, but instead estimate it with a sliding window whose width is chosen to have $5\%$ mass: that is, we are plotting $\P\{Y_{n+1}\in\hat{C}_n(X_{n+1})\mid (X_{n+1})_{\texttt{length}} \in [x - \Delta_0, x + \Delta_1)\}$ where, for each $x$, $\Delta_0$ and $\Delta_1$ are chosen so that $2.5\%$ of the abalone population has length in the range $[x - \Delta_0, x)$, and same for the range $[x,x+ \Delta_1)$. Here we again see that RLCP shows better conditional coverage compared with calLCP particularly at lower bandwidth $h=0.05$---the estimated conditional coverage for RLCP is approximately constant at the target level of $90\%$ while for calLCP, we see very prominent over- or undercoverage at various values of length.

\begin{figure}
    \centering
    \includegraphics[width=\textwidth]{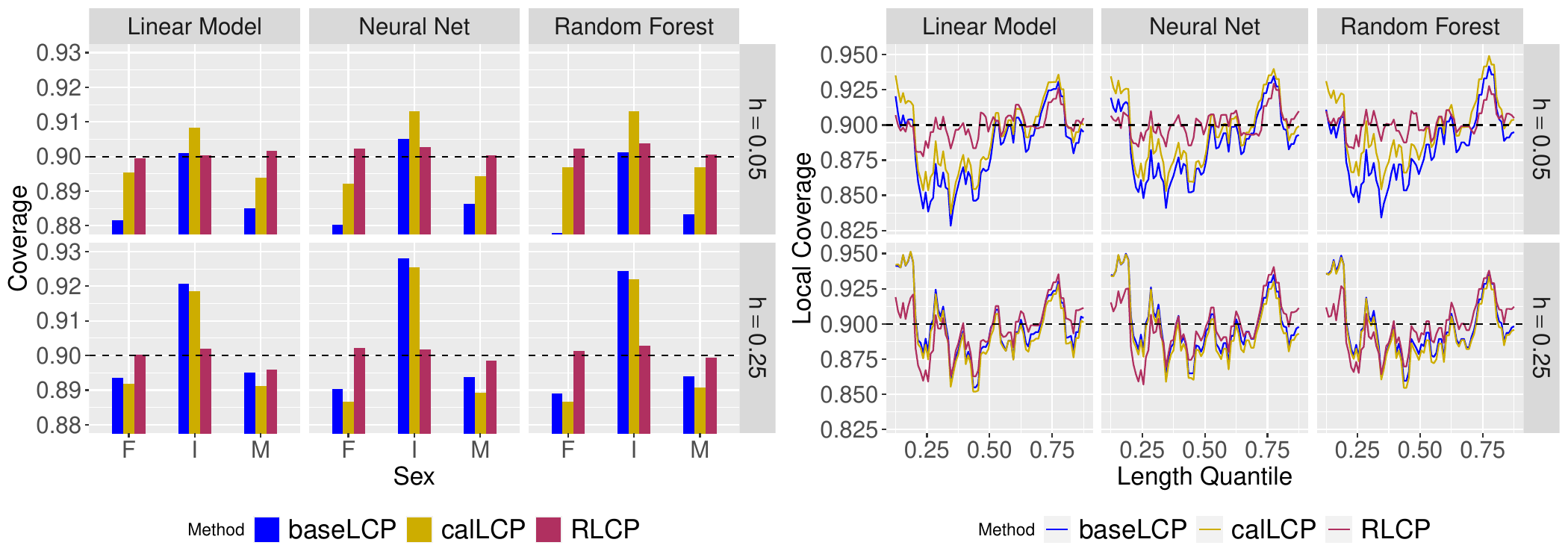}
    \caption{Comparison of local coverage of the three methods, conditioning on $X_{\texttt{sex}}$ (left) and on $X_{\texttt{length}}$ (right). See Section~\ref{sec:real_data} for details.}
    \label{fig:real_data_2}
\end{figure}

\section{Discussion and extensions}
\label{sec:discussion}
In this work, we have introduced a novel variant of the localized conformal prediction framework, RLCP, which attains distribution-free approximate coverage guarantees for relaxed test-conditional coverage and for smooth covariate shift problems while maintaining exact marginal coverage.  
As explained  in Section~\ref{sec:insights_to_results}, all of these coverage results are consequences of a core total variation distance based argument, which could potentially be used to derive other types of coverage guarantees as well. 
Additionally, RLCP also enjoys a training-conditional coverage guarantee. 
Our numerical simulations and data experiments advocate the use of RLCP over existing localized CP algorithms, due to its better performance at maintaining conditional coverage across a range of settings. 

In the remainder of this section, we will first compare RLCP with existing literature and methodology in closer detail, and will then explore the role of the randomization in RLCP. We will also consider an alternative formulation of calLCP, relating it to full conformal prediction.

\subsection{Related work}\label{sec:discussion_literature}
As we have described earlier, it is well known that achieving distribution-free pointwise test-conditional coverage is impossible to do (aside from trivial solutions such as $\hat{C}_n(X_{n+1})\equiv \calY$) \citep{vovk2012conditional,lei2014distribution}; even natural relaxations of the test-conditional coverage property also suffer from fundamental hardness results in the distribution-free setting \citep{barber2021limits}.
With these hardness results in mind, any progress on conditional coverage for conformal prediction methods must necessarily make compromises, whether by aiming for a weaker theoretical guarantee, by placing some assumptions on the problem, or by targeting empirical rather than theoretical performance. 

Almost all of these existing approaches in the conformal literature, including our proposal, can be thought of a specific instance of the following three-step general workflow:

\begin{center}
    \begin{tikzpicture}[node distance=2cm]
    \node (heurunc) [startstop,text width=3cm,label=Step (a)] {Choose conformal score $s$};
    \node (CP) [process,right of= heurunc,xshift=2.3cm,text width=3cm,label=Step (b)] {Choose a suitable localizer $H$\\ (if desired)};
    \node (rigunc) [decision,right of=CP,xshift=2.3cm,text width=3cm,label=Step (c)] {Calibration w.r.t a class of functions $\mathcal F$\\ (if desired)};
    
    \draw [arrow] (heurunc) -- (CP);
    \draw [arrow] (CP) -- (rigunc);
    \end{tikzpicture}
    \end{center}

We now explain these steps, and highlight how some existing methods can be interpreted in terms of this three-step workflow.

\paragraph{Step (a): choosing a score $s$.}
 The split conformal method, defined in Section~\ref{sec:background_CP}, can be implemented with any choice of the score function $s$ (as in Step (a)). For example, we might choose the residual score function, $s(x,y) = |y-\hat{f}(x)|$ (for a prefitted regression function $\hat{f}$), as discussed earlier. A popular alternative is the conformalized quantile regression (CQR) method \citep{romano2019conformalized}, which uses the score function $s(x,y)\max\{\hat Q_\alpha(x)-y,y-\hat Q_{1-\alpha}(x)\}$, where $\hat Q$ denotes a pretrained estimated quantile. 
 
 In general, a score function $s$ that captures the underlying heterogeneity of $P_{Y|X}$ (such as CQR) will tend to exhibit good local coverage empirically. For additional examples of score functions of this type, see the works of \citet{lei2018distribution, gyorfi2019nearest,chernozhukov2021distributional,sesia2021conformal,deutschmann2024adaptive,feldman2021improving}; see also \citet{gupta2022nested, sesia2020comparison} for some empirical comparison of local coverage properties of various choices of $s$. Some of these methods offer (asymptotic) local coverage guarantees under additional distributional assumptions, e.g., \citet{sesia2021conformal,gyorfi2019nearest}.
 However, split conformal---with \emph{any} choice of $s$---does not have any distribution-free guarantees for approximate local coverage. 

 \paragraph{Step (b): choosing a localizer $H$.} Next, we turn to Step (b)---choosing a localizer $H$. This  means that when constructing the prediction interval, instead of using a single threshold $\hat{q}$ for \emph{any} value of the test point $X_{n+1}$, we instead compute a localized threshold via a weighted quantile, with weights determined by the localizing kernel $H$---exactly as in both calLCP and RLCP. In particular, both calLCP and RLCP require \emph{both} choosing a score function $s$ and then choosing a localizer $H$, in order to implement an instance of the method (e.g., we might combine the CQR score with a box kernel localizer). As we have seen in this work, for RLCP, adding Step (b) into our workflow allows for distribution-free theoretical guarantees that approximate test-conditional coverage, as in Theorem~\ref{thm:test_conditional}. 
 
 Another type of approach that can be viewed through the lens of Step (b), is the approach of partitioning $\calX$ into bins, and then computing $\hat q$ separately within each bin; methods of this type have been explored by, e.g., \citet{lei2014distribution,izbicki2022cd,leroy2021md}, and can be related to Step (b) by choosing $H(x,x') = \ind\{\textnormal{$x,x'$ are in the same bin}\}$. Alternatively, $H$ can be chosen as a nearest-neighbour based kernel \citep{ghosh2023improving}.

 We remark that for the split conformal method (which does not incorporate Step (b)), we can equivalently view this method as choosing the \emph{trivial} localizer $H$ given by $H(x,x')\equiv 1$, since computing an unweighted quantile (to define the threshold $\hat q$ for the split conformal interval) is equivalent to computing a weighted quantile with constant weights $w_1 = \dots=w_{n+1}$. 

\paragraph{Step (c): choosing a function class $\cal F$.} Finally, we turn to the third step: calibration with respect to a class of functions $\cal F$. This step does not appear in our RLCP method, but plays a key role in some of the related methods in the literature. Specifically, in the work of \cite{gibbs2023conformal}, calibration with respect to $\cal F$ means that prediction intervals are required to satisfy\footnote{In the work of \cite{gibbs2023conformal}, functions $f\in \cal F $ are not constrained to be nonnegative, and so this requirement is written as $=0$, not $\geq 0$ (of course, obtaining coverage at level exactly $1-\alpha$ requires using a smoothed version of the method). Here we use a formulation more comparable to the notation of our paper.}
\[\E\left[ f(X_{n+1}) \cdot \left(\ind\{Y_{n+1}\in \hat C_n(X_{n+1})\} - (1-\alpha) \right)\right] \geq  0~~ \textnormal{for all}~~f\in \mathcal F.\]
If we take $\cal F$ to be the set of all measurable functions, this would be equivalent to a test-conditional coverage guarantee as in~\eqref{eqn:test_conditional_coverage}---but as discussed before, is impossible to achieve distribution-free. \cite{gibbs2023conformal} therefore implement their method with more restricted classes of functions $\cal F$, for instance linear functions or functions from a Reproducing Kernel Hilbert Space (RKHS), to achieve an approximate notion of local coverage; their method is implemented via solving a quantile regression problem over the class $\mathcal{F}$. As another example, we can choose $\mathcal F=\{\ind\{\cdot \in G\}:G\in \mathcal G\}$ for a collection $\mathcal G$ of subsets of $\calX$; this amounts to requiring group-conditional coverage for all groups from the collection $\mathcal G$, as has been studied in \citet{jung2023batch}. 

At the other extreme, a marginal coverage guarantee can be interpreted as requiring calibration with respect to the trivial class of functions, $\cal F = \{\mathbf{1}\}$ (where $\mathbf{1}$ is the constant function $x\mapsto 1$); thus methods such as split conformal, calLCP, or RLCP, which do not incorporate Step (c) into their workflow, can be equivalently interpreted as applying Step (c) with this trivial choice $\cal F = \{\mathbf{1}\}$.

\paragraph{Implications for local coverage and covariate shift.}
To summarize, we have now seen that many existing methods (and our proposed RLCP method) can be compared under this common high-level framework. What does this tell us in terms of their local coverage properties (and coverage under covariate shift)? As mentioned above, methods that only use Step (a) (and therefore can be viewed as choosing the trivial localizer $H(x,x')\equiv 1$ for Step (b), and the trivial function class $\cal F = \{\mathbf{1}\}$ for Step (c)) do not achieve distribution-free theoretical guarantees of (approximate) local coverage, although of course the choice of score $s$ in split conformal can lead to substantial empirical improvements in local coverage in practice. 

On the other hand, methods that use \emph{either} localization at Step (b) (such as our RLCP procedure) or a nontrivial function class at Step (c) (such as the method of \cite{gibbs2023conformal}) can offer meaningful approximate local coverage theory in addition to strong performance empirically. Using localization (Step (b)) or calibration with respect to a function class (Step (c)) can therefore be viewed as complementary options towards achieving the same type of aim---we will further investigate this claim empirically by comparing the local coverage of RLCP and the method from \citet{gibbs2023conformal} in some of our simulation settings in Appendix \ref{app:additional_univariate_results}.
An interesting open question is whether stronger guarantees and/or better empirical conditional coverage might be obtained by developing methods that incorporate both Step (b) and Step (c) simultaneously.

\subsection{Is it possible to derandomize RLCP?} \label{sec:derandomizing}
In our discussion in Section~\ref{sec:discussion_literature}, we showed that a range of methods from the literature can all be viewed through a three-step framework for constructing localized versions of conformal prediction---including our proposed method, RLCP. One way that RLCP differs from the other methods discussed, however, is the fact that randomization is inherent in its construction. In this section, we take a closer look at the question of randomization: is the randomized construction of RLCP necessary for ensuring its favourable theoretical and empirical properties, and does it cause substantial variability in the output in practice?

We will begin with the latter question, by examining empirically whether the random draw of $\tX_{n+1}$ in the construction of RLCP causes substantial randomness in the output. To explore this, we return to the univariate simulations from Section~\ref{sec:simulation_univariate}. Whereas before we plotted the \emph{average} lower and upper endpoints of the prediction intervals across $100$ independent trials of the experiment (in Figure~\ref{fig:simulation_univariate_2}), now we generate a \emph{single} dataset and compute both the calLCP and RLCP prediction intervals, but for RLCP, we repeat the random draw of $\tX_{n+1}$ $100$ times to examine the randomness in the construction. Figure~\ref{fig:RLCP_variability} shows the results: the shaded area visualizes the 5th and 95th percentile of the lower and upper endpoints of the RLCP prediction band. We can see that, on average, calLCP and RLCP are giving similar results here (as was already shown in Figure~\ref{fig:simulation_univariate_2}), but the variability of RLCP is substantial in some settings (e.g., particularly for Setting 1 at the lower bandwidth values, where the localization is strong). 
We also examine the randomness of RLCP in a more quantitative way in Appendix~\ref{app:RLCP_deviation}; we will see that, in the worst case, the RLCP prediction interval can vary by around $20$--$25\%$ in width relative to the average output.
\begin{figure}[!h]
\centering
\includegraphics[width = 0.7\textwidth]{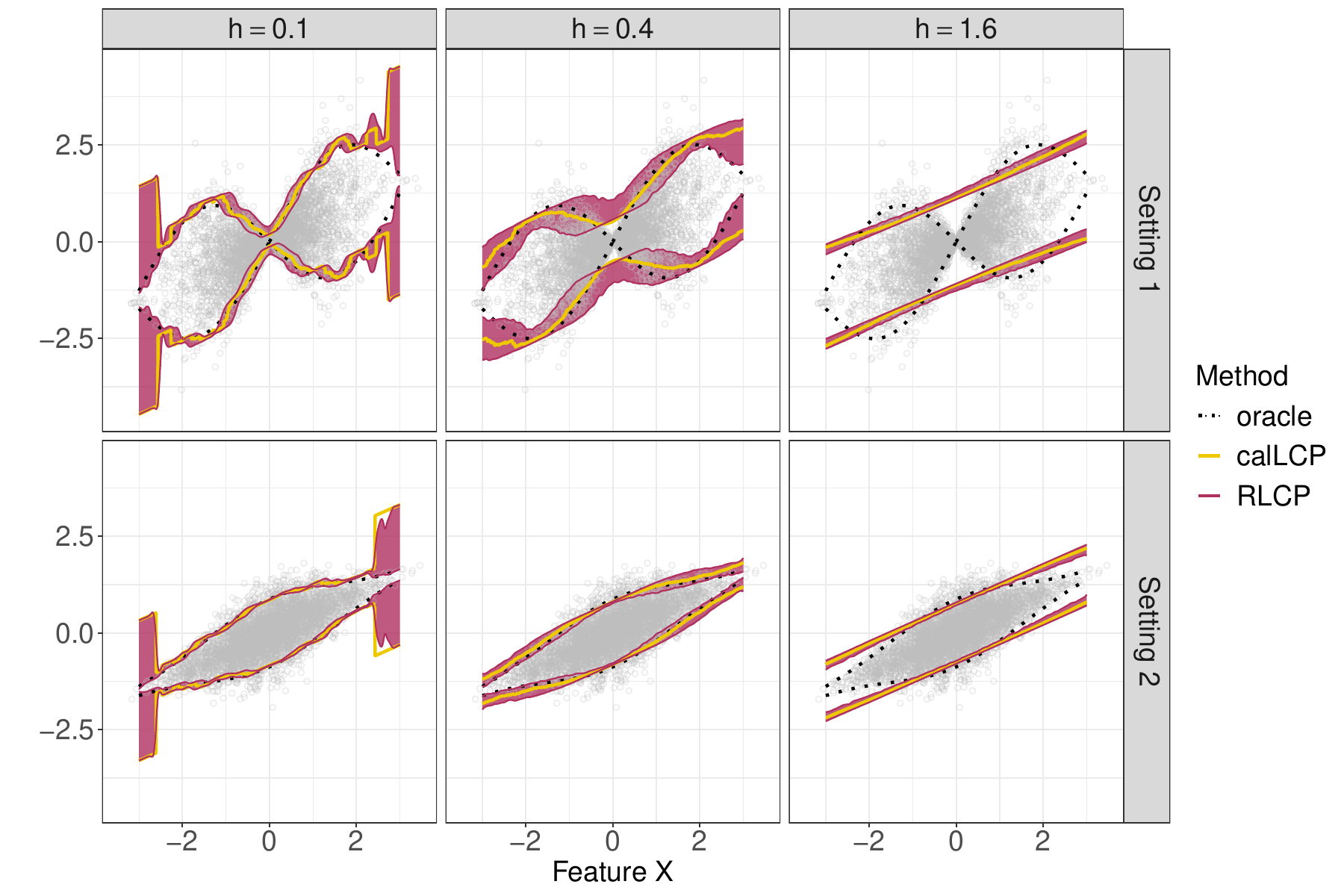}
\caption{A visualization of the variability of the RLCP prediction band for a fixed dataset, showing the 5th and 95th percentile of $300$ draws of constructing RLCP using the same simulated dataset. We also compare to the oracle prediction interval and to the calibrated localized conformal prediction interval. See Section~\ref{sec:derandomizing} for details.}
\label{fig:RLCP_variability}
\end{figure}
 
Hence, depending on the underlying setting, this variability of RLCP can be substantial, and may be a substantial limitation of this method in practice. On the other hand, the randomized construction is at the heart of the ability of RLCP to obtain stronger and cleaner theoretical guarantees relative to existing methods (as we have discussed 
in Sections~\ref{sec:role_of_Xt} and~\ref{sec:insights_to_results}), 
as well as its favourable empirical performance in terms of local coverage.

With this trade-off in mind, this naturally suggests the following question: is it possible to ``derandomize'' the RLCP method, perhaps by averaging or aggregating across multiple runs of RLCP on a given dataset, to retain its strong theoretical and empirical properties but reduce the issue of variability due to randomization?

\subsubsection{Derandomizing RLCP via $m$-RLCP}

We will now propose one  possible approach towards derandomizing RLCP, by averaging over $m$ many runs in a particular way. We will see that this construction is very natural, but unfortunately its empirical performance is substantially worse than RLCP itself. Through this exploration, we hope to inspire continued research on this question to determine whether the randomness of RLCP is somehow critical to its strong performance, or whether a different derandomization scheme may instead offer better results. 

Before we begin, we first reformulate RLCP through the language of p-values. For all methods within the conformal prediction framework, the prediction interval can be constructed by considering, for each possible value $y$, the hypothesis test that asks whether we can reject the hypothesis that $Y_{n+1}=y$ (in which case $y$ is excluded from the prediction interval), or cannot reject it (in which case $y$ is included); the p-value from this hypothesis test can be viewed as an equivalent formulation of the conformal prediction methodology. For the case of RLCP specifically, given the training set $\{(X_i,Y_i)\}_{i\in[n]}$, define a \emph{randomized} p-value
\begin{equation}\label{eqn:RLCP_as_pvalue}p(y;X_{n+1},\tX_{n+1}) = \sum_{i=1}^n\frac{H(X_i,\tX_{n+1})}{\sum_{j=1}^{n+1}H(X_j,\tX_{n+1})}\ind\{s(X_i,Y_i)\geq s(X_{n+1},y)\} + \frac{H(X_{n+1},\tX_{n+1})}{\sum_{j=1}^{n+1}H(X_j,\tX_{n+1})}.\end{equation}
As for other conformal methods, the prediction interval for RLCP can equivalently be constructed as
\[\hat{C}_n^{\textnormal{RLCP}}(X_{n+1},\tX_{n+1}) = \{y\in\calY: p(y;X_{n+1},\tX_{n+1}) > \alpha\}.\]

We are now ready to derandomize. Firstly, given the test feature we generate $m$ copies of $\tX_{n+1}$, i.e. $\tX_{n+1}^{[m]} = (\tX_{n+1}^1,\tX_{n+1}^2,\cdots, \tX_{n+1}^m)$, with each value drawn i.i.d.\ from $H(X_{n+1},\cdot)$. For each $y\in\calY$, we then calculate the corresponding $m$ p-values,
\[p(y;X_{n+1},\tX_{n+1}^1),\dots,p(y;X_{n+1},\tX_{n+1}^m)\]
as in~\eqref{eqn:RLCP_as_pvalue} above, and then average these p-values,
\begin{equation}\label{eqn:define_bar_p}\bar{p}(y;X_{n+1},\tX_{n+1}^{[m]}) = \frac{p(y;X_{n+1},\tX_{n+1}^1) + \dots + p(y;X_{n+1},\tX_{n+1}^m)}{m}.\end{equation}
Finally, we return the prediction interval
\[
\hat{C}_{n}^{\textnormal{$m$-RLCP}}(X_{n+1},\tX_{n+1}^{[m]})=\{Y_{n+1}: \bar{p}(y;X_{n+1};\tX^{[m]}_{n+1})> \alpha\}.
\]
We will refer to this method as $m$-RLCP. 

Since we are averaging the p-value over $m$ random draws of $\tX_{n+1}$, for large $m$, the $m$-RLCP prediction interval is essentially a deterministic function of the data (i.e., of the training data $\{(X_i,Y_i)\}_{i\in[n]}$ along with the test feature $X_{n+1}$). 

It is well known that the average of valid p-values is itself, up to a factor of 2, a valid p-value \citep{ruschendorf1982random,vovk2020combining}, i.e., we have
\[\P\{\bar{p}(Y_{n+1};X_{n+1},\tX_{n+1}^{[m]})\leq \alpha\} \leq 2\alpha.\]
This theoretical guarantee suggests that the $m$-RLCP prediction interval may be somewhat anticonservative, i.e., coverage is only guaranteed at level $\geq 1-2\alpha$.
Although this result only addresses the potential under-coverage, in practice we observe that $m$-RLCP can in fact severely \emph{overcover}, i.e., may be extremely conservative (we show  an empirical comparison of RLCP in $m$-RLCP later on in Appendix \ref{app:mRLCP experiments}). In fact, a heuristic argument explains that $m$-RLCP behaves similarly to baseLCP, but with a different localizer kernel; our empirical findings on $m$-RLCP thus agree with the fact that baseLCP can often be conservative in practice, as we have seen empirically in the earlier experiments. We provide more details on the connection between $m$-RLCP and baseLCP in Appendix~\ref{app:mRLCP results}.

To summarize, this intuitive approach towards derandomizing RLCP shows substantial drawbacks in terms of empirical performance (as well as weaker theory as compared with RLCP). Thus, finding a better way to aggregate across multiple runs of RLCP, to achieve a less random method without compromising on theoretical properties or empirical performance, remains an important open question.

\subsection{Reinterpreting calLCP as an instance of full conformal prediction}\label{sec:calLCP_fullCP}
To conclude our discussion of related work, we will now examine \citet{guan2023localized}'s calLCP method (as we recall, simply referred to as ``LCP'' in that work) through a new lens. Specifically, we make the observation that, with a properly chosen score function, calLCP is actually equivalent to running full conformal prediction \citep{vovk2005algorithmic}. 

In this section, we will use the following formulation of full conformal prediction. We are given a function $T:(\calX\times\calY)^{n+1}\rightarrow\RR^{n+1}$, which maps datasets of size $n+1$ to a vector of scores (i.e., the $i$th score corresponds to the $i$th data point, with large scores indicating that the data point appears unusual relative to the other data in the dataset). The map $T$ is required to commute with permutations (i.e., swapping data points $i$ and $j$ results in swapping the $i$th and $j$th scores returned by $T$).
The full conformal prediction interval (or more generally, prediction set) is then defined as
\begin{equation}\label{eqn:fullCP}\hat{C}^{\textnormal{fullCP}}_n(X_{n+1}) = \{y\in\calY : T_{n+1}^y \leq \textnormal{Quantile}_{1-\alpha}(T_1^y,\dots,T_n^y,T_{n+1}^y)\},\end{equation}
where for any $y\in\calY$, the scores are defined as
\[T_i^y = \Big(T\big((X_1,Y_1),\dots,(X_n,Y_n),(X_{n+1},y)\big)\Big)_i, \ i=1,\dots,n+1,\]
i.e., the scores of the $n+1$ data points, if we assume that the test point is equal to $(X_{n+1},y)$. Full conformal prediction satisfies the marginal coverage guarantee for i.i.d.\ (or exchangeable) data \citep{vovk2005algorithmic}.

Now we make the correspondence between calLCP and full CP precise. After pretraining the original score function $s=s(x,y)$ on a separate dataset, we define the map $T$ as follows:  for any dataset $(x_1,y_1),\dots,(x_{n+1},y_{n+1})$, the new score for the $i$th data point is given by
\[\Big(T\big((x_1,y_1),\dots,(x_{n+1},y_{n+1})\big)\Big)_i = \sum_{j=1}^{n+1} w_{i,j} \ind\{s(x_j,y_j) < s(x_i,y_i)\},\]
where
\[w_{i,j} = \frac{H(x_j,x_i)}{\sum_{k=1}^{n+1} H(x_k,x_i)}.\]
In Appendix~\ref{app:proof_prop:calLCP_fullCP}, we will prove that this is an equivalent representation of the calLCP method:
\begin{proposition}\label{prop:calLCP_fullCP}
The prediction interval $\hat{C}_n^{\textnormal{calLCP}}(X_{n+1})$ returned by calLCP, is always equal to the prediction interval $\hat{C}^{\textnormal{fullCP}}_n(X_{n+1})$ returned by full conformal prediction, when implemented with the score map $T$ defined as above.
\end{proposition}
Consequently, the marginal coverage guarantee for calLCP follows as a corollary: since full CP satisfies marginal coverage, calLCP must therefore inherit this property, since it is an instance of full CP. 
On the other hand, establishing conditional coverage type properties is quite
challenging for calLCP (indeed, this is often the case for full CP in general), while the randomized construction of RLCP opens the door to the local theoretical guarantees that we have established above.

\subsection*{Acknowledgements}
R.F.B. was partially supported by the Office of Naval Research via grant N00014-20-1-2337, and by the National Science Foundation via grant DMS-2023109.
The authors are grateful to Emmanuel Cand{\`e}s, Aaditya Ramdas, and Ryan Tibshirani for helpful discussions.

\bibliographystyle{apalike}
\bibliography{ref}
\appendix 
\renewcommand{\theequation}{\thesection.\arabic{equation}}
\renewcommand{\thelemma}{\thesection.\arabic{lemma}}
\renewcommand{\thefigure}{\thesection.\arabic{figure}}
\renewcommand{\thetheorem}{\thesection.\arabic{theorem}}
\renewcommand{\theproposition}{\thesection.\arabic{proposition}}

\section{Proofs and few additional results}

\subsection{Completing the proof of Theorem \ref{thm:covariate_shift}}\label{app:proof_thm:covariate_shift}
With the calculations shown in the proof sketch (in Section~\ref{sec:insights_to_results}) in place, we now need to verify two claims: we need to prove the bound~\eqref{eqn:preview_TV_bound} on the expected total variation distance, and we need to prove the claim in the theorem itself.

First we verify that~\eqref{eqn:preview_TV_bound} holds. First,  by definition, $P_X\circ g\circ H(\cdot,\tX)$ is absolutely continuous with respect to $P_X\circ H(\cdot,\tX)$, and we can compute the Radon--Nikodym derivative relating these two distributions as
\[\frac{\mathsf{d}(P_X\circ g\circ H(\cdot,\tX))(x)}{\mathsf{d}(P_X\circ H(\cdot,\tX))(x)} = g(x) \cdot \frac{\int_{\calX}H(x',\tX)\;\mathsf{d}P_X(x')}{\int_{\calX}g(x')H(x',\tX)\;\mathsf{d}P_X(x')}= \frac{g(x)}{\E_{X'\sim P_X\circ H(\cdot,\tX)}[g(X')]}.\]
Therefore, treating $\tX$ as fixed in all the following calculations, we have
\begin{multline*}\textnormal{d}_{\textnormal{TV}}\left(P_X\circ H(\cdot,\tX), P_X\circ g\circ H(\cdot, \tX)\right) = \frac{1}{2}\E_{P_X\circ H(\cdot,\tX)}\left[ \left| \frac{g(X)}{\E_{X'\sim P_X\circ H(\cdot,\tX)}[g(X')]} - 1 \right|\right] \\
=\frac{\E_{X\sim P_X\circ H(\cdot,\tX)}\left[ | g(X) - \E_{X'\sim P_X\circ H(\cdot,\tX)}[g(X')]|\right]}{2\E_{P_X\circ H(\cdot,\tX)}[g(X)]}\leq\frac{\E_{X,X'\stackrel{\textnormal{iid}}{\sim}P_X\circ H(\cdot,\tX)}\left[ | g(X) - g(X')|\right]}{2\E_{P_X\circ H(\cdot,\tX)}[g(X)]}\\ = 
\E_{X,X'\stackrel{\textnormal{iid}}{\sim}P_X\circ H(\cdot,\tX)}\left[ | g(X) - g(X')|\right] \cdot \frac{\E_{P_X}[H(X,\tX)]}{2\E_{P_X}[g(X)H(X,\tX)]},
\end{multline*}
where the inequality holds by Jensen's inequality.
Next we need to marginalize over $\tX$, i.e., $\tP_{\tX}$. We have defined this distribution as the marginal distribution of $\tX$ when we first draw $X\sim P_X\circ g$, then draw $\tX\sim H(X,\cdot)$ (where we recall that $H(X,\cdot)$ is a density with respect to some base measure $\nu$), meaning that $\tP_{\tX}$  has density at $x$ given by
$\E_{P_X\circ g}[H(X,x)] = \frac{\E_{P_X}[g(X)H(X,x)]}{\E_{P_X}[g(X)]}$.
Consequently,
\begin{multline*}\E_{\tP_{\tX}}\left[ \textnormal{d}_{\textnormal{TV}}\left(P_X\circ H(\cdot,\tX), P_X\circ g\circ H(\cdot, \tX)\right)\right]\\ 
\leq \int_{\calX} \frac{\E_{P_X}[g(X)H(X,x)]}{\E_{P_X}[g(X)]} \cdot \E_{X,X'\stackrel{\textnormal{iid}}{\sim}P_X\circ H(\cdot,x)}\left[ | g(X) - g(X')|\right] \cdot \frac{\E_{P_X}[H(X,x)]}{2\E_{P_X}[g(X)H(X,x)]}\;\mathsf{d}\nu(x)\\
= \int_{\calX} \E_{X,X'\stackrel{\textnormal{iid}}{\sim}P_X\circ H(\cdot,x)}\left[ | g(X) - g(X')|\right] \cdot \frac{\E_{P_X}[H(X,x)]}{2\E_{P_X}[g(X)]}\;\mathsf{d}\nu(x)\\
= \frac{\E_{P_{\tX}}\left[\E_{X,X'\stackrel{\textnormal{iid}}{\sim}P_{X|\tX}}\left[ | g(X) - g(X')|\right]\right]}{2\E_{P_X}[g(X)]},
\end{multline*}
where for the last step, we use the marginal distribution $P_{\tX}$ of $\tX$ under the joint distribution $P_{(X,\tX)}$, where we first draw $X\sim P_X$ and then $\tX\sim H(X,\cdot)$; this distribution has density $\E_{P_X}[H(X,x)]$ with respect to $\nu$.
Finally, standard calculations show that
\[\E_{X,X'\stackrel{\textnormal{iid}}{\sim}P_{X|\tX}}\left[ | g(X) - g(X')|\right]\leq\E_{X,X'\stackrel{\textnormal{iid}}{\sim}P_{X|\tX}}\left[ (g(X) - g(X'))^2\right]^{1/2} = \sqrt{2\textnormal{Var}_{P_{X|\tX}}(g(X))},\]
which completes the proof of~\eqref{eqn:preview_TV_bound}.

Now we need to verify the proof of the upper bound in the theorem. Returning to our calculations above, it suffices to prove that
\[\frac{1}{2}\E_{P_{\tX}}\left[\E_{X,X'\stackrel{\textnormal{iid}}{\sim}P_{X|\tX}}\left[ | g(X) - g(X')|\right]\right] \leq \epsilon \cdot L_{g,2\epsilon,A}  + \|g\|_\infty \cdot P_{(X,\tX)}\{\|X-\tX\|>\epsilon\}+ \|g\|_\infty\cdot  P_X(A^c) ,\]
for any $\eps>0$ and any $A\subseteq\calX$. By definition of $L_{g,2\eps,A}$, we can write
\[\frac{1}{2}| g(X) - g(X')| \leq \begin{cases}\frac{2\eps\cdot L_{g,2\eps,A}}{2} , & \textnormal{ if $X,X'\in A$ and $\|X-X'\|\leq 2\eps$,}\\ \|g\|_\infty/2, & \textnormal{ otherwise},\end{cases}\]
where for the second case we recall that $g$ is nonnegative and so $|g(x) - g(x')|\leq \|g\|_{\infty}$ for any $x,x'$.
Therefore,
\begin{multline*}\frac{1}{2}\E_{P_{\tX}}\left[\E_{X,X'\stackrel{\textnormal{iid}}{\sim}P_{X|\tX}}\left[ | g(X) - g(X')|\right]\right]\\
\leq \eps \cdot L_{g,2\eps,A} + \frac{\|g\|_\infty}{2}\cdot \E_{P_{\tX}}\left[\P_{X,X'\stackrel{\textnormal{iid}}{\sim}P_{X|\tX}}\{X\not\in A\textnormal{ or }X'\not\in A\textnormal{ or }\|X-\tX\|>\eps \textnormal{ or }\|X'-\tX\|>\eps\}\right]\\\leq \eps \cdot L_{g,2\eps,A} + \frac{\|g\|_\infty}{2}\cdot 
2\E_{P_{\tX}}\left[\P_{X\sim P_{X|\tX}}\{X\not\in A\textnormal{ or }\|X-\tX\|>\eps\}\right]\\
\leq  \eps \cdot L_{g,2\eps,A} + \|g\|_\infty\cdot \P_{P_{(X,\tX)}}\{\|X-\tX\|>\eps\} +\|g\|_\infty\cdot  P_X(A^c),
\end{multline*}
which completes the proof.

\subsection{Training-conditional coverage}\label{app:training_conditional}
In this subsection, we will establish training-conditional coverage guarantees for RLCP. Towards that, we start with defining miscoverage of RLCP conditioned on the training data set $D_n=\{(X_i,Y_i)\}_{i\in[n]}$,\footnote{Note that, since the score function $s=s(x,y)$ is treated as fixed, in a practical setting where $s$ was fitted on a \emph{pretraining} set, and so we might refer to the data points $\{(X_i,Y_i)\}_{i\in[n]}$ as the \emph{calibration set} (as in our empirical results in Section~\ref{sec:experiment}). This means that we are interested in coverage conditional on both the pretraining and calibration datasets.}
\[\alpha_{\textnormal{tr}}(D_n)=\P\left\{Y_{n+1}\not\in \hat{C}_{n}^{\textnormal{RLCP}}(X_{n+1},\tX_{n+1})  \ \middle| \  D_n\right\}.\]
We would then like to show that $\alpha_{\textnormal{tr}}(D_n)\lessapprox\alpha$ holds with high probability, with respect to the random draw of the training data $D_n$. This type of result is well known to hold for split conformal prediction \citep{vovk2012conditional}. On the other hand, for full conformal prediction, \citet{bian2023training} prove a hardness result establishing that it is impossible to guarantee training-conditional coverage without further assumptions. Since the localized CP methods are closely related to full conformal prediction (as discussed in Section~\ref{sec:calLCP_fullCP}), it is not immediately clear whether a training-conditional guarantee can hold for a localized method. 

In this section, we will establish a training-conditional coverage guarantee for RLCP. First we will show a finite-sample result, and then will discuss how to interpret it from an asymptotic point of view.

\begin{theorem}\label{thm:training_conditional}
    For any $\delta>0$, the RLCP method satisfies the training-conditional coverage guarantee
    \[\P\left\{\alpha_{\textnormal{tr}}(D_n)\leq \alpha + \frac{2}{(\delta n)^{1/3}}\E\left[\frac{ H^2_2(\tX_{n+1})}{H_1(\tX_{n+1})^2}\right] ^{1/3} \right\} \geq 1-\delta,\]
    where for $x\in\calX$, we define
    \[H_1(x) = \E_{P_X}[H(X,x)], \ H_2^2(x) = \E_{P_X}[H(X,x)^2].\]
\end{theorem}

Now we consider an asymptotic view---when can we expect the excess error term to be vanishing as $n\rightarrow\infty$? For example, in the case $\calX=\RR$, under suitable regularity conditions on $P_X$, if at sample size $n$ we take $H$ to be  the normalized box kernel with bandwidth $h_n$, then we would expect to have $H_1(\tX_{n+1})\asymp 1$ and $H_2^2(\tX_{n+1})\asymp h_n^{-1}$. Consequently, the upper bound on $\alpha_{\textnormal{tr}}(D_n)$ is of the form $\alpha + O_P\left((nh_n)^{-1/3}\right)$. As long as $nh_n\rightarrow\infty$, then, we expect to see asymptotic training-conditional coverage.

\begin{proof}[Proof of Theorem~\ref{thm:training_conditional}]
Fix any $\alpha^*>\alpha$.
Condition on the value of $\tX_{n+1}$, and  let $q^*(\tX_{n+1})$ be the $(1-\alpha^*)$-quantile of the distribution of $s(X,Y)$ under the distribution $(X,Y)\sim (P_X\circ H(\cdot,\tX_{n+1}))\times P_{Y|X}$. Then by our calculation of the conditional distribution of the test point $(X_{n+1},Y_{n+1})$ given $\tX_{n+1}$ (as in~\eqref{eqn:RLCP_is_covariate_shift}) we have 
\[\P\{s(X_{n+1},Y_{n+1})\leq q^*(\tX_{n+1})\mid \tX_{n+1}\}\geq 1-\alpha^*\textnormal{ almost surely}.\]
Moreover since $(X_{n+1},Y_{n+1},\tX_{n+1})$ is independent of the training data $D_n=\{(X_i,Y_i)\}_{i\in[n]}$ we can equivalently write this as 
\[\P\{s(X_{n+1},Y_{n+1})\leq q^*(\tX_{n+1})\mid \tX_{n+1},D_n\}\geq 1-\alpha^*\textnormal{ almost surely}.\]
We then calculate
\begin{align*}
    \alpha_{\textnormal{tr}}(D_n)
    &=\P\left\{Y_{n+1}\not\in \hat{C}_{n}^{\textnormal{RLCP}}(X_{n+1},\tX_{n+1})  \ \middle| \  D_n\right\}\\
    &=\P\left\{s(X_{n+1},Y_{n+1}) > \hat{q}_{1-\alpha}(X_{n+1},\tX_{n+1})  \ \middle| \  D_n\right\}\\    &=\E\left[\P\left\{s(X_{n+1},Y_{n+1}) > \hat{q}_{1-\alpha}(X_{n+1},\tX_{n+1})  \ \middle| \  \tX_{n+1},D_n\right\} \ \middle| \  D_n\right]\\
    &\leq \E\bigg[\P\left\{s(X_{n+1},Y_{n+1}) > q^*(\tX_{n+1})  \ \middle| \  \tX_{n+1},D_n\right\}\\&\hspace{1in} + \P\left\{q^*(\tX_{n+1}) > \hat{q}_{1-\alpha}(X_{n+1},\tX_{n+1})  \ \middle| \  \tX_{n+1},D_n\right\}\ \bigg| \ D_n\bigg]\\
    &\leq \alpha^* + \E\left[ \P\left\{q^*(\tX_{n+1}) > \hat{q}_{1-\alpha}(X_{n+1},\tX_{n+1}) \ \bigg| \ \tX_{n+1},D_n\right\}\ \bigg| \ D_n\right]\\
    &= \alpha^* + \P\left\{q^*(\tX_{n+1}) > \hat{q}_{1-\alpha}(X_{n+1},\tX_{n+1}) \ \bigg| \ D_n\right\}.
\end{align*}
Moreover, by Markov's inequality, for any $\delta>0$ we have
\[\P\left\{q^*(\tX_{n+1}) > \hat{q}_{1-\alpha}(X_{n+1},\tX_{n+1})  \ \middle| \  D_n\right\} \leq \delta^{-1} \P\left\{q^*(\tX_{n+1}) > \hat{q}_{1-\alpha}(X_{n+1},\tX_{n+1})\right\}\]
with probability at least $1-\delta$.
We have therefore shown that
\[\P\left\{\alpha_{\textnormal{tr}}(D_n)\leq \alpha^* + \delta^{-1} \P\left\{q^*(\tX_{n+1}) > \hat{q}_{1-\alpha}(X_{n+1},\tX_{n+1})\right\} \right\} \geq 1-\delta.\]

Now we need to bound the term $\P\left\{q^*(\tX_{n+1}) > \hat{q}_{1-\alpha}(X_{n+1},\tX_{n+1})\right\}$. By definition of the quantile $\hat{q}_{1-\alpha}(X_{n+1},\tX_{n+1})$, on the event $q^*(\tX_{n+1}) > \hat{q}_{1-\alpha}(X_{n+1},\tX_{n+1})$ we must have
\[\sum_{i=1}^n \tilde{w}_i \ind\{s(X_i,Y_i) < q^*(\tX_{n+1})\} \geq 1-\alpha,\]
for weights
\[\tilde{w}_i = \frac{H(X_i,\tX_{n+1})}{\sum_{j=1}^{n+1}H(X_j,\tX_{n+1})}.\]
Equivalently, we must have
\[\sum_{i=1}^n H(X_i,\tX_{n+1}) \ind\{s(X_i,Y_i) < q^*(\tX_{n+1})\} \geq (1-\alpha)\sum_{i=1}^{n+1} H(X_i,\tX_{n+1}),\]
which implies
\[\sum_{i=1}^n H(X_i,\tX_{n+1}) \left( \ind\{s(X_i,Y_i) < q^*(\tX_{n+1})\} - (1-\alpha)\right) \geq (1-\alpha)H(X_{n+1},\tX_{n+1})\geq 0.\]
Next, conditional on $\tX_{n+1}$, the terms in the sum are independent, and each term $i$ has variance
\begin{multline*}\textnormal{Var}\left( H(X_i,\tX_{n+1}) \left( \ind\{s(X_i,Y_i) < q^*(\tX_{n+1})\} - (1-\alpha)\right) \ \middle| \  \tX_{n+1}\right) 
\\\leq \E\left[H(X_i,\tX_{n+1})^2 \left( \ind\{s(X_i,Y_i) < q^*(\tX_{n+1})\} - (1-\alpha)\right)^2 \ \middle| \  \tX_{n+1}\right]\\\leq \E[H(X_i,\tX_{n+1})^2 \mid \tX_{n+1}] = H^2_2(\tX_{n+1}). \end{multline*}
We can also calculate, for any $x\in\calX$ and any $q\in\RR$,
\begin{align*}
\E[H(X_i,x)\ind\{s(X_i,Y_i)<q\}]
&=\E_{P_X\times P_{Y|X}}[H(X,x)\ind\{s(X,Y)<q\}]\\
&=\P_{(P_X\circ H(\cdot,x))\times P_{Y|X}}\{s(X,Y)<q\}\cdot \E_{P_X}[H(X,x)],
\end{align*}
and so, plugging in $x=\tX_{n+1}$ and $q=q^*(\tX_{n+1})$, we have
\begin{multline*}
    \E[H(X_i,\tX_{n+1})\ind\{s(X_i,Y_i)<q^*(\tX_{n+1}) \mid \tX_{n+1}]\\
    = \P_{P_X\circ H(\cdot,\tX_{n+1}))\times P_{Y|X}}\{s(X,Y)<q^*(\tX_{n+1})\mid \tX_{n+1}\} \cdot \E_{P_X}[H(X,\tX_{n+1})\mid \tX_{n+1}]\\
    = \P\{s(X_{n+1},Y_{n+1})<q^*(\tX_{n+1})\mid \tX_{n+1}\} \cdot \E[H(X_i,\tX_{n+1})\mid \tX_{n+1}]\\
    \leq (1-\alpha^*) \cdot \E[H(X_i,\tX_{n+1})\mid \tX_{n+1}].
\end{multline*}
Therefore, for each $i$,
\begin{multline*}\E\left[H(X_i,\tX_{n+1}) \left( \ind\{s(X_i,Y_i) < q^*(\tX_{n+1})\} - (1-\alpha)\right) \ \middle| \  \tX_{n+1}\right] \\\leq -(\alpha^*-\alpha) \E[H(X_i,\tX_{n+1})\mid \tX_{n+1}] = -(\alpha^*-\alpha) \cdot H_1(\tX_{n+1}).
\end{multline*}
Applying Chebyshev's inequality, then,
\[\P\left\{\sum_{i=1}^n H(X_i,\tX_{n+1}) \left( \ind\{s(X_i,Y_i) < q^*(\tX_{n+1})\} - (1-\alpha)\right) \geq 0 \ \middle| \   \tX_{n+1}\right\} \leq n^{-1}\cdot \frac{ H^2_2(\tX_{n+1})}{(\alpha^*-\alpha)^2 H_1(\tX_{n+1})^2}.\]

Combining everything, we have therefore shown that
\[\P\left\{q^*(\tX_{n+1}) > \hat{q}_{1-\alpha}(X_{n+1},\tX_{n+1})\right\}\leq \frac{1}{n(\alpha^*-\alpha)^2} \cdot \E\left[\frac{ H^2_2(\tX_{n+1})}{H_1(\tX_{n+1})^2}\right],\]
and so,
\[\P\left\{\alpha_{\textnormal{tr}}(D_n)\leq \alpha^* + \frac{1}{\delta n(\alpha^*-\alpha)^2} \cdot \E\left[\frac{ H^2_2(\tX_{n+1})}{H_1(\tX_{n+1})^2}\right]\right\} \geq 1-\delta.\]
Choosing
\[\alpha^* = \alpha + \frac{1}{(\delta n)^{1/3}}\E\left[\frac{ H^2_2(\tX_{n+1})}{H_1(\tX_{n+1})^2}\right] ^{1/3}, \]
we have established that
\[\P\left\{\alpha_{\textnormal{tr}}(D_n)\leq \alpha + \frac{2}{(\delta n)^{1/3}}\E\left[\frac{ H^2_2(\tX_{n+1})}{H_1(\tX_{n+1})^2}\right] ^{1/3} \right\} \geq 1-\delta,\]
as desired.

\end{proof}

\subsection{Asymptotic test-conditional coverage}\label{app:asymptotic_test_conditional}
Since the model-free conditional coverage guarantee is impossible to attain, it is a common practice in literature to instead prove asymptotic test-conditional coverage. In this subsection, we establish such guarantees for RLCP prediction intervals.

As in Section~\ref{sec:coverage_on_balls}, we consider an asymptotic regime where the distribution $P$ is fixed and the marginal $P_X$ has density $f_X$ with respect to Lebesgue measure on $\RR^d$ (meaning that, implicitly, dimension $d$ is also fixed), while taking $n\to\infty$ and using a normalized box kernel $H$ with bandwidth $h_n$ to implement RLCP. In this section, instead of considering coverage conditional on an event $X_{n+1}\in\mathbb{B}(x_0,r_n)$ (i.e., conditioning on increasingly localized balls, as $r_n\to 0$), we will instead see that under some additional regularity assumptions we can condition on $X_{n+1}=x_0$.
As before, let $s$ be a pretrained score function. We will write $P_{S|X}$ to denote the distribution of $s(X,Y)$ conditional on $X$, under the distribution $(X,Y)\sim P$. We also assume that, marginally over $(X,Y)\sim P$, the distribution of $s(X,Y)$ is nonatomic.

In this regime, the following result holds for \emph{pointwise} test-conditional coverage:
\begin{theorem}\label{thm:asymptotic_test_conditional}
    Under the notation and assumptions above,
    \begin{multline}\label{eqn:asymptotic_test_conditional}\left|\P\left\{Y_{n+1}\in\hat{C}_n^{\textnormal{RLCP}}(X_{n+1},\tX_{n+1})  \ \middle| \  X_{n+1}=x_0\right\} - (1-\alpha)\right|\\ \leq \frac{1}{n+1} + \frac{1}{n(h_n)^d \cdot V_d \cdot \inf_{x\in\mathbb{B}(x_0,2h_n)}f_X(x)} + 2\sup_{x\in\mathbb{B}(x_0,2h_n)}\textnormal{d}_{\textnormal{TV}}(P_{S|X}(\cdot|x_0), P_{S|X}(\cdot |x)).\end{multline}
    In particular, if we assume $h_n\to 0$, and $h_n\cdot n^{1/d}\to \infty$, and if we also assume that $f_X(x)$ is bounded away from zero in a neighbourhood of $x=x_0$, and that that $x\mapsto P_{S|X}(\cdot|x)$ is continuous at $x=x_0$ with respect to  the total variation distance, meaning that
    \[\sup_{\|x-x_0\|_2\leq\epsilon} \textnormal{d}_{\textnormal{TV}}(P_{S|X}(\cdot|x_0), P_{S|X}(\cdot |x)) \to 0\textnormal{ as }\epsilon\to 0,\]
    then it holds that
    \[ \P\left\{Y_{n+1}\in\hat{C}_n^{\textnormal{RLCP}}(X_{n+1},\tX_{n+1})  \ \middle| \  X_{n+1}=x_0\right\} \to 1-\alpha.\]
\end{theorem}
 (As before,  $V_d$ is the volume of the unit ball in $\RR^d$, while $d_{\textnormal{TV}}$ denotes the total variation distance.)

\begin{proof}
    Let $S_i = s(X_i,Y_i)$  for $i\in [n+1]$. Since we are using a box kernel, the quantile cutoff of RLCP prediction method only uses the training samples that lies within a $h_n$ radius ball of $\tX_{n+1}$. Define 
\[
I(\tX_{n+1})=\{i\in[n]:\|X_i-\tX_{n+1}\|_2\leq h_n\},
\]
and $N_n = \sum_{i=1}^n \ind\{\|X_i-\tX_{n+1}\|_2\leq h_n\} = |I(\tX_{n+1})|$ be its cardinality. Without loss of generality, let us enumerate the elements of $I(\tX_{n+1})$ as $\{i_1,\hdots,i_{N_n}\}$.
Hence, the quantile cutoff can be simplified to \[
\hat q_{1-\alpha}(X_{n+1},\tX_{n+1})=\textnormal{Quantile}_{1-\alpha} \left(S_{i_1},\hdots,S_{i_{N_n}}, S_{n+1}
\right).
\]
By a simple application of the tower law, we note that
\begin{multline*}
    \P\{Y_{n+1} \in \hat{C}_n^{\textnormal{RLCP}}(X_{n+1},\tX_{n+1})\mid X_{n+1}=x_0\}=\\ \E\left[\P\{S_{n+1}\leq\textnormal{Quantile}_{1-\alpha} \left(S_{i_1},\hdots,S_{i_{N_n}}, S_{n+1}\right)\mid I(\tX_{n+1}),\tX_{n+1}, X_{n+1}=x_0\} \ \middle| \  X_{n+1}=x_0\right].
\end{multline*}
Now, we will concentrate on the inner probability term. Conditional on $I(\tX_{n+1})$, $\tX_{n+1}$ and on $X_{n+1}=x_0$, the scores $S_{i_1},\dots,S_{i_{N_n}}$ are sampled i.i.d.\ from the distribution $\tilde P_{S|\tX}(\cdot \mid \tX_{n+1})$, where $\tilde P_{S\mid \tX}(\cdot\mid x)$ is  defined by
\[\begin{cases}X \sim P_X\circ H(\cdot, x),\\ S\mid X\sim P_{S|X}(\cdot \mid X),\end{cases}\]
while $S_{n+1}$ is drawn from $P_{S|X}(\cdot\mid x_0)$ (independently of the other scores). Since $\tilde P_{S\mid \tX}(\cdot\mid x)$ is absolutely continuous with respect to the marginal distribution of $s(X,Y)$, it must be nonatomic. By Lemma~\ref{lem:dtv_quantile} below, we have
\begin{multline*}\left|\P\{S_{n+1}\leq\textnormal{Quantile}_{1-\alpha} \left(S_{i_1},\hdots,S_{i_{N_n}}, S_{n+1}\right)\mid I(\tX_{n+1}),\tX_{n+1}, X_{n+1}=x_0\} - (1-\alpha)\right|\\\leq \frac{1}{N_n+1} + \textnormal{d}_{\textnormal{TV}}\left( \tilde P_{S|\tX}(\cdot\mid \tX_{n+1}) , P_{S|X}(\cdot\mid x_0)\right).\end{multline*}
Therefore,
\begin{multline*}\left|\P\{Y_{n+1} \in \hat{C}_n^{\textnormal{RLCP}}(X_{n+1},\tX_{n+1})\mid X_{n+1}=x_0\} - (1-\alpha)\right|\\\leq \E\left[ \frac{1}{N_n+1}  \ \middle| \  X_{n+1}=x_0\right] + \E\left[\textnormal{d}_{\textnormal{TV}}\left( \tilde P_{S|\tX}(\cdot\mid \tX_{n+1}) , P_{S|X}(\cdot\mid x_0)\right)  \ \middle| \  X_{n+1}=x_0\right].\end{multline*}
To bound the second term, we note that conditional on $X_{n+1}=x_0$, we have $\tX_{n+1}\in\mathbb{B}(x_0,h_n)$ almost surely. Therefore,
\[\textnormal{d}_{\textnormal{TV}}\left( \tilde P_{S|\tX}(\cdot\mid \tX_{n+1}) , P_{S|X}(\cdot\mid x_0)\right) \leq \sup_{x\in\mathbb{B}(x_0,h_n)}\textnormal{d}_{\textnormal{TV}}\left( \tilde P_{S|\tX}(\cdot\mid x) , P_{S|X}(\cdot\mid x_0)\right).\]
Moreover, $\tilde P_{S|\tX}(\cdot\mid x)$ is a mixture of distributions $P_{S|X}(\cdot\mid X)$ where $X$ is drawn from $P_X\circ H(\cdot,x)$, i.e., $X\in\mathbb{B}(x,h_n)$ almost surely---and so, for $x\in\mathbb{B}(x_0,h_n)$, we have $X\in\mathbb{B}(x_0,2h_n)$ almost surely. Therefore,
\[\textnormal{d}_{\textnormal{TV}}\left( \tilde P_{S|\tX}(\cdot\mid \tX_{n+1}) , P_{S|X}(\cdot\mid x_0)\right) 
\leq \sup_{x\in\mathbb{B}(x_0,2h_n)}\textnormal{d}_{\textnormal{TV}}\left( P_{S|X}(\cdot\mid x) , P_{S|X}(\cdot\mid x_0)\right).\]

Our last step is to bound $ \E\left[ \frac{1}{N_n+1}  \ \middle| \  X_{n+1}=x_0\right]$.
Since the data points $X_1,\dots,X_n$ are i.i.d.\ samples from $P_X$, and are independent from $\tX_{n+1}$, we see that
conditional on $\tX_{n+1}$,
\[N_n \sim\textnormal{Binomial}\left(n,P_X(\mathbb{B}(\tX_{n+1},h_n)\right).\]
And, by definition, we have
\[P_X(\mathbb{B}(\tX_{n+1},h_n)\geq V_d(h_n)^d \cdot \inf_{x\in\mathbb{B}(x_0,2h_n)}f_X(x).\]
Therefore,
\[\E\left[ \frac{1}{N_n+1}  \ \middle| \  X_{n+1}=x_0\right]\geq \E\left[\frac{1}{\textnormal{Binomial}\left(n, (h_n)^d \cdot V_d\cdot \inf_{x\in\mathbb{B}(x_0,2h_n)}f_X(x)\right) + 1}.\right]\]
Finally, for any $p\in[0,1]$, $\E\left[\frac{1}{\textnormal{Binomial}(n,p)+1}\right] \leq \frac{1}{np}$, which completes the proof of the first claim of the theorem~\eqref{eqn:asymptotic_test_conditional}. The second claim follows by simply observing that the assumptions are sufficient to ensure that every term in the first bound~\eqref{eqn:asymptotic_test_conditional} is vanishing.
\end{proof}

\begin{lemma}\label{lem:dtv_quantile}
    Let $P_0,P_1$ be distributions on $\RR$, where $P_0$ is nonatomic. Let $N\geq 0$ be fixed, and suppose $A_1,\dots,A_N\sim P_0$ and $B\sim P_1$, where $A_1,\dots,A_N,B$ are independent. Then
    \[1 - \alpha - \textnormal{d}_{\textnormal{TV}}(P_0,P_1)\leq \P\{B\leq \textnormal{Quantile}_{1-\alpha} (A_1,\dots,A_N,B)\} \leq 1 - \alpha +\frac{1}{N+1} + \textnormal{d}_{\textnormal{TV}}(P_0,P_1). \]
\end{lemma}
\begin{proof}
    Let $A_{N+1}\sim P_0$ be drawn independently from $A_1,\dots,A_N$. Then, by definition of the total variation distance, we have
    \begin{multline*}\left|\P\{B\leq \textnormal{Quantile}_{1-\alpha} (A_1,\dots,A_N,B)\} - \P\{A_{N+1}\leq \textnormal{Quantile}_{1-\alpha} (A_1,\dots,A_N,A_{N+1})\}\right| 
    \\\leq d_{\textnormal{TV}}\big( (A_1,\dots,A_N,A_{N+1}), (A_1,\dots,A_N,B)\big) = \textnormal{d}_{\textnormal{TV}}(A_{N+1},B) = \textnormal{d}_{\textnormal{TV}}(P_0,P_1).\end{multline*}
    Finally, 
    \[1-\alpha\leq \P\{A_{N+1}\leq \textnormal{Quantile}_{1-\alpha} (A_1,\dots,A_N,A_{N+1})\}\leq 1-\alpha + \frac{1}{N+1},\]
    as in \cite{vovk2005algorithmic}, since $A_1,\dots,A_{N+1}$ are exchangeable (because they are i.i.d.) and are distinct almost surely (since $P_0$ is nonatomic). 
\end{proof}

\subsection{Proving the connection between calLCP and full conformal prediction }\label{app:proof_prop:calLCP_fullCP}

\begin{proof}[Proof of Proposition \ref{prop:calLCP_fullCP}]
For any dataset $(x_1,y_1),\dots,(x_{n+1},y_{n+1})$, the new score in Section \ref{sec:calLCP_fullCP} for the $i$th data point is defined as
\[\Big(T\big((x_1,y_1),\dots,(x_{n+1},y_{n+1})\big)\Big)_i = \sum_{j=1}^{n+1} w_{i,j} \ind\{s(x_j,y_j) < s(x_i,y_i)\},~~
\textnormal{where}~~
w_{i,j} = \frac{H(x_j,x_i)}{\sum_{k=1}^{n+1} H(x_k,x_i)}.\]
This implies the following equivalence between the two scores, $s$ and $T$:
\begin{equation}\label{eqn:equivalence_for_scores_calLCP_fullCP}
s(x_i,y_i)\leq \textnormal{Quantile}_{1-a}\left( \sum_{j=1}^{n+1} w_{i,j} \delta_{s(x_j,y_j)}\right) \iff \Big(T\big((x_1,y_1),\dots,(x_{n+1},y_{n+1})\big)\Big)_i < 1-a.
\end{equation}
Now, for a potential test response $y$, let's consider $Y_1^y,Y_2^y,\cdots,Y_{n+1}^y$ as in Section \ref{sec:defining calLCP}. Then, the corresponding transformed scores $T_1^y,T_2^y,\cdots, T_{n+1}^y$, are defined as
\[T_i^y=T\big((X_1,Y_1^y),\dots,(X_{n+1},Y_{n+1}^y)\big)\Big)_i =T\big((X_1,Y_1),\dots,(X_n,Y_n),(X_{n+1},y)\big)\Big)_i~~\textnormal{for all }i\in[n+1].\] By the equivalence~\eqref{eqn:equivalence_for_scores_calLCP_fullCP} established above, the recalibrated $\tilde{\alpha} (X_{n+1},y)$, i.e. 
\[
\tilde{\alpha}(X_{n+1},y) = \max\left\{a \in\Gamma(w): \sum_{i=1}^{n+1} \ind\left\{s(X_i,Y^y_i)\leq \textnormal{Quantile}_{1-a}\left( \sum_{j=1}^{n+1} w_{i,j} \delta_{s(X_i,Y^y_i)}  \right) \right\}\geq (1-\alpha)(n+1)\right\},\]
with
$w_{i,j} = \frac{H(X_j,X_i)}{\sum_{k=1}^{n+1} H(X_k,X_i)}, \ i,j\in[n+1]$, and with $\Gamma(w)$ defined as in~\eqref{eqn:def_Gamma_w_calLCP},
can now be expressed as
\[   \tilde{\alpha}(X_{n+1},y) = \max\left\{a\in\Gamma(w) : \sum_{i=1}^{n+1} \ind\left\{ T_i^y< 1-a\right\}\geq (1-\alpha)(n+1)\right\}.
\]
This means that $\tilde{\alpha}(X_{n+1},y)$ is the largest value in $\Gamma(w)$ satisfying $1-\tilde{\alpha}(X_{n+1},y) >\textnormal{Quantile}_{1-\alpha}(T_i^y,T_2^y,\cdots,T_{n+1}^y)$. Since $1-T^y_i$ takes values in the grid $\Gamma(w)$, then, $T_i^y<1-\tilde{\alpha}(X_{n+1},y)$ if and only if $T_i^y \leq \textnormal{Quantile}_{1-\alpha}(T_i^y,T_2^y,\cdots,T_{n+1}^y)$. Thus
once again using the same equivalence~\eqref{eqn:equivalence_for_scores_calLCP_fullCP}, one can write $\hat{C}_n^{\textnormal{calLCP}}(X_{n+1})$ as
\[\hat{C}_n^{\textnormal{calLCP}}(X_{n+1}) = \left\{y\in\calY : T_{n+1}^y< 1-\tilde{\alpha}(X_{n+1},y)\right\}=\left\{y\in\calY : T_{n+1}^y\leq \textnormal{Quantile}_{1-\alpha}(T_i^y,T_2^y,\cdots,T_{n+1}^y)\right\}.\]
The set on the right-hand side is exactly the output of the full conformal prediction interval, using the score map $T$. This establishes the result that $\hat{C}_n^{\textnormal{calLCP}}(X_{n+1})$ and $\hat{C}_n^{\textnormal{fullCP}}(X_{n+1})$ are exactly same.
\end{proof}

\subsection{A connection between $m$-RLCP and baseLCP}\label{app:mRLCP results}
To understand the behavior of $m$-RLCP from another perspective, we will look at an equivalent formulation. Define an average weight,
$$
W(X_i,X_{n+1})= \E_{\tX_{n+1}|X_{n+1}\sim H(X_{n+1},\cdot)} \left[\frac{H(X_i,\tX_{n+1})}{\sum_{j=1}^{n+1}H(X_j,\tX_{n+1})}\right],
$$
which is the expected value of the weight $\tilde{w}_i$ placed on training point $i$ when we run RLCP, averaged over the random draw of $\tX_{n+1}$. We then define
\[\tilde{p}(y;X_{n+1}) = 
\sum_{i=1}^nW(X_i,X_{n+1})\ind\{s(X_i,Y_i)\geq s(X_{n+1},y)\} + W(X_{n+1},X_{n+1}).\]
Comparing to the p-value $p(y;X_{n+1},\tX_{n+1}) $ that defines the RLCP method (see~\eqref{eqn:RLCP_as_pvalue}), we can see that this construction is identical except that we use the expected weights in place of the random $\tilde{w}_i$'s. We then correspondingly define a prediction interval
\[\tilde{C}_n(X_{n+1}) = \{y\in\calY : \tilde{p}(y;X_{n+1})>\alpha\}.\]

We now make a key observation: $\tilde{C}_n(X_{n+1})$ is actually equivalent to the baseLCP prediction interval $C_n^{\textnormal{baseLCP}}(X_{n+1})$, if in place of the localizer kernel $H$, we use a new kernel defined by $W$ (note, however, that we are slightly abusing the notation here, since $W(X_i,X_{n+1})$ depends not only on the two feature values $X_i,X_{n+1}$ but also on the remaining data points in the training set).

Now we again consider $m$-RLCP. Observe that the derandomized p-value $\bar{p}(y;X_{n+1},\tX_{n+1}^{[m]})$, defined in~\eqref{eqn:define_bar_p}, can be viewed as a stochastic approximation to $\tilde{p}(y;X_{n+1})$ (i.e., the expected weight, $W(X_i,X_{n+1})$, is approximated by averaging $m$ independent samples). Thus, $\hat{C}^{\textnormal{$m$-RLCP}}_n(X_{n+1})$ can be viewed as an approximate version of the interval $\tilde{C}_n(X_{n+1})$, which is the interval produced by baseLCP if run with a different kernel.

Based on this connection, we see that the empirically conservative behavior of $m$-RLCP is not surprising; our experiments, and the experiments of \citet{guan2023localized}, have demonstrated that baseLCP can often be quite conservative, and this property is then inherited by $m$-RLCP due to the connection between these two methods.

\section{Details for smoothing to avoid overcoverage}\label{app:smoothing}
In order to have a meaningful comparison between calLCP and RLCP in terms of local coverage, we should firstly have exact marginal coverage guarantee for both methods. For all methods in the conformal prediction framework (e.g., split CP), the standard theory ensures that marginal coverage will be \emph{at least} $1-\alpha$, but the method can overcover (i.e., coverage strictly higher than $1-\alpha$). In the exchangeable setting, any overcoverage is due only to two things: first, the issue of rounding (e.g., if $(1-\alpha)(n+1)$ is a noninteger, then calculating the $(1-\alpha)$-quantile for split conformal prediction as in~\eqref{eqn:split_CP} is necessarily slightly conservative), and second, the possibility of ties between scores.

For split and full conformal prediction (i.e., in the unweighted and unlocalized setting), \citet[Proposition 2.4]{vovk2005algorithmic} propose a randomization strategy, often also referred as ``smoothed conformal predictors'', to ensure exact coverage and avoid overcoverage. Taking inspiration from that randomization strategy, we will now develop the smoothed version of our method.

We will begin this section with reviewing the weighted conformal prediction (WCP) framework, and presenting its smoothed version. We will then extend these ideas to the setting of localized methods.

\subsection{Smoothing for weighted conformal prediction (WCP)}
\label{app:smoothed WCP}

The idea behind RLCP (and indeed, behind baseLCP and calLCP as well) is closely connected to the weighted conformal prediction (WCP) framework \citep{tibshirani2019conformal}. 
The weighted conformal prediction (WCP) method (\citet{tibshirani2019conformal}) builds a prediction interval of the form
\[
\hat{C}_n^{\textnormal{WCP}}(X_{n+1})=\left\{y\in\calY: s(X_{n+1},y)\leq \textnormal{Quantile}_{1-\alpha}\left(\sum_{i=1}^n w_i \delta_{s(X_i,Y_i)}+ w_{n+1} \delta_{+\infty} \right)\right\},
\]
for any pre-trained score function $s$.
In the covariate shift setting, where we have training data $\{(X_i,Y_i)\}_{i\in[n]}\stackrel{\textnormal{iid}}{\sim} P_X\times P_{Y|X}$ and a test point $(X_{n+1},Y_{n+1})\sim \tP_X\times P_{Y|X}$. 
 the weight for $i$th data point is given by
\[w_i=\frac{\frac{\mathsf{d}\tP_X}{\mathsf{d}P_X}(X_i)}{\sum_{j=1}^{n+1}\frac{\mathsf{d}\tP_X}{\mathsf{d}P_X}(X_j)}.\]
The marginal validity of this prediction interval ensures $\P\{Y_{n+1}\in \hat{C}_n^{\textnormal{WCP}}(X_{n+1})\}\geq1-\alpha$.

In general, any weighted conformal prediction interval can be extremely conservative, much more so than (unweighted) conformal prediction. This is because, when computing $\textnormal{Quantile}_{1-\alpha}(\cdot)$ for an unweighted empirical distribution over $n+1$ many values, rounding means that we take the $k$-th smallest value for $k=\lceil(1-\alpha)(n+1)\rceil$; in other words, we are actually computing the quantile at level $\frac{\lceil (1-\alpha)(n+1)\rceil}{n+1}$, but this value is bounded as $\leq 1-\alpha + \frac{1}{n+1}$. However, for WCP, the conservativeness caused by rounding when computing the quantile of a \emph{weighted} discrete distribution can potentially be as large as the largest weight---which might be much larger than $\frac{1}{n+1}$, if the weights are concentrated on just a few samples (i.e., low effective sample size). This means that providing a smoothed version to avoid overcoverage is particularly important in the weighted setting.

To do so, we begin by returning to the p-value based interpretation. An alternative way to express the construction of WCP  prediction interval would be as follows. Define a p-value for any potential response $y$ by comparing it with the weighted empirical distribution, where the weight for $i$th data point is defined as above. Then, $\hat{C}_n^{\textnormal{WCP}}(X_{n+1})$ can be also expressed as $\{y\in\calY: p^{\textnormal{WCP}}(y)> \alpha\}$ where \[p^{\textnormal{WCP}}(y)=\sum_{i=1}^n w_i \ind\{s(X_i,Y_i)\geq s(X_{n+1},y)\}+w_{n+1} .\]
In the covariate shift setting described above, this p-value can be shown to follow a super-uniform distribution, i.e., $\P\{p^{\textnormal{WCP}}(Y_{n+1})\leq \alpha\}\leq \alpha$; this claim is equivalent to the claim that $Y_{n+1}\in\hat{C}_n^{\textnormal{WCP}}(X_{n+1})$ with probability at least $1-\alpha$.

Now, to apply smoothing, we need to adjust this p-value so that its distribution is \emph{exactly} uniform, which will ensure that $Y_{n+1}\in\hat{C}_n^{\textnormal{WCP}}(X_{n+1})$ holds with probability \emph{exactly} $1-\alpha$. We modify the definition of the p-value to now be given by
\[p^{\textnormal{WCP}}(y)=\sum_{i=1}^n w_i \ind\{s(X_i,Y_i)> s(X_{n+1},y)\}+U\left(\sum_{i=1}^n w_i \ind\{s(X_i,Y_i)=s(X_{n+1},y)\}
+w_{n+1}\right),\]
where $U\sim \textnormal{Uniform}[0,1]$ is drawn independently of the data (note that typically, a single draw of $U$ is used  across all possible values $y\in\calY$). Under the covariate shift setting, the new p-value is now exactly uniformly distributed. This construction is analogous to the smoothed version of (unweighted) split conformal prediction, as proposed by \citet{vovk2005algorithmic}. Then the smoothed prediction interval, defined as $\hat{C}_n^{\textnormal{WCP}}(X_{n+1}) = \{y\in\calY: p^{\textnormal{WCP}}(y)> \alpha\}$, satisfies marginal coverage at level exactly $1-\alpha$:
\begin{proposition}\label{prop:WCP_smoothed}
Under the notation and definitions above, the smoothed WCP prediction interval has marginal coverage exactly equal to $1-\alpha$ under the covariate shift model.
\end{proposition}
\begin{proof}
Let $\hat{P}_{\textnormal{score}} = \sum_{i=1}^{n+1}w_i\delta_{s(X_i,Y_i)}$ be the weighted empirical distribution of the (training and test) score values, where $w_i\propto \frac{\mathsf{d}\tP_X}{\mathsf{d}P_X}(X_i)$ as defined above.
Following \citet{tibshirani2019conformal}'s formulation of the covariate shift setting through the language of weighted exchangeability (see, e.g., the proof of \citet[Lemma 1]{tibshirani2019conformal}), we can see that the conditional distribution of the test score is given by
\[s(X_{n+1},Y_{n+1})\mid \hat{P}_{\textnormal{score}} \sim \hat{P}_{\textnormal{score}}.\]
In words, if we condition on the unordered set (or multiset, in the case of repeated values) of scores $\{s(X_i,Y_i)\}_{i\in[n+1]}$, the chance that the test score $s(X_{n+1},Y_{n+1})$ takes any one of these values is proportional to the weight assigned to that value. Then
\begin{align*}
\P\{Y_{n+1}\in\hat{C}^{\textnormal{WCP}}_n(X_{n+1})\}
&=1 - \P\{p^{\textnormal{WCP}}(Y_{n+1})\leq \alpha\}\\
&=1 - \E\left[ \P\{p^{\textnormal{WCP}}(Y_{n+1})\leq \alpha\mid \hat{P}_{\textnormal{score}}\} \right],
\end{align*}
so it suffices to show that 
\[\P\{p^{\textnormal{WCP}}(Y_{n+1})\leq \alpha\mid \hat{P}_{\textnormal{score}}\} = \alpha\]
holds almost surely. 
By definition of $p^{\textnormal{WCP}}(y)$ observe that we can calculate
\[p^{\textnormal{WCP}}(Y_{n+1})
= \sum_{i=1}^{n+1}w_i \ind\{ s(X_i,Y_i)>s(X_{n+1},Y_{n+1})\} + U\left(\sum_{i=1}^{n+1}\ind\{ s(X_i,Y_i)=s(X_{n+1},Y_{n+1})\}\right),\]
or equivalently, writing
\[p(s) = \P_{S\sim \hat{P}_{\textnormal{score}}}\{S > s\} + U\cdot \P_{S\sim \hat{P}_{\textnormal{score}}}\{S = s\}, \] we have
\[p^{\textnormal{WCP}}(Y_{n+1})
=  p(s(X_{n+1},Y_{n+1})). \]
Now, since we know that, conditionally on $\hat{P}_{\textnormal{score}}$ we have $s(X_{n+1},Y_{n+1})\sim \hat{P}_{\textnormal{score}}$, this means that
\[\P\{p^{\textnormal{WCP}}(Y_{n+1})\leq \alpha\mid \hat{P}_{\textnormal{score}}\}
= \P\{ p(s(X_{n+1},Y_{n+1}))\leq \alpha \mid \hat{P}_{\textnormal{score}}\}
= \P_{S\sim \hat{P}_{\textnormal{score}}}\{p(S) \leq \alpha \mid \hat{P}_{\textnormal{score}}\}=\alpha,\]
where the last step holds by definition of $p(s)$. This completes the proof.
\end{proof}

\subsection{From smoothed WCP to smoothed RLCP}\label{app:smoothed_RLCP}
First, we recall the p-value based formulation of RLCP, as in Section~\ref{sec:derandomizing}: writing
\[p^{\textnormal{RLCP}}(y)=\sum_{i=1}^n\frac{H(X_i,\tX_{n+1})}{\sum_{j=1}^{n+1}H(X_j,\tX_{n+1})}\ind\{s(X_i,Y_i)\geq s(X_{n+1},y)\} + \frac{H(X_{n+1},\tX_{n+1})}{\sum_{j=1}^{n+1}H(X_j,\tX_{n+1})},\]
the prediction interval is given by $\hat{C}_n^{\textnormal{RLCP}}(X_{n+1})=\{y\in\calY: p^{\textnormal{RLCP}}(y) > \alpha\}$.

As we observed in the proof of Proposition~\ref{prop:key_property}, which established coverage for RLCP conditional on the ``synthetic prototype'' $\tX_{n+1}$, we can think of RLCP as constructing a weighted conformal prediction interval for the covariate shift problem that is (randomly) defined via the shifted feature distribution $\tP_X = P_X\circ H(\cdot,\tX_{n+1})$. In other words, in a sense we can view RLCP as an instance of WCP. This means that RLCP can be smoothed in the same way. Thus, we redefine the p-value as
\begin{multline*}p^{\textnormal{RLCP}}(y)=\sum_{i=1}^n\frac{H(X_i,\tX_{n+1})}{\sum_{j=1}^{n+1}H(X_j,\tX_{n+1})}\ind\{s(X_i,Y_i)> s(X_{n+1},y)\}\\{} + U\left(\sum_{i=1}^n\frac{H(X_i,\tX_{n+1})}{\sum_{j=1}^{n+1}H(X_j,\tX_{n+1})}\ind\{s(X_i,Y_i)= s(X_{n+1},y)\}+\frac{H(X_{n+1},\tX_{n+1})}{\sum_{j=1}^{n+1}H(X_j,\tX_{n+1})}\right),\end{multline*}
where $U\sim\textnormal{Uniform}[0,1]$ is drawn independently of the data,
and again define $\hat{C}_n^{\textnormal{RLCP}}(X_{n+1})=\{y\in\calY: p^{\textnormal{RLCP}}(y) > \alpha\}$.
This leads to the following marginal coverage guarantee:
\begin{theorem}[Marginal coverage for smoothed RLCP]\label{thm:marginal_smoothed}
The smoothed RLCP method defined above satisfies the exact marginal coverage property,i.e., 
for any distribution $P$, if  $(X_1,Y_1),\dots,(X_{n+1},Y_{n+1})\stackrel{\textnormal{iid}}{\sim} P$, then
\[\P\left\{Y_{n+1}\in\hat{C}^{\textnormal{RLCP}}_n(X_{n+1},\tX_{n+1})\right\}= 1-\alpha.\]
\end{theorem}
\begin{proof}
Following the same arguments as in the proof of Theorem~\ref{thm:marginal_coverage}, we can see that it suffices to prove a result analogous to Proposition~\ref{prop:key_property}: that is, we need to establish that for smoothed RLCP, it holds that
\[\E[\alpha(X_{n+1},\tX_{n+1})\mid \tX_{n+1}] = \alpha\textnormal{ almost surely}.\]
As in the proof of Proposition~\ref{prop:key_property}, this result follows from reinterpreting RLCP as an instance of WCP, in the covariate shift setting induced by conditioning on $\tX_{n+1}$; since Proposition~\ref{prop:WCP_smoothed} verifies exact marginal coverage under covariate shift for smoothed WCP, this completes the proof.
\end{proof}

\subsubsection{Consequences for coverage with respect to unknown covariate shift}
In Theorem~\ref{thm:covariate_shift}, we established a lower bound on coverage for the RLCP method, for  covariate shifts of the form $P_X\circ g$ where $g$ is unknown. In particular, in Section \ref{sec:insights_to_results}, we discussed that the result of the thoerem arises from the bound
\[\P\left\{Y_{n+1}\in\hat{C}^{\textnormal{RLCP}}_n(X_{n+1},\tX_{n+1})\right\}
\geq 1 - \alpha - \E_{\tP_{\tX}}\left[\textnormal{d}_{\textnormal{TV}}\left(P_X\circ H(\cdot,\tX), P_X\circ g\circ H(\cdot, \tX)\right)\right],\]
since this expected total variation term is bounded by the excess error term in the theorem.

For the smoothed setting, this expected total variation distance bounds the deviation from $1-\alpha$ coverage in either direction---that is, it bounds overcoverage as well as undercoverage. In particular, we have the following result:
\begin{theorem}[Robustness to covariate shift for smoothed RLCP]
    The smoothed RLCP method defined above satisfies the following property: for any distribution $P$, 
for every $g:\calX\rightarrow\RR_{\geq 0}$ with $P_X\circ g$ well-defined (i.e., with $0<\E_{P_X}[g(X)]<\infty$) it holds that
    \begin{multline*}
    \left| \P\left\{Y_{n+1} \in \hat{C}_n^{\textnormal{RLCP}}(X_{n+1},\tX_{n+1}) \right\} - (1-\alpha)\right| \leq \E_{\tP_{\tX}}\left[\textnormal{d}_{\textnormal{TV}}\left(P_X\circ H(\cdot,\tX), P_X\circ g\circ H(\cdot, \tX)\right)\right]\\ \leq \frac{\inf_{A\subseteq\calX, \epsilon>0}\left\{\epsilon \cdot L_{g,2\epsilon,A}  + \|g\|_\infty \cdot P_{(X,\tX)}\{\|X-\tX\|>\epsilon\}+ \|g\|_\infty \cdot P_X(A^c) \right\}}{\E_{P_X}[g(X)]},
    \end{multline*}
where the probability on the left-hand side is taken with respect to if  $(X_1,Y_1),\dots,(X_n,Y_n)\stackrel{\textnormal{iid}}{\sim} P$ and $(X_{n+1},Y_{n+1})\sim (P_X\circ g)\times P_{Y|X}$.
\end{theorem}
We omit the proof, since it follows from analogous arguments as the results above.

This result for covariate shift also leads to a result on relaxed test-conditional coverage (just like, for the unsmoothed version of RLCP, Theorem~\ref{thm:test_conditional} is a special case of Theorem~\ref{thm:covariate_shift}): for $\{(X_i,Y_i)\}_{i\in[n+1]}\stackrel{\textnormal{iid}}{\sim}P$, and for any subset $B\subseteq \calX$, we have
\begin{multline*}
    \left| \P\left\{Y_{n+1} \in \hat{C}_n^{\textnormal{RLCP}}(X_{n+1},\tX_{n+1})  \ \middle| \  X_{n+1}\in B\right\} - (1-\alpha)\right| \\ \leq \frac{\inf_{\epsilon>0}\left\{P_{(X,\tX)}\{\|X-\tX\|>\epsilon\}+P_X(\textnormal{bd}_{2\epsilon}(B))\right\}}{P_X(B)}.
    \end{multline*}

\subsection{Smoothed versions of baseLCP and calLCP}
For calLCP, \citet[Theorem 2.6]{guan2023localized} proposes a smoothed version of the calLCP method, and verifies that it achieves exactly $1-\alpha$ marginal coverage. An interesting observation is that, since we have seen in Proposition~\ref{prop:calLCP_fullCP} that calLCP is actually an instance of full CP implemented with a particular choice of the score function, \citet[Theorem 2.6]{guan2023localized}'s smoothing approach can also be derived by formulating calLCP as an instance of full CP, and then applying \citet[Proposition 2.4]{vovk2005algorithmic}'s smoothed version of full CP. 
Specifically, defining weights $w_{i,j} = \frac{H(X_j,X_i)}{\sum_{k=1}^{n+1} H(X_k,X_i)}$ as before, we compute a smoothed p-value as
\[p^{\textnormal{calLCP}}(y)=\sum_{i=1}^n \frac{1}{n+1}\ind\{T^y_i > T^y_{n+1}\}\\{} + U\left(\sum_{i=1}^n\frac{1}{n+1}\ind\{T^y_i = T^y_{n+1}\}+\frac{1}{n+1}\right),\]
where $U\sim\textnormal{Uniform}[0,1]$ is drawn independently of the data, and where we define $T^y\in\RR^{n+1}$ with entries
\[T^y_i = \begin{cases} \sum_{j=1}^n w_{i,j}\ind\{s(X_j,Y_j)<s(X_i,Y_i)\} + w_{i,n+1}\ind\{s(X_{n+1},y)<s(X_i,Y_i)\}, & i\in[n], \\  \sum_{j=1}^n w_{n+1,j}\ind\{s(X_j,Y_j)<s(X_{n+1},y)\},   & i=n+1.\end{cases}\]

For baseLCP, as we have discussed, the method does not in fact offer a guarantee of marginal coverage $\geq 1-\alpha$. Nonetheless, to treat the methods in a comparable way, we implement a smoothed version of baseLCP. Recalling our discussion in Section~\ref{sec:new look at baseLCP}, since baseLCP is also constructed based on the WCP framework (as is RLCP), we can define a smoothed version of baseLCP by leveraging the smoothed version of WCP. To do so, first we reformulate (unsmoothed) baseLCP via a p-value: for weights $w_{i,j}$ as above, defining
\[p^{\textnormal{baseLCP}}(y) = \sum_{i=1}^n\frac{H(X_i,X_{n+1})}{\sum_{j=1}^{n+1}H(X_j,X_{n+1})}\ind\{s(X_i,Y_i)\geq s(X_{n+1},y)\} + \frac{H(X_{n+1},X_{n+1})}{\sum_{j=1}^{n+1}H(X_j,X_{n+1})},\]
the baseLCP prediction interval can be equivalently written as $\hat{C}_n^{\textnormal{baseLCP}}(X_{n+1}) = \{y\in\calY : p^{\textnormal{baseLCP}}(y) >\alpha\}$. We then replace the p-value with a smoothed version,\begin{multline*}p^{\textnormal{baseLCP}}(y)=\sum_{i=1}^n\frac{H(X_i,X_{n+1})}{\sum_{j=1}^{n+1}H(X_j,X_{n+1})}\ind\{s(X_i,Y_i)> s(X_{n+1},y)\}\\{} + U\left(\sum_{i=1}^n\frac{H(X_i,X_{n+1})}{\sum_{j=1}^{n+1}H(X_j,X_{n+1})}\ind\{s(X_i,Y_i)= s(X_{n+1},y)\}+\frac{H(X_{n+1},X_{n+1})}{\sum_{j=1}^{n+1}H(X_j,X_{n+1})}\right),\end{multline*}
where again $U\sim\textnormal{Uniform}[0,1]$ is drawn independently of the data,
and then redefine $\hat{C}_n^{\textnormal{baseLCP}}(X_{n+1}) = \{y\in\calY : p^{\textnormal{baseLCP}}(y) >\alpha\}$ to obtain a smoothed interval.

\subsection{A note on implementation}

For all three methods, the p-value based formulation immediately suggests an efficient implementation to compute the prediction interval $\hat{C}_n(X_{n+1})$. Specifically, let $s_{(1)} < \dots < s_{(m)}$ be the \emph{unique} values of the scores $\{s(X_i,Y_i)\}_{i\in[n]}$, and then choose any values $s^*_{(0)},\dots,s^*_{(m)}$ satisfying
\[s^*_{(0)} < s_{(1)} < s^*_{(1)} < s_{(2)} < \dots < s_{(m-1)} < s^*_{(m-1)} < s_{(m)} < s^*_{(m)}.\]
Then it suffices to evaluate the (deterministic or smoothed) p-value at each of these $2m+1$ many scores. Specifically, abusing notation and writing the p-value as a function of the score $s=s(X_{n+1},y)$ rather than of the response value $y$,
\begin{itemize}
    \item For RLCP, for any $s\in\RR$, define
    \begin{multline*}p^{\textnormal{RLCP}}(s)=\sum_{i=1}^n\frac{H(X_i,\tX_{n+1})}{\sum_{j=1}^{n+1}H(X_j,\tX_{n+1})}\ind\{s(X_i,Y_i)> s\}\\{} + U\left(\sum_{i=1}^n\frac{H(X_i,\tX_{n+1})}{\sum_{j=1}^{n+1}H(X_j,\tX_{n+1})}\ind\{s(X_i,Y_i)= s\}+\frac{H(X_{n+1},\tX_{n+1})}{\sum_{j=1}^{n+1}H(X_j,\tX_{n+1})}\right),\end{multline*}
    where $U\sim\textnormal{Uniform}[0,1]$ (for the smoothed version) or $U\equiv 1$ (for the deterministic version).
        \item For calLCP, for any $s\in\RR$, define
\[p^{\textnormal{calLCP}}(s)=\sum_{i=1}^n\frac{1}{n+1}\ind\{T^s_i > T^s_{n+1}\}\\{} + U\left(\sum_{i=1}^n\frac{1}{n+1}\ind\{T^s_i = T^s_{n+1}\}+\frac{1}{n+1}\right),\]
    where $U\sim\textnormal{Uniform}[0,1]$ (for the smoothed version) or $U\equiv 1$ (for the deterministic version), and where
\[T^s_i = \begin{cases} \sum_{j=1}^n w_{i,j}\ind\{s(X_j,Y_j)<s(X_i,Y_i)\} + w_{i,n+1}\ind\{s<s(X_i,Y_i)\}, & i\in[n], \\  \sum_{j=1}^n w_{n+1,j}\ind\{s(X_j,Y_j)<s\},   & i=n+1.\end{cases}\]
        \item For baseLCP, for any $s\in\RR$, define
    \begin{multline*}p^{\textnormal{baseLCP}}(s)=\sum_{i=1}^n\frac{H(X_i,X_{n+1})}{\sum_{j=1}^{n+1}H(X_j,X_{n+1})}\ind\{s(X_i,Y_i)> s\}\\{} + U\left(\sum_{i=1}^n\frac{H(X_i,X_{n+1})}{\sum_{j=1}^{n+1}H(X_j,X_{n+1})}\ind\{s(X_i,Y_i)= s\}+\frac{H(X_{n+1},X_{n+1})}{\sum_{j=1}^{n+1}H(X_j,X_{n+1})}\right),\end{multline*}
    where $U\sim\textnormal{Uniform}[0,1]$ (for the smoothed version) or $U\equiv 1$ (for the deterministic version).
\end{itemize}
(For each method, in the smoothed case, note that a single random draw of $U$ is shared across all score values.)

Finally, for each method, writing $p(s)$ to denote $p^{\textnormal{RLCP}}(s)$ or  $p^{\textnormal{calLCP}}(s)$ or  $p^{\textnormal{baseCP}}(s)$  as appropriate, we can observe that $p(s)$ is a nonincreasing function of $s$ (i.e., a higher test score corresponds to a lower p-value). This means that the prediction interval can be defined as follows:
\[  \hat{C}_n(X_{n+1}) =    \begin{cases} \calY, & p(s^*_{(m)})>\alpha,\\
    \{y\in\calY: s(X_{n+1},y)\leq s_{(m)}\}, & p(s^*_{(m)})\leq \alpha < p(s_{(m)}),\\
    \{y\in\calY: s(X_{n+1},y)< s_{(m)}\}, & p(s_{(m)})\leq \alpha < p(s^*_{(m-1)}),\\
    \{y\in\calY: s(X_{n+1},y)\leq s_{(m-1)}\}, & p(s^*_{(m-1)})\leq \alpha < p(s_{(m-1)}),\\
    \dots\\
    \{y\in\calY: s(X_{n+1},y)\leq s_{(1)}\}, & p(s^*_{(1)})\leq \alpha < p(s_{(1)}),\\
        \{y\in\calY: s(X_{n+1},y)< s_{(1)}\}, & p(s_{(1)})\leq \alpha < p(s^*_{(0)}),\\
        \emptyset, & p(s^*_{(0)}) < \alpha.
    \end{cases}\]
(See \citet[Lemma 3.2]{guan2023localized} for an equivalent derivation of this implementation in the case of calLCP specifically.)

\section{Additional experiments and implementation details}
\subsection{Quantitative analysis of the effect of randomization on RLCP prediction intervals}\label{app:RLCP_deviation}
In Figure \ref{fig:RLCP_variability}, we showed a visual representation of the effect of randomization on the variability of the RLCP prediction interval. In the following experiments, we aim to quantify this variability in a  more precise way, to better understand the extent to which the RLCP prediction interval can vary across different random draws of $\tX_{n+1}$. In particular, we will consider the following measure of variability:
\begin{equation}\label{deviation definition}
    D(h)=\E_{D_n, X_{n+1}} \left[\frac{\mbox{MAD}\left(\lambda(\hat{C}_{n}^{\textnormal{RLCP}}(X_{n+1},\tX_{n+1})) \ \middle| \  D_n,X_{n+1}\right)}{\mbox{Median}\left(\lambda(\hat{C}_{n}^{\textnormal{RLCP}}(X_{n+1},\tX_{n+1})) \ \middle| \  D_n,X_{n+1}\right)}\right]
\end{equation}
where $\lambda$ denotes the Lebesgue measure, and where MAD, the median absolute deviation, is defined for a random variable $X$ as $\mbox{Median}(|X-\mbox{Median}(X)|)$. In other words, $D(h)$ measures the extent to which the width of the RLCP prediction interval can vary, relative to the typical width of this interval, given a \emph{fixed} data set (i.e., the only distribution being considered here is the random draw of $\tX_{n+1}$). We plot (an estimate of) this measure $D(h)$ for RLCP prediction intervals against a wide range of bandwidths $h$, for Setting 1 and Setting 2 of the univariate simulations that were carried out in Section~\ref{sec:simulation_univariate} (and analyzed in terms of the randomness of RLCP in Section~\ref{sec:derandomizing}). The results are shown in the left panel of Figure~\ref{fig:RLCP_deviation}.  
The results suggest that, at least for smaller bandwidths, a random instance of RLCP prediction interval can be approximately $10\%$ wider or narrower than the average interval width, so the deviation is small but by no means negligible. For larger bandwidths, since the localization effect is now reduced, the variability is consequently lower.
\begin{figure}[!h]
    \centering
    \includegraphics[width=\textwidth]{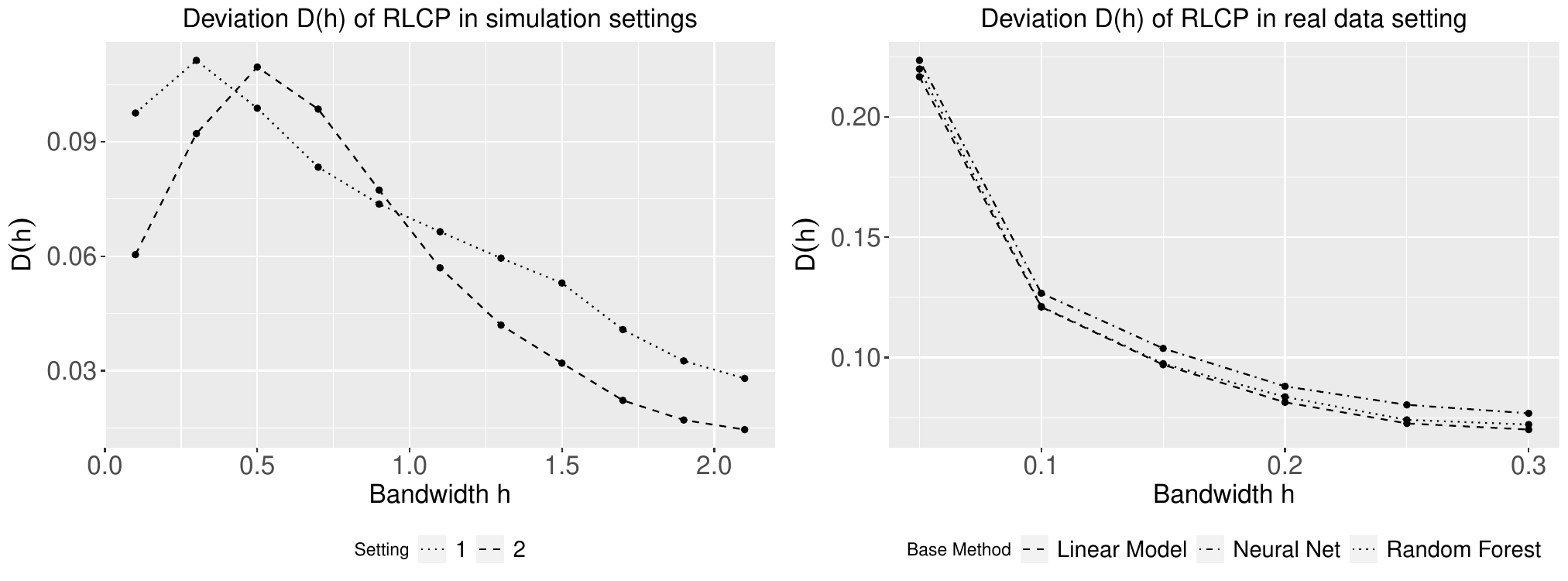}
    \caption{Quantification of the effect of randomization on RLCP via the estimated deviation $D(h)$, for the univariate simulations (Section~\ref{sec:simulation_univariate}) and the real data example (Section~\ref{sec:real_data}). See Section~\ref{app:RLCP_deviation} for details.}
    \label{fig:RLCP_deviation}
\end{figure}

We perform a similar analysis for our real data experiment (from Section~\ref{sec:real_data}), with results shown in the right panel of Figure~\ref{fig:RLCP_deviation}. These results show that, for our real data experiment, the RLCP prediction interval is more variable; a given draw of the RLCP prediction interval may be approximately $10$--$20\%$ wider or narrower in smaller bandwidth regimes and $5$--$10\%$ wider for relatively larger bandwidths, compared to the typical width.

\subsection{Experiments on $m$-RLCP}\label{app:mRLCP experiments}
\begin{figure}[!h]
    \centering
    \includegraphics[width=0.9\textwidth]{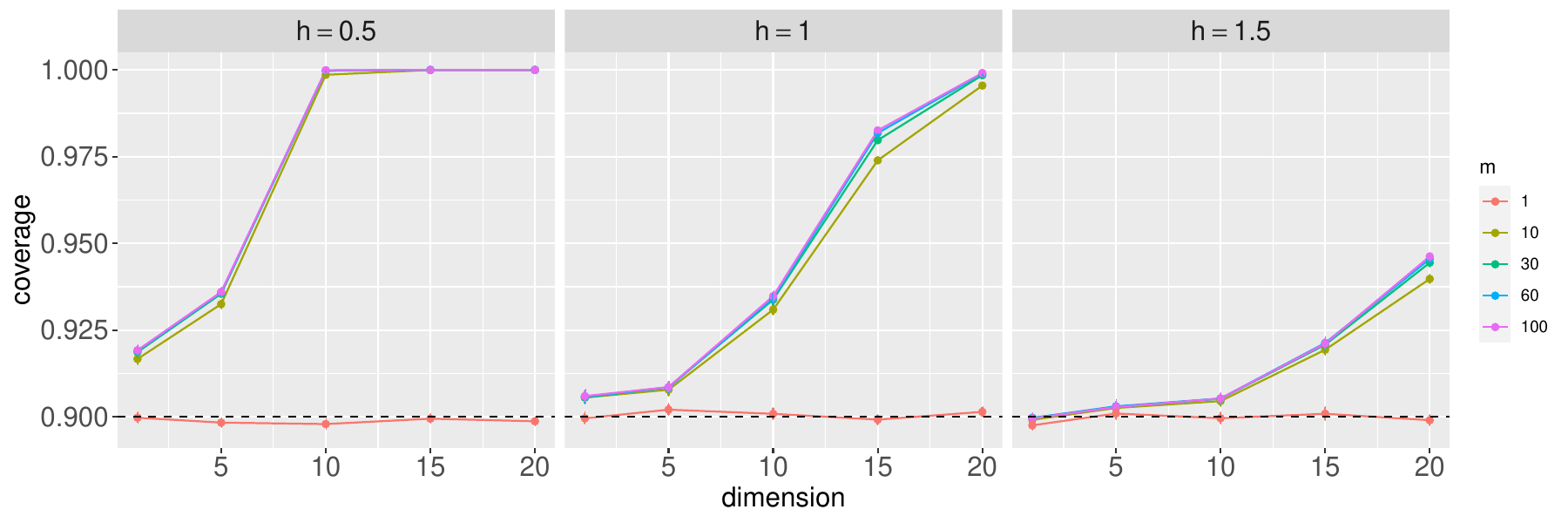}
    \caption{Marginal coverage of $m$-RLCP against increasing dimension for different choices of $m$.}
    \label{fig:mRLCP_coverage}
\end{figure}

The theoretical guarantee for $m$-RLCP only considers the issue of potential undercoverage (i.e., it provides a lower bound on coverage). In practice, $m$-RLCP can be extremely conservative. In Figure \ref{fig:mRLCP_coverage}, we plot the marginal coverage of $m$-RLCP against increasing dimension $d$, for the simulation setting previously studied in Section~\ref{sec:simulation_multivariate}.  The plots in Figure~\ref{fig:mRLCP_coverage} show marginal coverage over $30$ independent trials, with pretraining, calibration and test sample size each being $2000$.  Note that $m$-RLCP for $m=1$ is equivalent to the original RLCP method, while choosing large values of $m$ leads to a more derandomized method.
While RLCP (i.e., $m$-RLCP with $m=1$) shows coverage at the target level $1-\alpha$, all larger values of $m$ show substantial overcoverage and are thus extremely conservative, meaning that this form of aggregation does not offer a practical benefit.

\subsection{Additional experiments on local coverage in univariate settings}\label{app:additional_univariate_results}
\subsubsection{Local coverage experiments with additional bandwidths}
In Section~\ref{sec:simulation_univariate}, we showed results for our experiment comparing baseLCP, calLCP, and RLCP in the univariate setting, in terms of their local coverage properties across the range of values of the feature $X$. Here we show additional plots for this same experiment, with an expanded range of bandwidths, $h\in\{0.1,0.2,0.4,0.8,1.6\}$.
Settings 1 and 2 are shown in Figures \ref{fig:setting_1_local_coverage_full} and \ref{fig:setting_2_local_coverage_full}. All details for implementation are exactly as in Section~\ref{sec:simulation_univariate}.
\begin{figure}[!ht]
    \centering
    \includegraphics[width=0.75\textwidth]{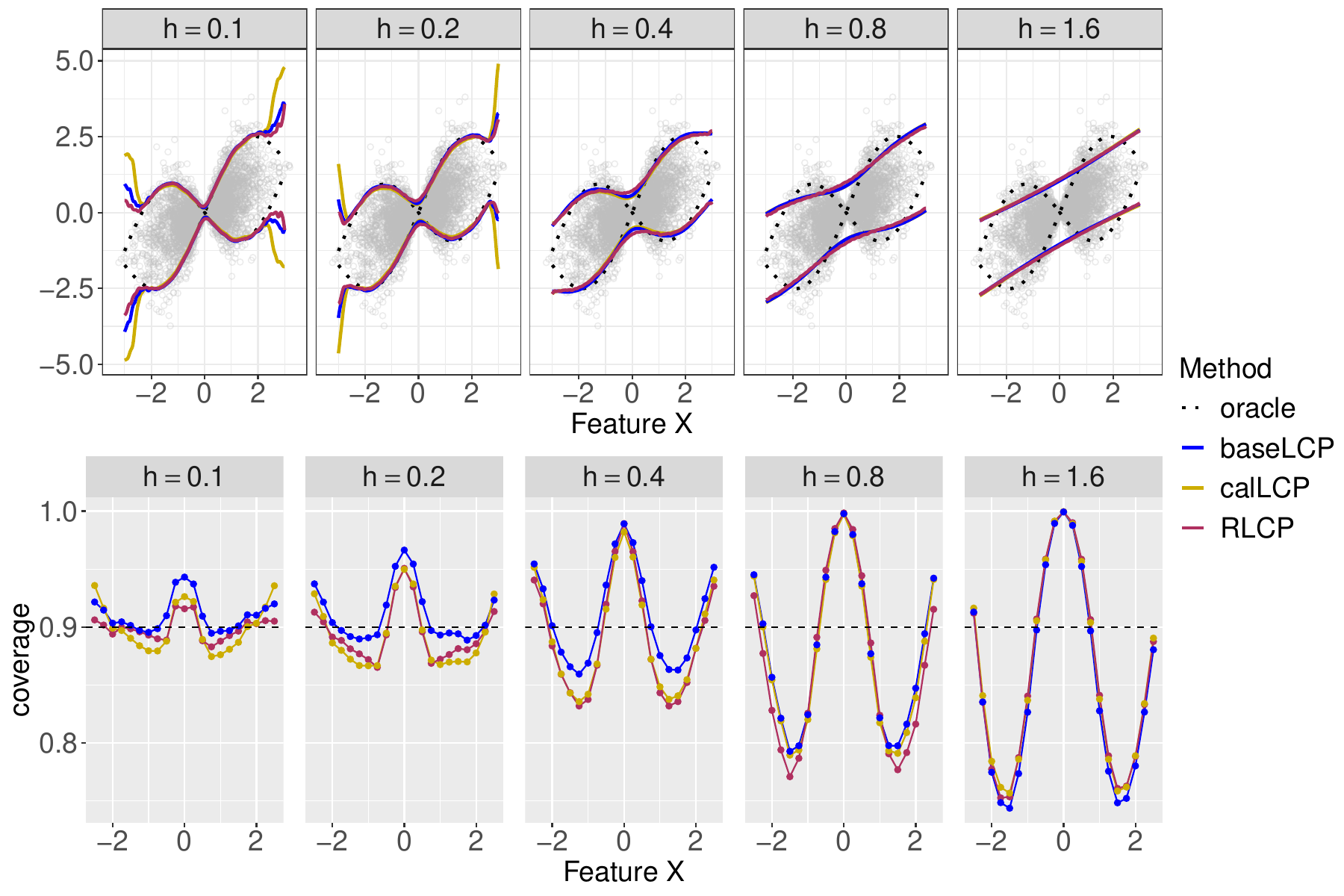}
    \caption{Additional results for Setting 1. Top panels: the average prediction intervals produced by each method (at each feature value $x$, the endpoints of the prediction interval are averaged over $50$ independent trials), compared to the oracle prediction interval. Bottom panels: The corresponding local coverage of each method, averaged over $50$ independent trials. See Section~\ref{sec:simulation_univariate} and Appendix~\ref{app:additional_univariate_results} for details.}
    \label{fig:setting_1_local_coverage_full}
\end{figure}

\begin{figure}[!ht]
    \centering
    \includegraphics[width=0.75\textwidth]{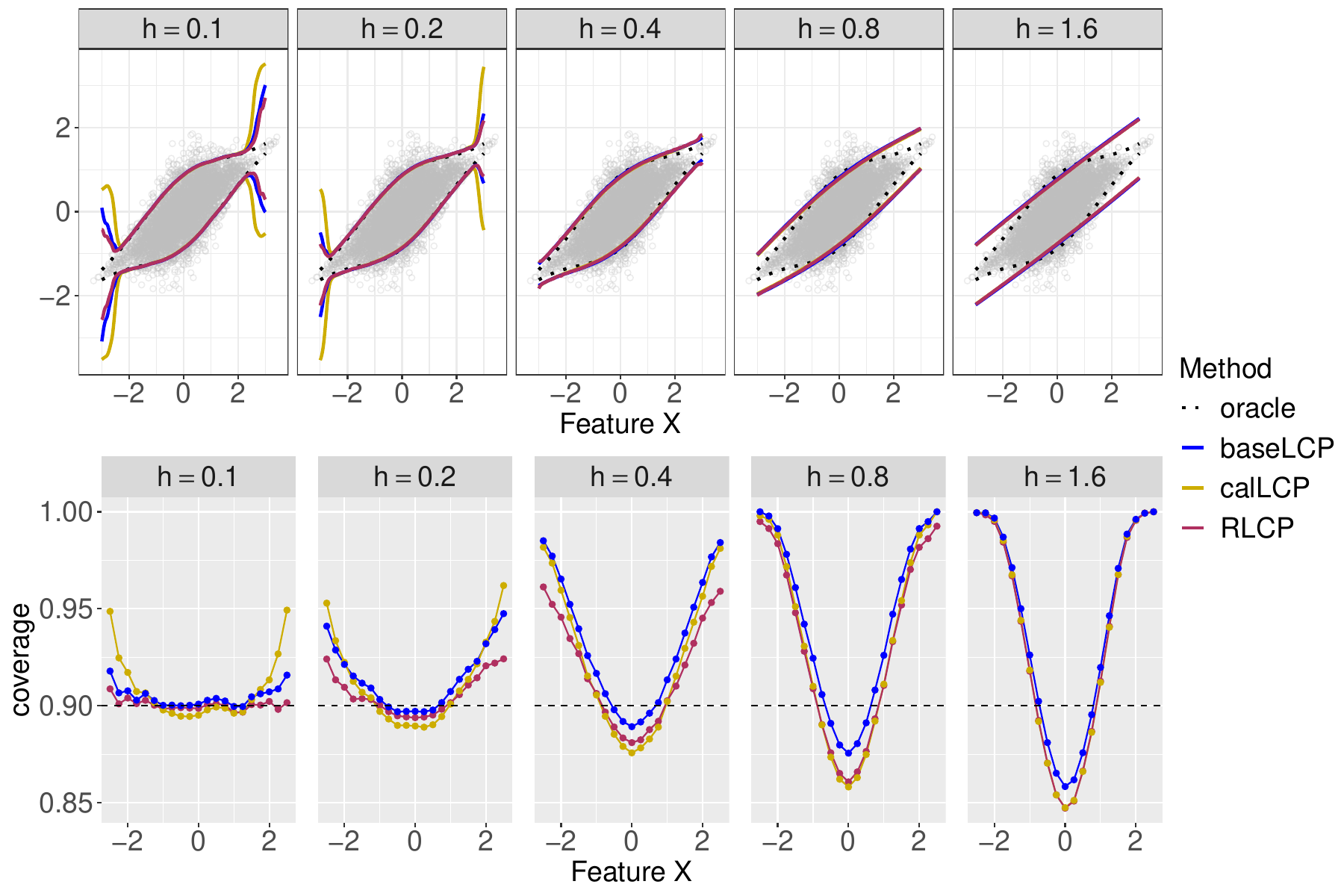}
    \caption{Additional results for Setting 2. Top panels: the average prediction intervals produced by each method (at each feature value $x$, the endpoints of the prediction interval are averaged over $50$ independent trials), compared to the oracle prediction interval. Bottom panels: The corresponding local coverage of each method, averaged over $50$ independent trials. See Section~\ref{sec:simulation_univariate} and Appendix~\ref{app:additional_univariate_results} for details.}
    \label{fig:setting_2_local_coverage_full}
\end{figure}
\subsubsection{Comparison with \citet{gibbs2023conformal}'s method}
In Section~\ref{sec:discussion_literature}, we have seen how localization based approaches, such as RLCP, and calibration based approaches, such as the method of \citet{gibbs2023conformal}, fit into a  more general framework of conformal algorithms. As mentioned before, the localization step of RLCP with a particular bandwidth, (Step (b) of the three-step workflow described in Section~\ref{sec:discussion_literature}), and the calibration step of \citet{gibbs2023conformal} with respect to a particular functional class $\mathcal F$ (Step (c) of the workflow) can be viewed as approximately interchangeable---both approaches serve to approximate test-conditional coverage via a local relaxation. In particular, the two approaches lead to a similar trade-off between localization and sample size. Choosing a richer class of functions $\mathcal{F}$ (for the method of \citet{gibbs2023conformal}) or a smaller bandwidth for the kernel $H$ (for RLCP) would have similar effects on the outcome, by placing a stronger requirement on the coverage properties of the algorithm and requiring more adaptivity to local properties of the data; on the other hand, for both types of methods, this would come at a cost of reduced effective sample size, and therefore, the potential to return wider or less informative prediction intervals.

In order to compare these methods, and better understand how they share the same trade-off,  we repeat the univariate experiment of Section~\ref{sec:simulation_univariate} (Setting 1).\footnote{The method of \citet{gibbs2023conformal} is run using the code provided in that paper, which is available at \url{https://github.com/jjcherian/conditional-conformal}.} In Figure \ref{fig:experiments_comparing_with_gibbs_conformal} we plot the local coverage (as defined in Section \ref{sec:simulation_univariate}) for each method, averaged over $50$ independent trials. Both \citet{gibbs2023conformal}'s method and RLCP are implemented in two versions, ``high adaptivity'' and ``low adaptivity''. For \citet{gibbs2023conformal}'s method,  the function class is chosen to be more fine or more coarse, 
\[\mathcal{F}_i = \left\{\ind\{x\in G\} : G \in\mathcal{G}_i\right\} ~\textnormal{for}~ i=1,2 \]
where, $\mathcal{G}_1$, $\mathcal{G}_2$ are partitions of $[-2.5,2.5]$ into $8$ bins of equal size or $4$ bins of equal size, respectively. For RLCP, we use a Gaussian kernel with bandwidth
\[h_1 = 0.025 \textnormal{ or }h_2 = 0.05.\]

The richer function class $\mathcal{F}_1$, and similarly the smaller bandwidth $h_1$---both provide more localization and ``high adaptivity'', and lead to better local coverage, as compared to $\mathcal{F}_2$ or $h_2$---the choices with ``low adaptivity''. For both methods, the ``high adaptivity'' version shows  local coverage lying very close to target coverage $0.9$ uniformly over the feature space, but we see larger fluctuations for the versions with ``low adaptivity''.

\begin{figure}[!h]
    \centering
    \includegraphics[width=0.75\textwidth]{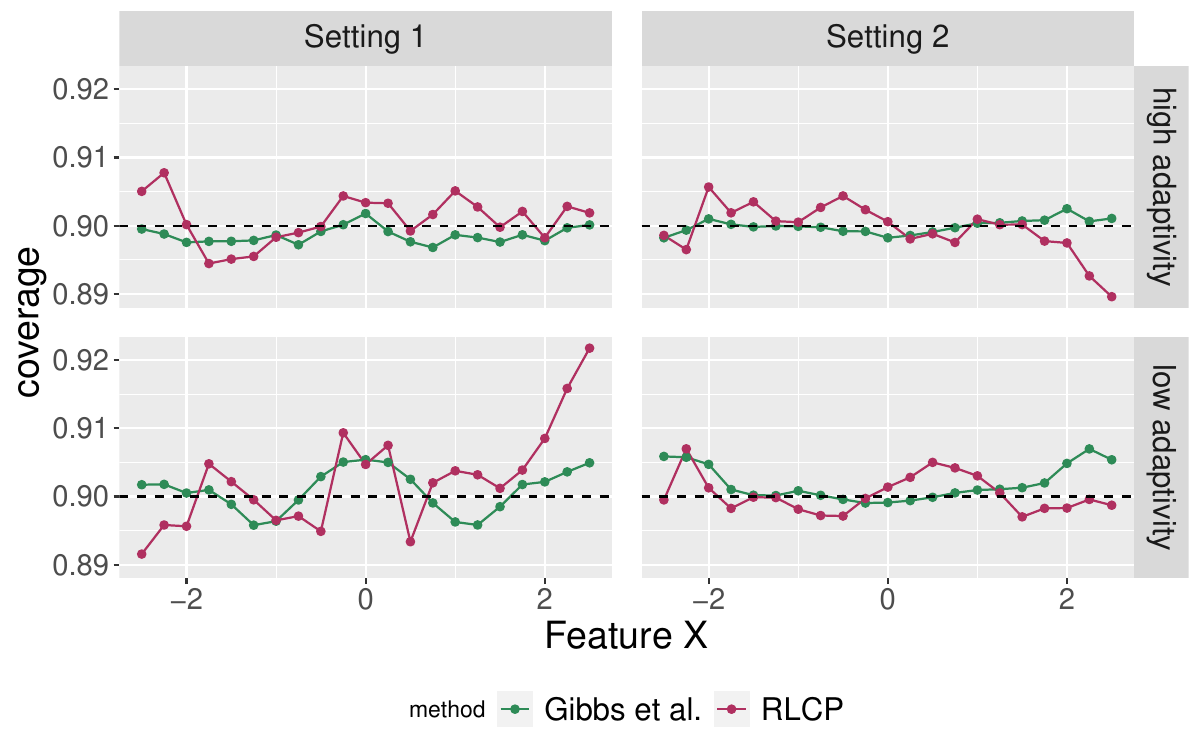}
    \caption{Local coverage of RLCP and the method from \citet{gibbs2023conformal} in setting $1$ for $d=1$}
    \label{fig:experiments_comparing_with_gibbs_conformal}
\end{figure}

\subsection{Additional experiments on local coverage in multivariate setting}\label{app:experiments_coverage_on_highD_bins}
In this section, we expand our experiments on local coverage in the multivariate setting. Recall in Section \ref{sec:simulation_multivariate} that we have studied coverage conditional on two sets, $B_{\textnormal{in}}$ and $B_{\textnormal{out}}$, which are defined to characterize the center and the tails of the multivariate normal data distribution. Here, to take a closer look at local coverage properties of RLCP as compared to calLCP, we will consider larger collections of sets by taking a finer partition of the space. Ideally, we will want to see empirical coverage on these sets  to cluster around the target level $0.9$---in particular, if each set $B$ is not too small, for RLCP our coverage guarantee in Theorem \ref{thm:test_conditional} ensures approximate coverage on each $B$. 

Specifically, we consider the data distribution given by
$$ X\sim \textnormal{Uniform}\bigl([-3,3]^d\bigr),~ Y|X\sim \calN \bigl(\sum_{i=1}^d\frac{ X_i}{2},\sum_{i=1}^d|\sin (X_i)|\bigr),$$
where $d\in \{5,10,\hdots,30\}$. We define intervals $I_i=[i,i+1)$, and we note that $\{I_i: i=-3,-2,-2,0,1,2\}$ define a partition of $[-3,3)$. Let us write the feature $x$ as $(x_1,x_2,\hdots,x_d)$, and consider the collection of sets $B$ of the form $B= I_i\times I_j\times I_k \times \RR^{d-3}$,
\[
\mathcal B = \left\{I_i\times I_j\times I_k \times \RR^{d-3} : i,j,k \in \{-3,-2,-2,0,1,2\}\right\}.
\]
This is a partition of the support of the distribution of $X$---we use the first $3$ coordinates of $x$ to determine which set $B\in\mathcal B$, it belongs to. 

We compute the empirical coverage of calLCP and RLCP conditional on each set $B\in\mathcal B$, averaged over $30$ repetitions. Each of the methods are run with $3000$ training and calibration samples, and $5000$ test samples with their own dimension-adaptive bandwidth choices, as defined in Section \ref{sec:simulation_multivariate}, which maintains a constant effective sample size of $50$ across all values of dimension $d$. Finally, in Figure \ref{fig:highd_local_coverage_full} we plot the kernel density estimates of the conditional coverages on these partitions across different dimensions. While both methods show conditional coverage that is distributed near the target value $0.9$, we do observe a difference between the two methods---in higher dimensions, the conditional coverage values are more widely spread around the target value $0.9$ for calLCP, as compared to RLCP. 
\begin{figure}[!h]
    \centering
    \includegraphics[width=0.75\textwidth]{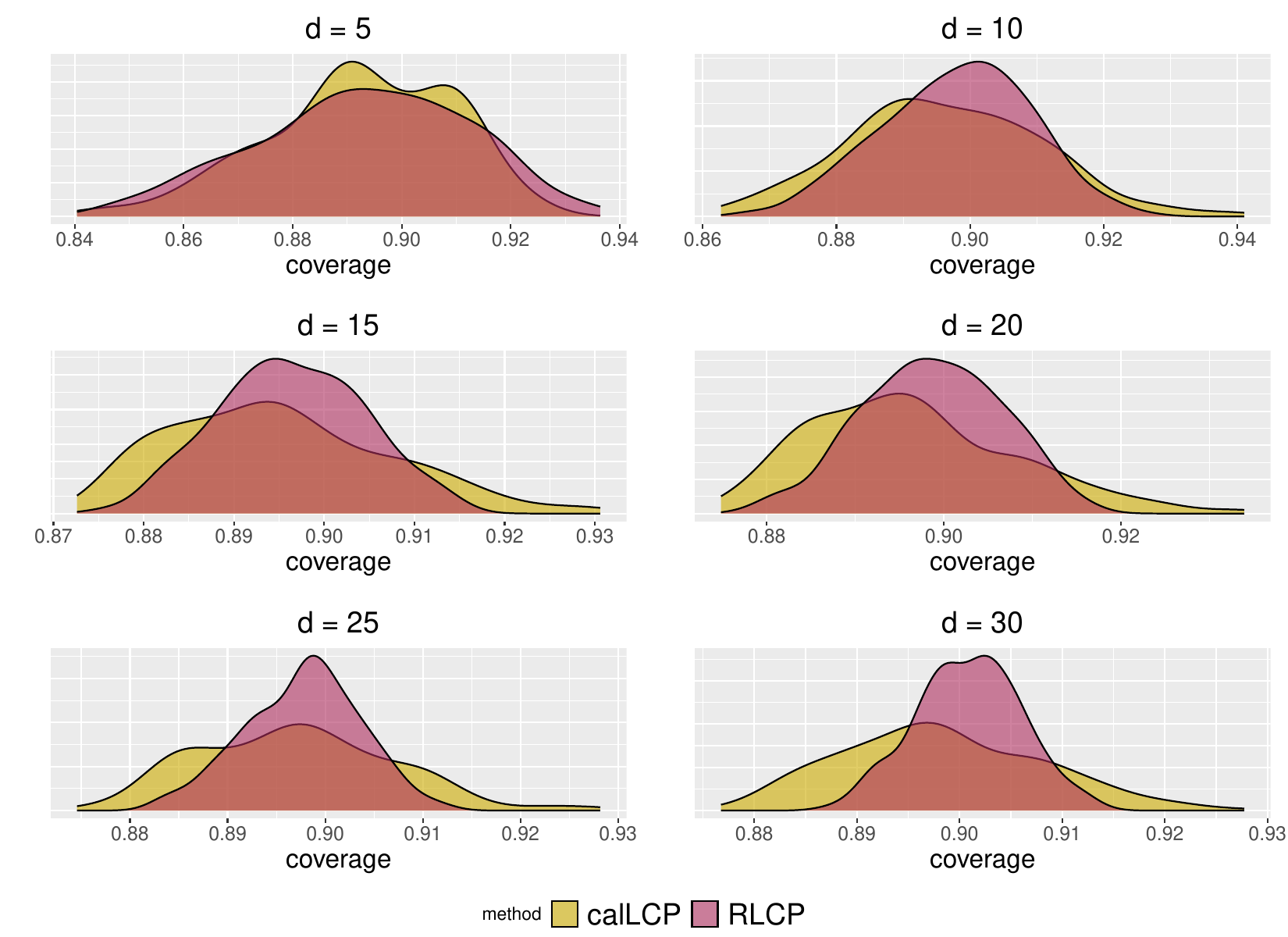}
    \caption{Local coverage of calLCP and RLCP prediction intervals across different dimensions.}
    \label{fig:highd_local_coverage_full}
\end{figure}


\end{document}